\documentclass[prodmode,acmjacm]{acmsmall}
\usepackage{latexsym,xspace,calc,amsthm}
\usepackage{amsmath,amssymb} 
\usepackage{mdwlist} 
\usepackage{marvosym}

\usepackage{bm}

\makeatletter
\def\@biblabel#1{[#1]} 
\def\thebibliography#1{%
    \footnotesize
    \refsection*{{\refname}
        \@mkboth{\uppercase{\refname}}{\uppercase{\refname}}%
    }
    \list{\@biblabel{\@arabic\c@enumiv}}
       {\settowidth\labelwidth{\@biblabel{#1}}%
        \leftmargin\labelwidth
        \advance\leftmargin\bibindent
        \itemindent-\bibindent
        \itemsep2pt
        \parsep \z@
        \usecounter{enumiv}
        \let\p@enumiv\@empty
        \renewcommand\theenumiv{\@arabic\c@enumiv}%
    }%
    \let\newblock\@empty
    \sloppy
    \sfcode`\.=1000\relax
}
\makeatother

\usepackage{microtype}
\usepackage{tgtermes} 

\usepackage{float}
\floatstyle{boxed} \restylefloat{figure}

\usepackage[normalem]{ulem}

\usepackage{tikz,xparse}
\usetikzlibrary{calc}
\DeclareDocumentCommand{\hcancel}{mO{0pt}O{1pt}O{0pt}O{-1pt}}{%
    \tikz[baseline=(tocancel.base)]{
        \node[inner sep=0pt,outer sep=0pt] (tocancel) {#1};
        \draw[gray] ($(tocancel.south west)+(#2,#3)$) -- ($(tocancel.north east)+(#4,#5)$);
    }%
}%

\let\myfresh\#
\def\#{\ensuremath{\text{\tt\myfresh}}}

\usepackage{bm}
\newcommand\cin{\in}
\newcommand\tneg{{\pmb\neg}}
\newcommand\ttop{{\pmb\top}}
\newcommand\tbot{{\pmb\bot}}
\newcommand\teq{{\pmb{=}}}
\newcommand\tand{{\pmb\wedge}}
\newcommand\tor{{\pmb\vee}}
\newcommand\timp{{\pmb\Rightarrow}}
\newcommand\tiff{{\pmb\Leftrightarrow}}
\newcommand\tall{{\pmb\forall}}
\newcommand\texi{{\pmb\exists}}
\newcommand\lneg{{\neg}}

\newcommand\ltop{\top}
\newcommand\lbot{\bot}
\renewcommand\land{\wedge}
\renewcommand\lor{\vee}
\newcommand\limp{\Rightarrow}

\newcommand\Tarski[1]{\f{Tarski}(#1)}

\newcommand\Func{{\Rightarrow}}
\newcommand{\fa}{\f{fa}}
\newcommand\den[1]{{\hspace{.00ex}\scalebox{.55}{$#1$}}}
\newcommand\iden{\den{\interp I}} 
\newcommand\lmathcal{\den{\mathcal L}} 
\newcommand\nden{\den{\interp N}}

\newcommand{\idenot}[1]{\denot{\interp I}{}{#1}}
\usepackage[mathscr]{eucal}
\newcommand\interp[1]{\ensuremath{\ns #1}}
\newcommand{\denot}[3]{\llbracket #3 \rrbracket_{\scalebox{.6}{$#2$}}^\den{#1}} 

\newcommand{\ndenot}[2]{\denot{\mathcal N}{#1}{#2}}

\newcommand\nontriv{\f{nontriv}}

 \renewenvironment{thebibliography}[1]{%
   \begin{odlthebibliography}{#1}%
     \setlength{\parskip}{0ex}%
     \setlength{\itemsep}{3pt}%
     \fontsize{10}{10} 
     \selectfont
}%
 {%
   \end{odlthebibliography}%
 }

\usepackage[colorlinks]{hyperref}
\hypersetup{urlcolor=black,colorlinks}

\usepackage{fancybox}
\newlength{\mylength}
\newenvironment{frameqn}%
{\setlength{\fboxsep}{5pt}
\setlength{\mylength}{\linewidth}%
\addtolength{\mylength}{-2\fboxsep}%
\addtolength{\mylength}{-2\fboxrule}%
\Sbox
\minipage{\mylength}%
\setlength{\abovedisplayskip}{0pt}%
\setlength{\belowdisplayskip}{0pt}%
$$}%
{$$\endminipage\endSbox
{\setlength{\abovedisplayskip}{1pt}%
\setlength{\belowdisplayskip}{0pt}%
\[\fbox{\TheSbox}\]}}

\newenvironment{frametxt}%
{\setlength{\fboxsep}{5pt}
\setlength{\mylength}{\linewidth}%
\addtolength{\mylength}{-2\fboxsep}%
\addtolength{\mylength}{-2\fboxrule}%
\Sbox
\minipage{\mylength}%
\setlength{\abovedisplayskip}{5pt}%
\setlength{\belowdisplayskip}{5pt}%
}%
{\endminipage\endSbox
{\setlength{\abovedisplayskip}{1pt}%
\setlength{\belowdisplayskip}{0pt}%
\[\fbox{\TheSbox}\]}}

\usepackage{graphics}
\usepackage{ifthen}
\usepackage{verbatim}
\newdimen\proofrulebreadth \proofrulebreadth=.05em
\newdimen\proofdotseparation \proofdotseparation=1.25ex
\newdimen\proofrulebaseline \proofrulebaseline=2ex
\newcount\proofdotnumber \proofdotnumber=3
\let\then\relax
\def\hfi{\hskip0pt plus.0001fil}
\mathchardef\squigto="3A3B
\newif\ifinsideprooftree\insideprooftreefalse
\newif\ifonleftofproofrule\onleftofproofrulefalse
\newif\ifproofdots\proofdotsfalse
\newif\ifdoubleproof\doubleprooffalse
\let\wereinproofbit\relax
\newdimen\shortenproofleft
\newdimen\shortenproofright
\newdimen\proofbelowshift
\newbox\proofabove
\newbox\proofbelow
\newbox\proofrulename
\def\shiftproofbelow{\let\next\relax\afterassignment\setshiftproofbelow\dimen0 }
\def\shiftproofbelowneg{\def\next{\multiply\dimen0 by-1 }%
\afterassignment\setshiftproofbelow\dimen0 }
\def\setshiftproofbelow{\next\proofbelowshift=\dimen0 }
\def\setproofrulebreadth{\proofrulebreadth}

\def\prooftree{
\ifnum  \lastpenalty=1
\then   \unpenalty
\else   \onleftofproofrulefalse
\fi
\ifonleftofproofrule
\else   \ifinsideprooftree
        \then   \hskip.5em plus1fil
        \fi
\fi
\bgroup
\setbox\proofbelow=\hbox{}\setbox\proofrulename=\hbox{}%
\let\justifies\proofover\let\leadsto\proofoverdots\let\Justifies\proofoverdbl
\let\using\proofusing\let\[\prooftree
\ifinsideprooftree\let\]\endprooftree\fi
\proofdotsfalse\doubleprooffalse
\let\thickness\setproofrulebreadth
\let\shiftright\shiftproofbelow \let\shift\shiftproofbelow
\let\shiftleft\shiftproofbelowneg
\let\ifwasinsideprooftree\ifinsideprooftree
\insideprooftreetrue
\setbox\proofabove=\hbox\bgroup$\displaystyle 
\let\wereinproofbit\prooftree
\shortenproofleft=0pt \shortenproofright=0pt \proofbelowshift=0pt
\onleftofproofruletrue\penalty1
}

\def\eproofbit{
\ifx    \wereinproofbit\prooftree
\then   \ifcase \lastpenalty
        \then   \shortenproofright=0pt  
        \or     \unpenalty\hfil         
        \or     \unpenalty\unskip       
        \else   \shortenproofright=0pt  
        \fi
\fi
\global\dimen0=\shortenproofleft
\global\dimen1=\shortenproofright
\global\dimen2=\proofrulebreadth
\global\dimen3=\proofbelowshift
\global\dimen4=\proofdotseparation
\global\count255=\proofdotnumber
$\egroup  
\shortenproofleft=\dimen0
\shortenproofright=\dimen1
\proofrulebreadth=\dimen2
\proofbelowshift=\dimen3
\proofdotseparation=\dimen4
\proofdotnumber=\count255
}

\def\proofover{
\eproofbit 
\setbox\proofbelow=\hbox\bgroup 
\let\wereinproofbit\proofover
$\displaystyle
}%
\def\proofoverdbl{
\eproofbit 
\doubleprooftrue
\setbox\proofbelow=\hbox\bgroup 
\let\wereinproofbit\proofoverdbl
$\displaystyle
}%
\def\proofoverdots{
\eproofbit 
\proofdotstrue
\setbox\proofbelow=\hbox\bgroup 
\let\wereinproofbit\proofoverdots
$\displaystyle
}%
\def\proofusing{
\eproofbit 
\setbox\proofrulename=\hbox\bgroup 
\let\wereinproofbit\proofusing
\kern0.3em$
}

\def\endprooftree{
\eproofbit 
  \dimen5 =0pt
\dimen0=\wd\proofabove \advance\dimen0-\shortenproofleft
\advance\dimen0-\shortenproofright
\dimen1=.5\dimen0 \advance\dimen1-.5\wd\proofbelow
\dimen4=\dimen1
\advance\dimen1\proofbelowshift \advance\dimen4-\proofbelowshift
\ifdim  \dimen1<0pt
\then   \advance\shortenproofleft\dimen1
        \advance\dimen0-\dimen1
        \dimen1=0pt
        \ifdim  \shortenproofleft<0pt
        \then   \setbox\proofabove=\hbox{%
                        \kern-\shortenproofleft\unhbox\proofabove}%
                \shortenproofleft=0pt
        \fi
\fi
\ifdim  \dimen4<0pt
\then   \advance\shortenproofright\dimen4
        \advance\dimen0-\dimen4
        \dimen4=0pt
\fi
\ifdim  \shortenproofright<\wd\proofrulename
\then   \shortenproofright=\wd\proofrulename
\fi
\dimen2=\shortenproofleft \advance\dimen2 by\dimen1
\dimen3=\shortenproofright\advance\dimen3 by\dimen4
\ifproofdots
\then
        \dimen6=\shortenproofleft \advance\dimen6 .5\dimen0
        \setbox1=\vbox to\proofdotseparation{\vss\hbox{$\cdot$}\vss}%
        \setbox0=\hbox{%
                \advance\dimen6-.5\wd1
                \kern\dimen6
                $\vcenter to\proofdotnumber\proofdotseparation
                        {\leaders\box1\vfill}$%
                \unhbox\proofrulename}%
\else   \dimen6=\fontdimen22\the\textfont2 
        \dimen7=\dimen6
        \advance\dimen6by.5\proofrulebreadth
        \advance\dimen7by-.5\proofrulebreadth
        \setbox0=\hbox{%
                \kern\shortenproofleft
                \ifdoubleproof
                \then   \hbox to\dimen0{%
                        $\mathsurround0pt\mathord=\mkern-6mu%
                        \cleaders\hbox{$\mkern-2mu=\mkern-2mu$}\hfill
                        \mkern-6mu\mathord=$}%
                \else   \vrule height\dimen6 depth-\dimen7 width\dimen0
                \fi
                \unhbox\proofrulename}%
        \ht0=\dimen6 \dp0=-\dimen7
\fi
\let\doll\relax
\ifwasinsideprooftree
\then   \let\VBOX\vbox
\else   \ifmmode\else$\let\doll=$\fi
        \let\VBOX\vcenter
\fi
\VBOX   {\baselineskip\proofrulebaseline \lineskip.2ex
        \expandafter\lineskiplimit\ifproofdots0ex\else-0.6ex\fi
        \hbox   spread\dimen5   {\hfi\unhbox\proofabove\hfi}%
        \hbox{\box0}%
        \hbox   {\kern\dimen2 \box\proofbelow}}\doll%
\global\dimen2=\dimen2
\global\dimen3=\dimen3
\egroup 
\ifonleftofproofrule
\then   \shortenproofleft=\dimen2
\fi
\shortenproofright=\dimen3
\onleftofproofrulefalse
\ifinsideprooftree
\then   \hskip.5em plus 1fil \penalty2
\fi
}

\usepackage{array}
\newcolumntype{L}[1]{>{$}p{#1}<{$}}
\newcolumntype{C}[1]{>{\centering$}p{#1}<{$}}
\newcolumntype{R}[1]{>{\raggedleft$}p{#1}<{$}}
\newcommand\maketab[2]
   {\newenvironment{#1}{\begin{quotation}\noindent\begin{tabular}{#2}}{\end{tabular}\end{quotation} }
    \newenvironment{#1noquote}{\noindent\begin{tabular}{#2}}{\end{tabular}}
    }

\usepackage[mathscr]{eucal}
\newcommand\ns[1]{\mathscr{#1}}

\newcommand{\freshcap}[1]{\mbox{$\bigcap^{\hspace{-.1ex}\raisebox{-.2ex}{\scalebox{.6}{$\# #1$}}}$}}
\newcommand{\freshwedge}[1]{\mbox{$\bigwedge^{\hspace{-.5ex}\raisebox{-.2ex}{\scalebox{.6}{$\# #1$}}}$}}

\newcommand{\freshvee}[1]{\mbox{$\bigvee^{\hspace{-.2ex}\raisebox{-.2ex}{\scalebox{.6}{$\# #1$}}}$}}

\newcommand\strict{\f{Strict}}
\newcommand\powerset{\f{PowSet}}
\newcommand\nompow{\f{NomPow}} 
\newcommand{\powsigma}{\f{Pow}_{\hspace{-1pt}\sigma}}
\newcommand{\powamgis}{\f{Pow}_{\hspace{-1pt}\scalebox{.74}{$\amgis$}}}

\newcommand{\ar}{\f{ar}}
\def\:{{\hspace{-1pt}{:}\hspace{-1.25pt}{:}\hspace{-.5pt}}}
\newcommand\theory[1]{\ensuremath{\mathsf{#1}}}

\def\id{\mathtxt{id}}

\newcommand\minus{{\text{-}}}

\newcommand\ssm{{{:}\text{=}}}

\newcommand\deffont[1]{{\bf #1}}

\newcommand\mone{{{\text{-}1}}}
\newcommand\liff{\Leftrightarrow}

\newcommand\supp{\f{supp}}

\newcommand\f[1]{\mathit{#1}}
\newcommand\Points{{\f{P}\hspace{-1.75pt}\f{oints}}}
\newcommand\tf[1]{\mathsf{#1}}
\newcommand\cent{\vdash}

\newcommand{\curvedarrowtop}{\hspace{-.17ex}\raisebox{1.85ex}{\scalebox{1.15}{\rotatebox{270}{$\curvearrowleft$}}}}

\newcommand\ii[1]{{\hspace{.5pt}\raisebox{.4pt}{$\curvedarrowtop$}^{\hspace{-1pt}#1}}} 

\def\atoms{\ensuremath{\mathbb{A}}\xspace}

\newcommand\dact[1]{}

\newcommand\lmodel{[\hspace{-0.2em}[}
\newcommand\rmodel{]\hspace{-0.2em}]}
\newcommand\model[1]{{\lmodel #1 \rmodel}}
\newcommand\vect[1]{\overline{#1}}

\newcommand\act[0]{{\cdot}}

\newcommand{\Defiff}
 {\mathrel{\ \ \stackrel{\scriptstyle \mathrm{def}}{\Leftrightarrow}\ \ }}

\newcommand{\defeq}
  {\stackrel{\mathrm{def}}{\,=\,}}

\newcommand\fix{\f{fix}}

\newcounter{jamieitemcounter}

\newtheoremstyle{jamiestyle}
  {4pt}
  {0pt}
  {\it}
  {0pt}
  {\sc}
  {.}
  { }
  {}
\theoremstyle{jamiestyle}
\newtheorem{thrm}{Theorem}[subsection]
\newtheorem{prop}[thrm]{Proposition}
\newtheorem{lemm}[thrm]{Lemma}
\newtheorem{corr}[thrm]{Corollary}
\newtheoremstyle{jamienfstyle}
  {4pt}
  {0pt}
  {\normalfont}
  {0pt}
  {\sc}
  {.}
  { }
  {}
\theoremstyle{jamienfstyle}
\newtheorem{nttn}[thrm]{Notation}
\newtheorem{defn}[thrm]{Definition}
\newtheorem{xmpl}[thrm]{Example}
\newtheorem{rmrk}[thrm]{Remark}

\newcommand\sm{{\mapsto}}
\usepackage{stmaryrd}\newcommand\ms{{\mapsfrom}}

\newcommand\mathtxt[1]{ \ensuremath{\mathrm{#1}} }

\newcommand\rulefont[1]{\scalebox{.9}{\ensuremath{\mathrm{\bf (#1)}}}}

\newcommand\Forall[1]{\forall #1.}
\newcommand\Exists[1]{\exists #1.}

\newcommand\amgis{\reflectbox{\ensuremath{\sigma}}}

\newcommand\new{\reflectbox{\ensuremath{\mathsf{N}}}}

\newcommand\New[1]{\new #1.}
\newcommand\lam[1]{\lambda #1.}

\newcommand\hnu{\scalebox{.8}{$\forall$}} 

\title{Semantics out of context: nominal absolute denotations for first-order logic and computation}
\author{\href{http://www.gabbay.org.uk}{Murdoch J. Gabbay}
\affil{Heriot-Watt University, Scotland, UK\quad 
\href{http://www.gabbay.org.uk}{\it http://www.gabbay.org.uk}}
}

\begin{abstract}
Call a semantics for a language with variables \emph{absolute} when variables map to fixed entities in the denotation.
That is, a semantics is absolute when the denotation of a variable $a$ is a copy of itself in the denotation. 

We give a trio of lattice-based, sets-based, and algebraic absolute semantics to first-order logic.
Possibly open predicates are directly interpreted as lattice elements / sets / algebra elements, subject to suitable interpretations of the connectives and quantifiers.
In particular, universal quantification $\Forall{a}\phi$ is interpreted using a new notion of \emph{`fresh-finite'} limit $\freshwedge{a}\model{\phi}$ and using a novel dual to substitution.

The interest of this semantics is partly in the non-trivial and beautiful technical details, which also offer certain advantages over existing semantics---but also the fact that such semantics exist at all suggests a new way of looking at variables and the foundations of logic and computation, which may be well-suited to the demands of modern computer science. 
\end{abstract}

\begin{document}
\maketitle
\tableofcontents

\section{Introduction}

We give three nominal absolute semantics to first-order logic with equality, based on lattices, sets, and algebra in a nominal universe.
Thus we provide alternatives to the \emph{de facto} standard semantics based on valuations, and in doing this we question deep-seated mathematical habits in syntax and semantics, and give evidence that the correct environment for doing logic in computer science is a mathematical foundation with \emph{names}, modelling \emph{variables}, building on the foundational work of Fraenkel and Mostowski.

We treat first-order logic because this is a paradigmatic formal language with variables. 
It is probably the simplest language with variables of any importance; 
and, it is of great importance, since it is a language for axiomatising set theory and arithmetic, amongst other things. 

The expert and impatient reader, wanting to see just how our models really differ from valuation-based models, might like to skip right to Remark~\ref{rmrk.argue}.
This is not the most technically advanced part of the paper, but it expresses a point at which we can see one concrete way in which the nominal models are clearly \emph{not} just a rephrasing of standard constructions.
Such a reader might also browse Subsection~\ref{subsect.map}, which maps out the underlying mathematics, and consider Example~\ref{xmpl.approx}.

We now take a step back and discuss the background issues in more depth.
 
\subsection{What are variables?}

Variables and quantifiers (or more generally, binders) are widespread. 
The integral $\int_a$ is a binder which binds $a$, so that $a$ in $\int_a f(a,b)$ is bound.
The same phenomenon appears in logic and computation, so that e.g. we write $\Forall{a}P(a,b)$,\ $\lam{a}ab$, $\f{let}\,a{=}2\,\f{in}\,a{+}b$, and so on.

Our notation for integration has an intended meaning: integration of the function.
In logic and computation there are many binders, and much design freedom in interpreting them.
The choices we make early on will influence the nature of the mathematics that follows out of them. 

We will compare and contrast this in detail in the Conclusions, but for here it is probably fair to say that the main methodology is to treat variables as a look-up to an external context; a quantifier interacts with this context by scanning possible values for the variable.
For instance:
\begin{itemize*}
\item
$\int_a f(a,b)$ means ``take a value for $b$ to be fixed by some context and vary possible values for $a$, taking an integral''.
\item
$\Forall{a}P(a,b)$ means ``take a value for $b$ to be fixed by some context and vary possible values for $a$, taking a(n infinite) logical conjunction''.
\item
$\lam{a}f(a)$ means ``input a value from the user, associate that value to $a$, and calculate $f$ in that context''.
\end{itemize*} 
In logic and computation this context of variable-to-value assignments is called a \emph{valuation}, and the idea is attributed to Tarski \cite{tarski:semct}.

But here are two other semantics of variables and quantifiers:
\begin{itemize*}
\item
We can treat quantifiers as \emph{modalities} (operators taking a formula and making a new formula), satisfying certain axiomatic properties.
Variables are used to label an infinite family of modalities: $\int_a$, $\int_b$, $\int_c$ or $\forall_a$, $\forall_b$, $\forall_c$, \dots

This interpretation is useful for proof---a formula may be too complicated to compute, but we can still prove it equivalent to another formula by axiomatic manipulations. 

In the context of logic, this is exemplified by cylindric and polyadic algebras due to Halmos and Tarski amongst others \cite{halmos:algl,henkin:cyla}.
\item
We can treat quantifiers as \emph{binding sites} (distinguished points in the formula), and variables are links/wires connecting different parts of the formula, via the binding site.

This interpretation is useful in programming.
For instance, a method name or function declaration binds the invocations of that method or function to the location where it is defined, and a pointer binds a location in memory to the locations in the program where it is dereferenced.
\end{itemize*}

In this paper we investigate a mathematical semantics which builds on and unifies the three interpretations above.
This works by reviving an old alternative to Zermelo--Fraenkel sets, known as \emph{Fraenkel-Mostowski sets} (\deffont{FM sets})---we will simply call this \emph{nominal techniques}. 
Variables have the properties of look-up, \emph{and} axiomatic-modality, \emph{and} binding-sites---depending on how one looks at the nominal semantics.

We can do this because we analyse variables using nominal techniques.
Detailed discussions of FM sets and their applications to computing are elsewhere \cite{gabbay:newaas-jv,gabbay:fountl,pitts:nomsns}.
What matters to us here is that in FM sets, atomic symbols are assumed to exist as  
\emph{urelemente} or \emph{atoms}.

In the language of programming we would say that atoms are a \emph{datatype} whose job is to contain infinitely many elements that are all symmetric up to permutations, and indeed this is exactly what happens in the Mur$\phi$ system \cite{dill:betvts} and the FreshML programming language \cite{shinwell:freocu}.
In some sense, atoms are a polar opposite to the familiar datatype of \emph{natural numbers}, which is a datatype whose job is to contain infinitely many elements that are all totally ordered in a single fixed and canonical manner.

Atoms are \emph{symmetric} up to permutations; symmetry is a primitive property of the mathematical universe.
To put this in the context of similar axioms, the Axiom of Infinity assumes the infinitude of the natural numbers, the Axiom of Comprehension assumes that predicates can be used to select subsets of sets, and the Axiom of Replacement assumes that functions can be applied pointwise to sets.
In FM sets we assume variable symbols, and their purpose is to be symmetric.

Because symmetry properties are foundationally assumed, they propagate naturally to constructions in FM sets.
How this works in full generality is a field of study in its own right.
What interests us here is what happens when we try to match up FM sets atoms with $\int_a$, $\forall a$, and $\lambda a$.
That is, can we model the behaviour of variables and binding as special cases of general FM behaviour?
This question was answered positively in previous work on abstract syntax, where FM atoms were used to model $\alpha$-equivalence in syntax \cite{gabbay:newaas-jv}.

But now we want to model the more complex \emph{semantic} behaviour of variables too.
We shall find that this works surprisingly well. 

At a high level, we shall see that variables are moved from being a specific property of a formal language, which requires explanation on a per-language/per-quantifier basis, to being a generic property of the mathematical universe on a par with generic concepts such as `set', `cardinality', and `function'.
These can be handled at a high level of abstraction and generalisation, and then instantiated to specific languages and applications. 

Doing this is a technical challenge, of course.
Just as important as the technical details is the \emph{ideas} that motivate them.
And, because variables and binders are so common in formal languages, and so fundamental to how they work, such new ideas about general semantics can pay worthwhile dividends.
The necessary mathematics is not particularly complicated, once we understand that we are dealing with a symmetric datatype with characteristic, albeit deep, properties.
The necessary nominal background is handled in Section~\ref{sect.basic.defs}, with plenty of examples and exposition.
We now discuss the specific technical application.

\subsection{The three paradigmatic semantics}

We study the concrete example of \emph{first-order logic} (\deffont{FOL}).
This is a paradigmatic language with binders, and is of practical importance since it is a base language for set theory and arithmetic.
The basic FOL connectives are $\tbot$ (false), $\tand$ (conjunction), $\tneg$ (negation), and $\tall a$ (first-order quantification).

FOL has three standard semantics: \emph{lattices}, \emph{sets}, and \emph{algebra} (a fourth is \emph{topology}, which is handled in a sister paper \cite{gabbay:stodfo}).
We will briefly summarise how these work:
\begin{itemize*}
\item
In lattices we take a partially ordered set and explain connectives using limits.
By this semantics we take some underlying set $X$ of \emph{truth-values} with an \emph{entailment} ordering $\leq$ and assume greatest lower bounds (and some other structure).
Given $x,y\in X$ we declare $x\land y$ to be the greatest lower bound (the limit) of $x$ and $y$.
\item
In sets we take some underlying set and explain connectives by combining subsets of the underlying set (sets intersection, complement, and so on).
That is, we take some set $X$ and subsets $U,V\subseteq X$, and we declare $U\land V$ to be $U\cap V$.
\item
In algebra we equip an underlying set with functions satisfying equalities, which should be abstractly specified, and explain connectives in terms of those functions (de Morgan laws, commutativity, associativity, and so on).
That is, we take some set $X$ and functions $\lbot$, $\land$, $\lneg$, and $\hnu$ on $X$ which must insist certain axioms---the correct axioms for $\hnu$ is a major contribution of this paper. 
\end{itemize*}
All these semantics are supposed to match up in some suitable sense.
The relevant theorems have standard names:
going from lattices to sets is a \emph{representation theorem}; going from algebras to sets is a (sets) \emph{semantics}; and going from sets back to lattices or algebras is a pair of \emph{completeness theorems}. 
We will treat these, establishing nominal lattice, nominal set, and nominal algebra treatments of FOL and indicating how to move between them.

As discussed, by the nominal approach we intepret variables directly as atoms.
We call the result \emph{absolute} because the meaning assigned to a term or predicate will not require any context or valuation; contrast this with the traditional Tarski-style valuation semantics, where the meaning of a term or predicate only exists in the context of a valuation assigning (non-nominal) denotations to the variables.
We will also give a detailed comparison with Tarski-style valuation semantics in Section~\ref{sect.complete}, showing how to build one of our models out of a Tarski-style model.
This translation will be natural, almost obvious, and it will exhibit a Tarski-style model as a particular special case of our framework (the key idea intuitively is that Tarski-style models are complete in a lattice-theoretic sense, whereas our models are in general only \emph{fresh-finite} complete, which is a weaker condition which is only expressible in a nominal semantics; see Remark~\ref{rmrk.argue}).

Returning to our three inter-translatable denotations, how specifically do we address the problem of variables and binding in each of them?
Detailed answers are in the body of the paper, but for the reader's convenience we give here---not a summary but---some precise pointers to where the key points in those answers will appear:
\begin{itemize*}
\item
In lattices, $\forall$ is a \emph{fresh-finite limit} (Subsection~\ref{subsect.fresh-finite.limit}).  
This is a new idea.
\item
In sets, $\forall$ is characterised twice: as a fresh-finite limit in the powerset considered as a lattive, and as an 
infinite sets intersection of substitution instances.
Theorem~\ref{thrm.powsigma.FOLeq} notes that these two characterisations are equivalent.   
\item
This raises an interesting question: since $\phi$ has an absolute semantics in sets (so $\phi$ is interpreted as a set, even if $\phi$ has free variables), then what notion of substitution is given to those sets?
Our answer uses \emph{$\amgis$-algebras} (Subsection~\ref{subsect.sigma.amgis}).  
These are also a new idea.
\item
In algebras, $\forall$ is an equivariant function satisfying certain axioms.
This idea is relatively recent, but has also been studied in previous work.
The axiomatisation is reminiscent of the cylindric algebra or polyadic algebra axiomatisations, though it exists in a nominal algebra framework and has its own distinct character.
We discuss this in Subsection~\ref{subsect.foleq.alg}.
\item
Lemma~\ref{lemm.technical} and Proposition~\ref{prop.all.sub.commute} are important technical results.
\end{itemize*}
More discussion is in the body of the paper and in the Conclusions.

\subsection{Map of the paper}
\label{subsect.map}

Section~\ref{sect.basic.defs} introduces the necessary nominal background.
This section need not necessarily be read first, because the results are abstract.
However, this investment in abstraction will pay dividends later when it is applied.
These applications are explained in detail in Section~\ref{sect.basic.defs}, see especially the discussion preceding Theorem~\ref{thrm.equivar}. 

In Section~\ref{sect.fol-algebra} we introduce the notion of algebra over nominal sets (i.e. a nominal set with functions on it satisfying nominal axioms).
Specifically, we consider $\sigma$- and $\amgis$-algebras.
A $\sigma$-algebra abstracts those properties of term- and predicate-syntax having to do with substitution; an $\amgis$-algebra is a dual to this.

Section~\ref{sect.nom.pow} introduces nominal posets (a nominal poset is a partially-ordered set, built in the universe of FM sets), but then does some new things with the idea: we study \emph{fresh-finite limits}, how $\sigma$-algebra structure interacts with the partial order, we note that simultaneous $\sigma$-action can also easily be modelled, we specify equality in the poset, and we conclude with the notion of \emph{FOLeq} algebra.
This establishes our lattice-theoretic semantics for first-order logic.

Section~\ref{sect.interp} spells this semantics out, by declaring first-order logic syntax and defining a formal notion of interpretation.
We prove soundness in Theorem~\ref{thrm.fol.sound}.

Section~\ref{sect.sigma.foleq} notes the important fact that every $\sigma$-powerset (Definition~\ref{defn.powsigma}) is a FOLeq algebra---just as every powerset is a Boolean algebra in Zermelo-Fraenkel set theory.
This is our sets semantics.

Section~\ref{sect.completeness} proves completeness of the sets semantics for FOLeq algebras (Corollary~\ref{corr.completeness}).
The construction visibly parallels the usual ultrafilter construction, but the details are significantly different, due to all the extra structure: 
the constructions are in nominal sets; points form an $\amgis$-algebra (not just a set); the set of all points forms a $\sigma$-powerset (not just a powerset); and we treat not only quantification $\tall$ but also equality $\teq$ which requires careful design of the structure of points.
In short, all the structure of a usual completeness proof is still there and still evident, along with extra structure relating the richer foundations. 

We are used to seeing models of first-order logic using valuations, in Zermelo-Fraenkel sets; we attribute this idea to Tarski.
Section~\ref{sect.complete} translates such models to our nominal framework, and we see that Tarski valuation models were a special case of FOLeq algebras all along.
The main result is Proposition~\ref{prop.standard.nom.bool}.
The converse is not true: not every FOLeq algebra can be expressed as a Tarski valuation model.
This is, intuitively, because Tarski models are very complete whereas FOLeq algebras satisfy a weaker property of being only fresh-finite complete in general.
What makes this interesting is that in our lattice semantics, being fresh-finite complete captures exactly what is necessary to model first-order logic; so FOLeq semantics are in this sense canonical, and the extra strength of Tarski-style models is due to the relative inexpressivity of the Zermelo-Fraenkel sets foundation which, implicitly, it assumes. 

Another standard semantics is syntax-quotiented-by-derivable-equivalence.
We call this a \emph{Herbrand} (or \emph{Lindenbaum-Tarski}) semantics and we consider it in Section~\ref{sect.herbrand}.
We do this briefly since, by now, the result should be clear; enough detail is given to reconstruct full proofs if desired.

The algebraisation of FOLeq algebras---their equational axiomatisation---is not without interest or subtlety, but it is also rather easy and has in fact already been treated in previous work \cite{gabbay:oneaah,gabbay:oneaah-jv}.\footnote{Chronologically the axioms came first, and this paper emerged from efforts to understand these axioms' semantics.}
So we treat it in a brief appendix, Appendix~\ref{sect.alg}.

\section{Background on nominal techniques}
\label{sect.basic.defs}

Intuitively, a nominal set is ``a set $\ns X$ whose elements $x\in\ns X$ may `contain' finitely many names $a,b,c\in\mathbb A$''.
We may call names \emph{atoms}.
The notion of `contain' used here is not the obvious notion of `is a set element of': formally, we say that $x$ has \emph{finite support} (Definition~\ref{defn.supp}).

For instance, here are some nominal sets:
\begin{itemize*}
\item
The set of finite sets of atoms: 
$$
\bigl\{\ \varnothing,\ \{a\},\ \{b\},\ \{c\},\dots,\ \{a,b\},\ \{a,c\},\dots\ \bigr\}.
$$
\item
The set of \emph{complements} of finite sets of atoms: 
$$
\bigl\{\ \mathbb A,\ \mathbb A{\setminus}\{a\},\ \mathbb A{\setminus}\{b\},\ \mathbb A{\setminus}\{c\},\dots,\ \mathbb A{\setminus}\{a,b\},\ \mathbb A{\setminus}\{a,c\},\dots\ 
\bigr\}.
$$
\end{itemize*}
Nominal sets are formally defined in Subsection~\ref{subsect.basic.definitions}, and examples are in Subsections~\ref{subsect.pow} and~\ref{subsect.more.examples}.

What is most important to realise is that the notion of `being in the support of $x$' is not equal to the notion of `being a set element of $x$'.
For instance if we take $x=\mathbb A{\setminus}\{a\}$, then $x$ contains infinitely many elements---but its support contains precisely the atom that is not an element of $x$, namely $a$.

Support measures \emph{name-symmetry}, not name-elementhood.
More on this below. 

The reader not interested in nominal techniques \emph{per se} might like to read this section only briefly in the first instance, and use it as a reference for the later sections, where these underlying ideas get applied. 
More detailed expositions are also in \cite{gabbay:newaas-jv,gabbay:fountl}.

In the context of the broader literature, the message of this section is as follows:
\begin{itemize*}
\item
The reader with a category-theory background can read this section as stating that we work in the category of nominal sets, or equivalently in the Schanuel topos (more on this in \cite[Section III.9]{MLM:sgl},\ \cite[A.21, page 79]{johnstone:skeett},\ or \cite[Theorem~9.14]{gabbay:fountl}).
\item
The reader with a sets background can read this section as stating that our constructions can be carried out in Fraenkel-Mostowski set theory (\deffont{FM sets}).

A discussion of this sets foundation, tailored to nominal techniques, can be found in \cite[Section~10]{gabbay:fountl}).
FM sets add \emph{urelemente} or \emph{atoms} to the sets universe.
\item
The reader not interested in foundations can note that previous work \cite{gabbay:newaas-jv,gabbay:fountl} has shown that just assuming names as primitive entities in Definition~\ref{defn.atoms} 
yields a remarkable clutch of definitions and results, notably Theorem~\ref{thrm.supp} and Corollary~\ref{corr.stuff}, and Theorems~\ref{thrm.equivar} and~\ref{thrm.new.equiv}.

Empirically these properties turn out to be incredibly useful, and they will be just what we need next in Section~\ref{sect.fol-algebra}.
\end{itemize*}

\subsection{Basic definitions}
\label{subsect.basic.definitions}

\begin{defn}
\label{defn.atoms}
Fix a countably infinite set of \deffont{atoms} $\mathbb A$.
We use a \deffont{permutative convention} that $a,b,c,\ldots$ range over \emph{distinct} atoms.
\end{defn}

\begin{defn}
A \deffont{(finite) permutation} $\pi$ is a bijection on atoms such that $\nontriv(\pi)=\{a\mid \pi(a)\neq a\}$ is finite.

Write $\id$ for the \deffont{identity} permutation such that $\id(a)=a$ for all $a$.
Write $\pi'\circ\pi$ for composition, so that $(\pi'\circ\pi)(a)=\pi'(\pi(a))$.
Write $\pi^\mone$ for inverse, so that $\pi^\mone\circ\pi=\id=\pi\circ\pi^\mone$.
Write $(a\;b)$ for the \deffont{swapping} (terminology from \cite{gabbay:newaas-jv}) mapping $a$ to $b$,\ $b$ to $a$,\ and all other $c$ to themselves, and take $(a\;a)=\id$.
\end{defn}

\begin{nttn}
\label{nttn.fix}
If $A\subseteq\mathbb A$ write 
$$
\fix(A)=\{\pi\mid \Forall{a{\in} A}\pi(a)=a\}.
$$
\end{nttn}

\begin{defn}
\label{defn.fin.supp}
\begin{enumerate}
\item
A \deffont{set with a permutation action} $\ns X$ is a pair $(|\ns X|,\act)$ of an \deffont{underlying set} $|\ns X|$ and a \deffont{permutation action} written $\pi\act x$ 
 which is a group action on $|\ns X|$, so that $\id\act x=x$ and $\pi\act(\pi'\act x)=(\pi\circ\pi')\act x$ for all $x{\in}|\ns X|$ and permutations $\pi$ and $\pi'$.
\item
Say that $A\subseteq\mathbb A$ \deffont{supports} $x{\in}|\ns X|$ when $\Forall{\pi}\pi\in\fix(A)\limp \pi\act x=x$.
If a finite $A$ supporting $x$ exists, call $x$ \deffont{finitely supported}.
\end{enumerate}
\end{defn}

\begin{frametxt}
\begin{defn}
\label{defn.nominal.set}
Call a set with a permutation action $\ns X$ a \deffont{nominal set} when every $x{\in}|\ns X|$ has finite support.
$\ns X$, $\ns Y$, $\ns Z$ will range over nominal sets.
\end{defn}
\end{frametxt}

\begin{defn}
\label{defn.equivariant}
Call a function $f\in |\ns X\Func\ns Y|$ \deffont{equivariant} when $\pi\act (f(x))=f(\pi\act x)$ for all permutations $\pi$ and $x{\in}|\ns X|$.
In this case write $f:\ns X\Func\ns Y$.

The category of nominal sets and equivariant functions between them is usually called the category of \emph{nominal sets} \cite{gabbay:thesis,gabbay:newaas-jv,gabbay:fountl,pitts:nomsns}.
\end{defn}

\begin{rmrk}
`Equivariant' appears in Definition~\ref{defn.equivariant} referring to functions between nominal sets, in Definition~\ref{defn.supp} referring to elements of nominal sets, and in Theorem~\ref{thrm.equivar} referring to predicates and functions in the language of ZFA set theory.

The notions of equivariance in Definition~\ref{defn.equivariant} and Theorem~\ref{thrm.equivar} are evidently related; the former existing inside the sets universe and the latter outside out.
The notion of equivariance in Definition~\ref{defn.supp} looks different, but we shall see in Lemma~\ref{lemm.equivar.equivar} how it closely relates to the other two. 
\end{rmrk}

\begin{defn}
\label{defn.supp}
Suppose $\ns X$ is a nominal set and $x{\in}|\ns X|$.
Define the \deffont{support} of $x$ by
$$
\supp(x)=\bigcap\{A\mid A\text{ finite and supports }x\} .
$$
If $\supp(x)=\varnothing$ we call $x$ \deffont{equivariant}.
\end{defn}

\begin{nttn}
\begin{itemize*}
\item
Write $a\#x$ as shorthand for $a\not\in\supp(x)$ and read this as $a$ is \deffont{fresh for} $x$.
\item
Given atoms $a_1,\dots,a_n$ and elements $x_1,\dots,x_m$ write $a_1,\dots,a_n\#x_1,\dots,x_m$ as shorthand for $\{a_1,\dots,a_n\}\cap\bigcup_{1{\leq}j{\leq}m}\supp(x_j)=\varnothing$.
That is: $a_i\#x_j$ for every $i$ and $j$.
\end{itemize*}
\end{nttn}

%
%

\begin{prop}
\label{prop.intersect.AB}
If $A\subseteq\atoms$ is finite and supports $x$, and $a\in\atoms$ and $a\#x$, then $A{\setminus}\{a\}$ supports $x$.
\end{prop}
\begin{proof}
Suppose $\pi\in\fix(A{\setminus}\{a\})$.
We assumed $a\#x$ so choose an $A'\subseteq\atoms$ such that $A'$ is finite, $A'$ supports $x$, and $a\not\in A'$.
Also, choose some fresh $a'$ (so $a'\not\in A{\cup}\{a\}{\cup}A'$).

Write $\tau=(a'\ a)$.
Note that $\tau\act x=x$ by Definition~\ref{defn.fin.supp}, because $A'$ supports $x$ and $\tau\in\fix(A')$.

It is a fact that $(\tau\circ\pi\circ\tau)(a)=a$ for every $a{\in}A$, so $\tau\circ\pi\circ\tau\in\fix(A)$.
Also by the group action $(\tau\circ\pi\circ\tau)\act x = \tau\act(\pi\act(\tau\act x))$.
Since $A$ supports $x$ we have $\tau\act(\pi\act(\tau\act x))=x$ by Definition~\ref{defn.fin.supp}.

We apply $\tau$ to both sides, recall that $\tau\act x=x$, and conclude that $\pi\act x=x$ as required.
\end{proof}

\begin{thrm}
\label{thrm.supp}
Suppose $\ns X$ is a nominal set and $x\in|\ns X|$.
Then $\supp(x)$ is the unique least finite set of atoms that supports $x$.
\end{thrm}
\begin{proof}
Consider a permutation $\pi\in\fix(\supp(x))$ and write $\{a_1,\dots,a_n\}=\f{nontriv}(\pi)$.
Choose any finite $A{\subseteq}\atoms$ that supports $x$, so by construction $\supp(x){\subseteq}A$.

By Proposition~\ref{prop.intersect.AB} $A{\setminus}\nontriv(\pi)$ supports $x$.
By construction $\pi\in\fix(A{\setminus}\nontriv(\pi))$, so by Definition~\ref{defn.fin.supp} $\pi\act x=x$ as required.
\end{proof}

\begin{corr}
\label{corr.stuff}
\begin{enumerate*}
\item\label{stuff.fixsupp.fixelt}
If $\pi(a)=a$ for all $a\in\supp(x)$ then $\pi\act x=x$.
\item
If $\pi(a)=\pi'(a)$ for every $a{\in}\supp(x)$ then $\pi\act x=\pi'\act x$.
\item
$a\#x$ if and only if $\Exists{b}(b\#x\land (b\;a)\act x=x)$.
\end{enumerate*}
\end{corr}
\begin{proof}
By routine calculations from the definitions and from Theorem~\ref{thrm.supp}.
\end{proof}

\subsection{Examples}
\label{subsect.pow}

Suppose $\ns X$ and $\ns Y$ are nominal sets, and suppose $\ns Z$ is a set with a permutation action.
We consider some examples of sets with a permutation action and of nominal sets.
These will be useful later on in the paper.

\subsubsection{Atoms}
\label{subsect.xmpl.atoms}

$\mathbb A$ is a nominal set with the \emph{natural permutation action} $\pi\act a=\pi(a)$.

For the case of $\mathbb A$ only we will be lax about the difference between $\mathbb A$ (the set of atoms) and $(|\mathbb A|,\act)$ (the nominal set of atoms with its natural permutation action).
What that means in practice is that we will write $a\in\mathbb A$ and never write $a{\in}|\mathbb A|$.\footnote{Just sometimes, pedantry has its limit.}

\subsubsection{Cartesian product}

$\ns X\times\ns Y$ is a nominal set with underlying set $\{(x,y)\mid x{\in}|\ns X|, y{\in}|\ns Y|\}$ and the \emph{pointwise} action $\pi\act(x,y)=(\pi\act x,\pi\act y)$.

\subsubsection{Tensor product}
\label{subsect.otimes}

$\ns X\otimes\ns Y$ is a nominal set with underlying set $\{(x,y)\mid x{\in}|\ns X|, y{\in}|\ns Y|,\ \supp(x)\cap\supp(y)=\varnothing\}$ and the pointwise action.
For the pointwise action here to be well-defined depends on $\pi$ being a permutation and the fact (Proposition~\ref{prop.pi.supp} below) that $\supp(\pi\act x)=\pi\act \supp(x)$.

\subsubsection{Full function space}
\label{subsect.full.function.space}

$\ns X{\to}\ns Y$ is a set with a permutation action with underlying set all functions from $|\ns X|$ to $|\ns Y|$, and the \deffont{conjugation} permutation action 
$$
(\pi\act f)(x)=\pi\act(f(\pi^\mone\act x)) .
$$

The conjugation action can be rephrased as `permutations distribute' (cf. Theorem~\ref{thrm.equivar} below):
\begin{lemm}
\label{lemm.useful.distrib}
If $f{\in}|\ns X{\to}\ns Y|$ then $\pi\act f(x)=(\pi\act f)(\pi\act x)$.
\end{lemm}
\begin{proof}
By easy calculations.
\end{proof}

\subsubsection{Finitely supported function space}

$\ns X\Func\ns Y$ is a nominal set with underlying set the functions from $|\ns X|$ to $|\ns Y|$ with finite support under the conjugation action, and the conjugation permutation action.

\begin{lemm}
\label{lemm.equivar.equivar}
$f\in |\ns X\Func\ns Y|$ is equivariant in the sense of Definition~\ref{defn.equivariant} if and only if it is equivariant in the sense of Definition~\ref{defn.supp}. 
\end{lemm}
\begin{proof}
We sketch the proof: If $\pi\act f=f$ then for any $x{\in}\ns X$ we have by Lemma~\ref{lemm.useful.distrib} that $\pi\act (f(x))=(\pi\act f)(\pi\act x)=f(\pi\act x)$.
Conversely if for any $x{\in}\ns X$ we have $\pi\act(f(x))=f(\pi\act x)$ then by the conjugation action 
$$
(\pi\act f)(x)=
\pi\act (f(\pi^\mone\act x))=f(\pi\act(\pi^\mone\act x))=f(x).
\qedhere$$
\end{proof}

\subsubsection{Full powerset}
\label{subsubsect.powerset}

\begin{defn}
\label{defn.pointwise.action}
Suppose $\ns X$ is a set with a permutation action.
Give subsets $X\subseteq|\ns X|$ the \deffont{pointwise} permutation action
$$
\pi\act X=\{\pi\act x\mid x\in X\} .
$$
\end{defn}

\begin{xmpl}
A useful instance of the pointwise action is for sets of atoms.
As discussed in Subsection~\ref{subsect.xmpl.atoms} above, if $a\in\mathbb A$ then $\pi\act a=\pi(a)$.
Thus if $A\subseteq\mathbb A$ then 
$$
\pi\act A\quad\text{means}\quad \{\pi(a)\mid a\in A\} .
$$
\end{xmpl}

\begin{lemm}
\label{lemm.when.set.equivar}
Continuing the notation of Definition~\ref{defn.pointwise.action}, 
$$
X\ \text{is equivariant in the sense of Definition~\ref{defn.supp}}
\quad\liff\quad
\Forall{\pi}\Forall{x{\in}|\ns X|}(x\in X\liff \pi\act x\in X) .
$$
\end{lemm}
\begin{proof}
If $X$ is equivariant then by Corollary~\ref{corr.stuff} $\Forall{\pi}\pi\act X=X$. 
It follows that $x\in X\liff x\in\pi^\mone\act X \liff \pi\act x\in X$.

Conversely if 
$\Forall{\pi}\Forall{x{\in}|\ns X|}(x\in X\liff \pi\act x\in X)$
then also $x\in X\liff \pi^\mone\act x\in X \liff x\in\pi\act X$, so that $\pi\act X=X$.
This is for any $\pi$, so using Corollary~\ref{corr.stuff} we have that $\supp(X)=\varnothing$.
\end{proof}

\begin{defn}
\label{defn.powerset}
Define $\powerset(\ns X)$ the \deffont{full powerset} of $\ns X$ to be the set with a permutation action with 
\begin{itemize*}
\item
underlying set $\{X\mid X\subseteq|\ns X|\}$ (the set of all subsets of $|\ns X|$), and 
\item
the pointwise action $\pi\act X=\{\pi\act x\mid x\in X\}$.
\end{itemize*}
\end{defn}

\begin{rmrk}
Even if $\ns X$ is a nominal set, $\powerset(\ns X)$ need not be a nominal set.
To see why, take $\ns X$ to be equal to $\mathbb A=\{a,b,c,d,e,f,\dots\}$ and consider the set 
$$
\f{comb}=\{a,c,e,\dots\}
$$
of `every other atom'.
This does not have finite support, though permutations still act on it pointwise.
For more discussion of this point, see \cite[Remark~2.18]{gabbay:fountl}.
\end{rmrk}

We consider further examples in Subsection~\ref{subsect.more.examples}.

\subsection{The principle of equivariance and the NEW quantifier}
\label{subsect.pre-equivar}

We come to Theorem~\ref{thrm.equivar}, a result which is central to the `look and feel' of nominal techniques.
It enables a particularly efficient management of renaming and $\alpha$-conversion in syntax and semantics and captures why it is so useful to use \emph{names} in the foundations of our semantics and not some other infinite set, such as numbers.

Names are by definition symmetric (i.e. can be permuted).
Taking names and permutations as \emph{primitive} implies that permutations propagate to the things we build using them.
This is the \emph{principle of equivariance} (Theorem~\ref{thrm.equivar} below; see also \cite[Subsection~4.2]{gabbay:fountl} and \cite[Lemma~4.7]{gabbay:newaas-jv}).

The principle of equivariance implies that, provided we permute names uniformly in all the parameters of our definitions and theorems, we then get another valid set of definitions and theorems.
This is not true of e.g. numbers because our mathematical foundation equips numbers by construction with numerical properties such as \emph{less than or equal to $\leq$}, which can be defined from first principles with no parameters.

So if we use numbers for names then we do not care about $\leq$ because we just needed a countable set of elements, but we repeatedly have to \emph{prove} that we did not use an asymmetric property like $\leq$.
In contrast, with nominal foundations and atoms, we do not have to explicitly prove symmetry because we can just look at our mathematical foundation and note that it is naturally symmetric under permuting names; we reserve numbers for naturally \emph{a}symmetric activities, such as counting.

This style of name management is characteristic of nominal techniques.
The reader will find it used often, e.g. in Lemmas~\ref{lemm.fresh.sub}, \ref{lemm.pow.closed}, and~\ref{lemm.pow.nu.closed}, Propositions~\ref{prop.amgis.2} and~\ref{prop.char.freshwedge}, and Definitions~\ref{defn.exact.amgis.algebra}, \ref{defn.sub.sets}, and~\ref{defn.eq.powamgis}.

\begin{rmrk}
The languages of ZFA set theory and FM set theory are identical: first-order logic with equality $=$ and sets membership $\in$.
\end{rmrk}

\begin{thrm}
\label{thrm.equivar}
\label{thrm.no.increase.of.supp}
If $\vect x$ is a list $x_1,\ldots,x_n$, write $\pi\act \vect x$ for $\pi\act x_1,\ldots,\pi\act x_n$.
Suppose $\Phi(\vect x)$ is a predicate in the language of ZFA/FM set theory, with free variables $\vect x$.
Suppose $\Upsilon(\vect x)$ is a function specified in the language of ZFA/FM set theory, with free variables $\vect x$.
Then we have the following principles:
\begin{enumerate*}
\item\label{equivar.pred}
\deffont{Equivariance of predicates.} \ $\Phi(\vect x) \liff \Phi(\pi\act \vect x)$.\footnote{Here $\vect x$ is understood to contain \emph{all} the variables mentioned in the predicate.
It is not the case that $a=a$ if and only if $a=b$---but it is the case that $a=b$ if and only if $b=a$.}
\item
\deffont{Equivariance of functions.}\quad\,$\pi\act \Upsilon(\vect x) = \Upsilon(\pi\act \vect x)$.
\item
\deffont{Conservation of support.}\quad\ \
If $\vect x$ denotes elements with finite support \\ then
$\supp(\Upsilon(\vect x)) \subseteq \supp(x_1){\cup}\cdots{\cup}\supp(x_n)$.
\end{enumerate*}
\end{thrm}
\begin{proof}
See Theorem~4.4, Corollary~4.6, and Theorem~4.7 from \cite{gabbay:fountl}.
\end{proof}

\begin{rmrk}
Theorem~\ref{thrm.equivar} is three fancy ways of observing that if a specification is symmetric in atoms,
the the result must be at least as symmetric as the inputs.
The benefit of using atoms (instead of e.g. numbers) to model names makes this a one-line argument.\footnote{%
The reasoning in this paper could in principle be fully formalised in a sets foundation with atoms, such as Zermelo-Fraenkel set theory with atoms \deffont{ZFA}.
Nominal sets can be implemented in ZFA sets such that nominal sets map to equivariant elements (elements with empty support) and the permutation action maps to `real' permutation of atoms in the model.
See \cite[Subsection~9.3]{gabbay:fountl} and \cite[Section~4]{gabbay:fountl}.
}
\end{rmrk}

\begin{prop}
\label{prop.pi.supp}
$\supp(\pi\act x)=\pi\act\supp(x)$ (which means $\{\pi(a)\mid a\in\supp(x)\}$).
\end{prop}
\begin{proof}
Immediate consequence of part~2 of Theorem~\ref{thrm.equivar}.\footnote{There is also a nice proof of this fact by direct calculations; see \cite[Theorem~2.19]{gabbay:fountl}.
However, it just instantiates Theorem~\ref{thrm.equivar} to the particular $\Upsilon$ specifying support.}
\end{proof}

\begin{frametxt}
\begin{defn}
\label{defn.New}
Write 
$\New{a}\Phi(a)$ for `$\{a\mid \neg\Phi(a)\}$ is finite'.
We call this the \deffont{$\new$ quantifier}.
\end{defn}
\end{frametxt}

\begin{rmrk}
We can read $\new$ as `for all but finitely many $a$', `for cofinitely many $a$', `for fresh $a$', or `for new $a$'.
It captures a \emph{generative} aspect of names, that for any $x$ we can find plenty of atoms $a$ such that $a\not\in\supp(x)$.
$\new$ was designed in \cite{gabbay:newaas-jv} to model the quantifier being used when we informally write 
``rename $x$ in $\lam{x}t$ to be fresh'', or ``emit a fresh channel name'' or ``generate a fresh memory cell''.
\end{rmrk}

\begin{rmrk}
$\new$ is a `for most' quantifier \cite{westerstahl:quafnl}, and is a \emph{generalised quantifier} \cite[Section 1.2.1]{keenan:genqll}.

But importantly, $\new$ over nominal sets satisfies the \emph{some/any property} that to prove a $\new$-quantified property we test it for \emph{one} fresh atom;
we may then use it for \emph{any} fresh atom.
This is Theorem~\ref{thrm.new.equiv}, which we use implicitly when later we choose a `fresh atom' without proving that it does not matter which one we choose.
We will do this often.
See for instance the proofs of Lemmas~\ref{lemm.pow.closed} and~\ref{lemm.technical.contradiction} (where we write `for fresh $a$', we are using the $\new$-quantifier) and Definitions~\ref{defn.exact.amgis.algebra}, \ref{defn.sub.sets}, and~\ref{defn.eq.powamgis} where $\new$ is made explicit.

It is important to understand that $\new$ simply means `for all but finitely many atoms'.
This can be encoded in first-order logic.
What makes this quantifier so special is the $\forall/\exists$ symmetry property which arises \emph{specifically} when $\new$ is applied to symmetric atoms with a background assumption of finite (or more generally, `small') support---i.e. in a nominal context.

So it is not any one piece of the puzzle that makes this work, but how the pieces interact when they are fit together.
This is expressed in Theorem~\ref{thrm.new.equiv}: 
\end{rmrk}

\begin{thrm}
\label{thrm.new.equiv}
Suppose $\Phi(\vect z,a)$ is a predicate in the language of ZFA/FM set theory, with free variables $\vect z,a$. 
Suppose $\vect z$ denotes elements with finite support.
Then the following are equivalent:
$$
\Forall{a}(a\in\atoms\land  a\# \vect z) \limp \Phi(\vect z,a)
\qquad
\New{a}\Phi(\vect z,a)
\qquad
\Exists{a}a\in\atoms\land a\#\vect z  \land \Phi(\vect z,a)
$$
\end{thrm}
\begin{proof}
Where convenient we may write $\vect z$ as $z_1,\dots,z_n$.
\begin{itemize}
\item
Suppose $\Phi(\vect z,a)$ holds for every atom $a\in\atoms{\setminus}\bigcup_{1{\leq}i{\leq}n}\supp(z_i)$.

By assumption $\vect z$ denotes elements with finite support, and it is a fact that a finite union of finite sets is finite, so $\atoms\setminus\bigcup_{1{\leq}i{\leq}n}\supp(z_i)$ is cofinite.

It follows that $\New{a}\Phi(\vect z,a)$ holds.
\item
Suppose $A\subseteq\atoms$ is cofinite and $\Phi(\vect z,a)$ for every $a{\in}A$.
As in the previous point, there exists some $a{\in}A$ such that $a\#z_i$ for every $1{\leq}i{\leq}n$.

It follows that
$\Exists{a{\in}\atoms}\bigl(a\#\vect z \land \Phi(\vect z,a)\bigr)$.
\item
Now suppose $\Phi(\vect z,a)$ holds for some $a\in\atoms{\setminus}\bigcup_{1{\leq}i{\leq}n}\supp(z_i)$.

By part~\ref{equivar.pred} of Theorem~\ref{thrm.equivar} $\Phi((a'\ a)\act\vect z,a')$ holds for any $a'{\in}\atoms$.
Choosing $a'\#\vect z$ we have by part~\ref{stuff.fixsupp.fixelt} of Corollary~\ref{corr.stuff} that $(a'\ a)\act z_i=z_i$ for every $1{\leq}i{\leq}n$.

Thus $\Forall{a{\in}\atoms}\bigl(a\# \vect z \limp \Phi(\vect z,a)\bigr)$ holds.
\qedhere\end{itemize}
\end{proof}

\begin{rmrk}
It is impossible to overstate the importance and convenience of the $\new$-quantifier and Theorem~\ref{thrm.new.equiv}, which appears in the literature for instance as Theorem~6.5 from \cite{gabbay:fountl} or Proposition~4.10 from \cite{gabbay:newaas-jv}.

Consider Definition~\ref{defn.eq.powamgis} and the proof of Proposition~\ref{prop.xeqx}.  In that proof, we use $\alpha$-equivalence to assume $a$ is fresh for $u$ and $v$.  
By of the some/any property we can immediately apply Definition~\ref{defn.eq.powamgis}.

\emph{If} we had used a $\forall$ or $\exists$ quantifier in Definition~\ref{defn.eq.powamgis} then we would have had to worry whether the fresh $a$ in the proof of Proposition~\ref{prop.xeqx} was the same fresh atom as that used in Definition~\ref{defn.eq.powamgis}.  
`Obviously' this is just a hassle; `obviously' this choice does not matter.
Theorem~\ref{thrm.new.equiv} and the $\new$-quantifier capture precisely and succinctly what that word `obvious' means.
\end{rmrk}

\subsection{Further examples}
\label{subsect.more.examples}

We now consider the finitely supported powerset and the strictly finitely supported powerset.
These examples are more technically challenging and will be key to the later constructions.

\subsubsection{Finitely supported powerset}
\label{subsect.finsupp.pow}

$\nompow(\ns X)$ (the \deffont{nominal powerset}) is a nominal set, with 
\begin{itemize*}
\item
underlying set those $X{\in}|\powerset(\ns X)|$ that are finitely supported, and
\item
with the \deffont{pointwise} action $\pi\act X=\{\pi\act x\mid x\in X\}$ inherited from Definition~\ref{defn.pointwise.action}.
\end{itemize*}
As the name suggests, the nominal powerset is the powerset object in the category of nominal sets \cite[Lemma~9.10]{gabbay:fountl}.

A common source of confusion is to suppose that if $A$ supports $X{\in}|\nompow(\ns X)|$ then $A$ must support every $x\in X$.
This is incorrect:
\begin{lemm}
\label{lemm.pow.not.true}
It is not true in general that if $X{\in}|\nompow(\ns X)|$ and $x\in X$ then $\supp(x)\subseteq\supp(X)$.
\end{lemm}
\begin{proof}
It suffices to provide a counterexample.
Take $\ns X=\mathbb A$ (the nominal set of atoms with the natural permutation action, from Subsection~\ref{subsect.xmpl.atoms}) and $X=\mathbb A\subseteq|\mathbb A|$ (the underlying set of the nominal set of all atoms, i.e. the set of all atoms!).

It is easy to check that $\supp(X)=\varnothing$ and $a\in X$ and $\supp(a)=\{a\}\not\subseteq\varnothing$.
\end{proof}

Lemma~\ref{lemm.pow.not.true} will lead us to the notion of the \emph{strictly finitely supported} powerset in a moment.
For completeness we take just a moment to mention Lemma~\ref{lemm.weak.pow.true}, which describes a weaker property than that of Lemma~\ref{lemm.pow.not.true} which \emph{is} valid in general; see \cite[Lemma~5.2]{gabbay:forcie} and \cite[Corollary~4.30]{gabbay:nomuae} for applications, and \cite[Lemma~7.6.2]{gabbay:nomtnl} for the more general context.
\begin{lemm}
\label{lemm.weak.pow.true}
If $X{\in}|\nompow(\ns X)|$ and $a\#X$ then there exists an $x\in X$ with $a\#x$.
\end{lemm}
\begin{proof}
Choose $x'\in X$ and fresh $b$ (so $b\#x',X$).
By Corollary~\ref{corr.stuff} $(b\ a)\act X=X$, so $(b\ a)\act x\in X$.
By Proposition~\ref{prop.pi.supp} $a\#(b\ a)\act x$. 
\end{proof}

\subsubsection{Strictly finitely supported powerset}
\label{subsect.strict.pow}

Suppose $\ns X$ is a nominal set.

\begin{defn}
\label{defn.strictpow}
Call $X\subseteq|\ns X|$ \deffont{strictly supported} by $A\subseteq\mathbb A$ when 
$$
\Forall{x{\in} X} \supp(x)\subseteq A .
$$
If there exists some finite $A$ which strictly supports $X$, then call $X$ \deffont{strictly finitely supported} (see \cite[Theorem~2.29]{gabbay:fountl}).
 
Write $\strict(\ns X)$ for the set of strictly finitely supported $X\subseteq|\ns X|$.
That is:
$$
\strict(\ns X)=\{X\subseteq|\ns X|\mid \Exists{A{\subseteq}\mathbb A}A\text{ finite}\land X\text{ strictly supported by }A\}
$$
\end{defn}

\begin{lemm}
\label{lemm.strict.support}
If $X\in\strict(\ns X)$ then:
\begin{enumerate*}
\item
$\bigcup\{\supp(x)\mid x{\in}X\}$ is finite.
\item
\label{item.strict.support.elts}
$\bigcup\{\supp(x)\mid x{\in}X\}=\supp(X)$. 
\item
If $X\subseteq|\ns X|$ is strictly finitely supported then it is finitely supported. 
\item
$x\in X$ implies $\supp(x)\subseteq\supp(X)$ (contrast this with Lemma~\ref{lemm.pow.not.true}).
\item
$\strict(\ns X)$ with the pointwise permutation action is a nominal set. 
\end{enumerate*}
\end{lemm}
\begin{proof}
The first part is immediate since by assumption there is some finite $A{\subseteq}\mathbb A$ that bounds $\supp(x)$ for all $x\in X$.
The second part follows by an easy calculation using part~3 of Corollary~\ref{corr.stuff}; full details are in \cite[Theorem~2.29]{gabbay:fountl}, of which Lemma~\ref{lemm.strict.support} is a special case. 
The other parts follow by definitions from the first and second parts.
\end{proof}

\begin{xmpl}
\begin{enumerate*}
\item
$\varnothing\subseteq|\mathbb A|$ is finitely and strictly finitely supported by $\varnothing$.
\item
$\{a\}$ is finitely supported by $\{a\}$ and also strictly finitely supported by $\{a\}$.
\item
$\mathbb A\subseteq|\mathbb A|$ is finitely supported by $\varnothing$ but not strictly finitely supported.
\item
$\mathbb A{\setminus}\{a\}$ is finitely supported by $\{a\}$ but not strictly finitely supported.
\end{enumerate*}
\end{xmpl}

\section{Algebras over nominal sets}
\label{sect.fol-algebra}

\subsection{Definitions: sigma-algebra and amgis-algebra} 
\label{subsect.sigma.amgis}

\maketab{tab0}{@{\hspace{2em}}L{6em}@{\ }L{4em}@{\ }R{8em}@{\ }L{10em}@{\ }L{10em}}

\begin{figure}[tH]
\begin{minipage}{\textwidth}
\begin{tab0}
\rulefont{\sigma a} && a[a\sm x]=&x
\\[.5ex]
\rulefont{\sigma id} && x[a \sm a]=&x
\\
\rulefont{\sigma\#} &a\#x\limp& x[a \sm u]=&x
\\
\rulefont{\sigma\alpha}&b\#x\limp&x[a \sm u]=&((b\;a)\act x)[b \sm u]
\\
\rulefont{\sigma\sigma} &a\#v\limp & x[a \sm u][b \sm v]=&x[b \sm v][a \sm u[b \sm v]]
\\[2.2ex]
\rulefont{\amgis\sigma}&a\#v\limp& p[v \ms b][u \ms a]=&p[u[b \sm v] \ms a][v \ms b]
\end{tab0}
\end{minipage}
\caption{Nominal algebra axioms for $\sigma$ and $\protect\amgis$} 
\label{fig.nom.sigma}
\label{fig.amgis}
\end{figure}

Definitions~\ref{defn.term.sub.alg},\ \ref{defn.sub.algebra},\ and~\ref{defn.bus.algebra} assemble three key technical structures (see also Definitions~\ref{defn.powamgis} and~\ref{defn.powsigma}).

\begin{defn}
\label{defn.term.sub.alg}
A \deffont{termlike $\sigma$-algebra} is a tuple $\ns U=(|\ns U|,\act,\tf{sub}_{\ns U},\tf{atm}_{\ns U})$ of: 
\begin{itemize*}
\item
a nominal set $(|\ns U|,\act)$ which we may write just as $\ns U$; and
\item
an equivariant \deffont{$\sigma$-action} $\tf{sub}_{\ns U}:(\ns U\times\mathbb A\times\ns U)\Func \ns U$, written infix $v[a{\sm}u]$; and 
\item
an equivariant injection $\tf{atm}_{\ns U}:\mathbb A\Func\ns U$, 
\end{itemize*}
such that the equalities \rulefont{\sigma a}, \rulefont{\sigma id}, \rulefont{\sigma\#}, \rulefont{\sigma\alpha}, and \rulefont{\sigma\sigma} of Figure~\ref{fig.nom.sigma} hold, where $x$, $u$, and $v$ range over elements of $|\ns U|$.
We usually write $\tf{atm}_{\ns U}$ invisibly (so we write $\tf{atm}_{\ns U}(a)$ just as $a$), and we may omit subscripts (so we may write $\tf{sub}_{\ns U}$ as $\tf{sub}$ if we are confident that $\ns U$ is clear and understood). 
\end{defn}

\begin{rmrk}
We unpack what equivariance from Definition~\ref{defn.equivariant} means for the $\sigma$-action from Definition~\ref{defn.sub.algebra}: for every $x{\in}|\ns X|$,\ atom $a$, and $u{\in}|\ns U|$, and for every permutation $\pi$, we have that
$$
\pi\act(x[a\sm u])=(\pi\act x)[\pi(a)\sm\pi\act u] .
$$
Similarly for the equivariant $\amgis$-action in Definition~\ref{defn.bus.algebra} below.
\end{rmrk}

\begin{xmpl}
\label{xmpl.fot}
First-order terms form a termlike $\sigma$-algebra:
\begin{itemize*}
\item
Variables are atoms, 
\item
the permutation action is pointwise, and 
\item
the $\sigma$-action is `real' substitution.
\end{itemize*}
Consider for instance a first-order term language for arithmetic: then $a*b$ is a term, $\pi\act(a*b)$ is $\pi(a)*\pi(b)$, and $(a*b)[a\sm 2][b\sm 3]$ is $2*3$.
Such a syntax is defined in this paper in Definition~\ref{defn.terms.and.predicates} as 
$
r::= a \mid \tf f(r_1,\dots,r_{\ar(\tf f)}).
$

Similarly, untyped $\lambda$-terms quotiented by $\alpha$-equivalence form a termlike $\sigma$-algebra, where the $\sigma$-action is capture-avoiding substitution.
So $[(\lam{a}b)]_\alpha[b\sm a]=[\lam{a'}a]_\alpha$.

Predicates of first-order logic quotiented by $\alpha$-equivalence are not a termlike $\sigma$-algebra under capture-avoiding substitution action for terms 
because predicates do not belong to the same class as terms $u$.
However, predicates do form a (not-necessarily-termlike) $\sigma$-algebra; see Definition~\ref{defn.sub.algebra}. 
\end{xmpl}

\begin{rmrk}
The `$\sigma$' in $\sigma$-action stands for `substitution'.
No connection is suggested with the notion of sigma-algebra from measure theory.
\end{rmrk}

\begin{defn}
\label{defn.sub.algebra}
Suppose $\ns U=(|\ns U|,\act,\tf{sub},\tf{atm})$ is a termlike $\sigma$-algebra.
A \deffont{$\sigma$-algebra} over $\ns U$ is a tuple $\ns X=(|\ns X|,\act,\ns U,\tf{sub})$ of:
\begin{itemize*}
\item
A nominal set $(|\ns X|,\act)$ which we may write just as $\ns X$; and
\item
an equivariant \deffont{$\sigma$-action} $\tf{sub}_{\ns X}:(\ns X\times\mathbb A\times\ns U)\Func\ns X$, written infix $x[a{\sm}u]$; 
\end{itemize*}
such that the equalities \rulefont{\sigma id}, \rulefont{\sigma\#}, \rulefont{\sigma\alpha}, and \rulefont{\sigma\sigma} of Figure~\ref{fig.nom.sigma} hold,%
\footnote{That is, the $\sigma$ axioms except \rulefont{\sigma a}, since we do not assume a function $\tf{atm}_{\ns X}$.} 
where $x$ ranges over elements of $|\ns X|$ and $u$ and $v$ range over elements of $|\ns U|$.
As for termlike $\sigma$-algebras, we may omit the subscript $\ns X$.
\end{defn}

\begin{defn}
\label{defn.bus.algebra}
Suppose $\ns U=(|\ns U|,\act,\tf{sub},\tf{atm})$ is a termlike $\sigma$-algebra. 

An \deffont{$\amgis$-algebra} (spoken: \deffont{amgis}-algebra) over $\ns U$ is a tuple $\ns P=(|\ns P|,\act,\ns U,\tf{amgis}_{\ns P})$ of: 
\begin{itemize*}
\item
a set with a permutation action $(|\ns P|,\act)$ which we may write just as $\ns P$; and 
\item
an equivariant \deffont{amgis}-action $\tf{amgis}_{\ns P}:(\ns P\times\ns U\times\mathbb A)\Func \ns P$, written infix $p[u\ms a]$, 
\end{itemize*}
such that the equality \rulefont{\amgis \sigma} of Figure~\ref{fig.nom.sigma} holds, where $p$ ranges over elements of $|\ns P|$ and $u$ and $v$ range over elements of $|\ns U|$.
We may omit the subscript $\ns P$.
\end{defn}

\begin{rmrk}
$[u\ms a]$ looks like $[a\sm u]$ written backwards, and a casual glance at \rulefont{\amgis\sigma} suggests that it is just \rulefont{\sigma\sigma} written backwards.
This is not quite true: we have $u[b\sm v]$ on the right in \rulefont{\amgis\sigma} and not `$u[v\ms b]$' (which would make no sense, since $\ns U$ has no amgis-action).

Discussion of the origin of the axioms of $\amgis$-algebras is in Subsections~\ref{subsect.sigma.to.amgis} and~\ref{subsect.amgis.to.sigma}; see also Proposition~\ref{prop.amgis.2} and Remark~\ref{rmrk.subsequent}.
\end{rmrk}

We conclude this subsection with some technical lemmas which will be useful later.

Lemma~\ref{lemm.fresh.sub} is a technical corollary of \rulefont{\sigma\alpha}, and is used in Lemma~\ref{lemm.sim.sigma}:
\begin{lemm}
\label{lemm.fresh.sub}
If $a\#u$ then $a\#x[a\sm u]$.
As a corollary, 
$$
\supp(x[a\sm u])\subseteq(\supp(x){\setminus}\{a\})\cup\supp(u) .
$$
\end{lemm}
\begin{proof}
Choose fresh $b$ (so $b\#x,u$).
By \rulefont{\sigma\alpha} $x[a\sm u]=((b\ a)\act x)[b\sm u]$.
Also by part~1 of Corollary~\ref{corr.stuff} $(b\ a)\act u=u$ and by Theorem~\ref{thrm.equivar} $(b\ a)\act(x[a\sm u])=((b\ a)\act x)[b\sm (b\ a)\act u]$.
We put this all together and we deduce that $(b\ a)\act (x[a\sm u])=x[a\sm u]$.
It follows by part~3 of Corollary~\ref{corr.stuff} that $a\not\in\supp(x[a\sm u])$.

The corollary follows since by Theorem~\ref{thrm.no.increase.of.supp} $\supp(x[a\sm u])\subseteq\supp(x){\cup}\{a\}{\cup}\supp(u)$.
\end{proof}

\begin{rmrk}
The reader should know Lemma~\ref{lemm.fresh.sub} for concrete syntax, which is a $\sigma$-algebra.

For instance, if $\phi$ is a predicate of first-order logic (up to $\alpha$-equivalence) and $r$ is a term, then \emph{support} coincides with \emph{free atoms/variables} $\fa(\text{-})$ and $\fa(\phi[a\sm r])\subseteq(\fa(\phi){\setminus}\{a\})\cup\fa(r)$.  

But, Lemma~\ref{lemm.fresh.sub} is an abstract property of models of nominal algebra axioms.
Syntax is one concrete instance of the abstract class (as natural numbers are a concrete instance of rings). 
\end{rmrk}

Lemma~\ref{lemm.sub.alpha} goes back to \cite{gabbay:capasn,gabbay:capasn-jv}, where it was taken as an axiom or the original nominal algebraic treatment of substitution. 
In the presence of \rulefont{\sigma\alpha} it is equivalent to \rulefont{\sigma id};
it is useful in the proofs of Propositions~\ref{prop.char.freshwedge} and~\ref{prop.eq.closed}: 
\begin{lemm}
\label{lemm.sub.alpha}
If $b\#x$ then $x[a\sm b]=(b\ a)\act x$.
\end{lemm}
\begin{proof}
By \rulefont{\sigma id} $(b\ a)\act x=((b\ a)\act x)[b\sm b]$.
By \rulefont{\sigma\alpha} $((b\ a)\act x)[b\sm b]=x[a\sm b]$.
\end{proof}

Lemma~\ref{lemm.a.fresh.bigset} will be useful to prove Proposition~\ref{prop.char.freshwedge} and~\ref{prop.char.freshwedge.names}.  See also the related Lemma~\ref{lemm.fresh.ii}:
\begin{lemm}
\label{lemm.a.fresh.bigset}
If $U{\subseteq}|\ns U|$ is finitely supported and $a\#U$ then $a\#\{x[a\sm u]\mid u{\in}U\}$.
As corollaries, the following freshnesses all hold: 
$$
\begin{array}{l}
a\#\{x[a\sm u]\mid u{\in}|\ns U|\}
\\
a\#\{x[a\sm n]\mid n{\in}\mathbb A\}
\end{array}
$$
\end{lemm}
\begin{proof}
We use part~3 of Corollary~\ref{corr.stuff}.
Choose fresh $b$ (so $b\#x,U$).
Note by Corollary~\ref{corr.stuff} (since $a,b\#U$) that $(b\ a)\act U=U$ so that by Lemma~\ref{lemm.when.set.equivar} $\Forall{u{\in}|\ns U|}u\in U\liff (b\ a)\act u\in U$.
Then we reason as follows:
$$
\begin{array}[b]{r@{\ }l@{\qquad}l}
(b\ a)\act \{x[a\sm u]\mid u{\in}U\}
=&
\{(b\ a)\act (x[a\sm u])\mid u{\in}U\}
&\text{Pointwise action}
\\
=&
\{((b\ a)\act x)[b\sm (b\ a)\act u]\mid u{\in}U\}
&\text{Theorem~\ref{thrm.equivar}}
\\
=&
\{((b\ a)\act x)[b\sm u]\mid u{\in}U\}
&(b\ a)\act U=U
\\
=&
\{x[a\sm u]\mid u{\in}|\ns U|\}
&\rulefont{\sigma\alpha},\ b\#x
\end{array}
$$
The proof for $a\#\{x[a\sm n]\mid n{\in}\mathbb A\}$ is similar.
\end{proof}

\maketab{tab1}{@{\hspace{-2em}}R{10em}@{\ }L{12em}L{12em}}
\maketab{tab2}{@{\hspace{-2.3em}}R{10em}@{\ }L{14em}L{13.25em}}
\maketab{tab2r}{@{\hspace{-0.3em}}R{10em}@{\ }L{14em}L{12em}}
\maketab{tab2b}{@{\hspace{-5em}}R{10em}@{\ }L{19.5em}L{12em}}
\maketab{tab3}{@{\hspace{-3em}}R{10em}@{\ }L{10em}L{16.25em}}

\subsection{Duality I: sigma to amgis} 
\label{subsect.sigma.to.amgis}

In Subsection~\ref{subsect.sigma.to.amgis} we explore how to move from a $\sigma$-algebra to an $\amgis$-algebra; we explore the other direction in Subsection~\ref{subsect.amgis.to.sigma}.

Given a $\sigma$-algebra we generate an $\amgis$-algebra out of its subsets.
This is Proposition~\ref{prop.amgis.2}.

\begin{defn}
\label{defn.p.action}
Suppose $\ns X=(|\ns X|,\act,\ns U,\tf{sub}_{\ns X})$ is a $\sigma$-algebra over a termlike $\sigma$-algebra $\ns U$.

Give subsets $p\subseteq|\ns X|$ \deffont{pointwise} actions as follows:
\begin{frameqn}
\begin{array}{r@{\ }l@{\qquad}l}
\pi\act p=&\{\pi\act x\mid x\in p\}
\\
p[u\ms a]=&\{x{\in}|\ns X| \mid x[a\sm u]\in p\} 
&u{\in}|\ns U|
\end{array}
\end{frameqn}
\end{defn}

\begin{prop}
\label{prop.sigma.iff}
Suppose $\ns X$ is a $\sigma$-algebra over a termlike $\sigma$-algebra $\ns U$. 
Suppose $p\subseteq|\ns X|$.
Then:
\begin{itemize*}
\item
$x\in p[u\ms a]$ if and only if $x[a\sm u]\in p$.
\item
$x\in \pi\act p$ if and only if $\pi^\mone\act x\in p$.
\end{itemize*}
\end{prop}
\begin{proof}
By easy calculations on the pointwise actions in Definition~\ref{defn.p.action}.
\end{proof}

\begin{rmrk}
We take a moment to suggest intuitively why Definition~\ref{defn.p.action} is natural.

The pointwise permutation action is the natural action on subsets, also mentioned in Subsection~\ref{subsect.finsupp.pow}. 
This comes from the Fraenkel-Mostowski foundations.

The pointwise $\amgis$-action is its natural generalisation from a group to a monoid (which need not necessarily have inverses).
But why \emph{amgis}?  Why does $[a\sm u]$ get turned round to $[u\ms a]$?

When we take the powerset of a set $V$, $V$ is in \emph{negative position} (i.e. $\powerset(V)$ is equivalent to a function-space ${V{\to}2}$ and here $V$ is to the left of the arrow).
This implies that any modal operations which we assume on elements of $V$, need to be `flipped'.
Thus, $\sigma$-algebras turn into $\amgis$-algebras and vice versa.
This is why an underlying set with a $\sigma$-action gives rise to a `flipped'---a \emph{dual}---$\amgis$-action on the subsets.
\end{rmrk}

\begin{defn}
\label{defn.powamgis}
Suppose $\ns X$ is a $\sigma$-algebra over a termlike $\sigma$-algebra $\ns U$.

Define the \deffont{$\amgis$-powerset} algebra $\powamgis(\ns X)$ by setting 
\begin{itemize*}
\item
$|\powamgis(\ns X)|$ to be the set of all subsets $p\subseteq|\ns X|$ (Definition~\ref{defn.powerset}) with 
\item
permutation action $\pi\act p$ and 
\item
amgis-action $p[u\ms a]$ from Definition~\ref{defn.p.action}.
\end{itemize*}
\end{defn}

\begin{rmrk}
\label{rmrk.p.not.finite.support.1}
Note that $p\subseteq|\ns X|$ need not have finite support in Definition~\ref{defn.powamgis}.

Our notion of $\amgis$-algebra (Definition~\ref{defn.bus.algebra}) admits $p$ without finite support.
We will need this: the $\amgis$-algebras we construct in Definition~\ref{defn.points} need not have finite support; see the discussion in Remark~\ref{rmrk.p.not.finite.support}.
The action happens in Theorem~\ref{thrm.maxfilt.zorn} where we use Zorn's Lemma to make infinitely many choices.
\end{rmrk}

\begin{prop}
\label{prop.amgis.2}
If $\ns X$ is a $\sigma$-algebra over a termlike $\sigma$-algebra $\ns U$ then $\powamgis(\ns X)$ (Definition~\ref{defn.powamgis}) is an $\amgis$-algebra over $\ns U$.
\end{prop}
\begin{proof}
By Theorem~\ref{thrm.equivar} the operations are equivariant.
We verify rule \rulefont{\amgis\sigma} from Figure~\ref{fig.amgis}:
\begin{itemize*}
\item
\emph{Property \rulefont{\amgis\sigma}.}\quad
We use \rulefont{\sigma\sigma}.
Suppose $a\#v$.
Then:
\begin{tab2}
x\in p[v\ms b][u\ms a]\liff &x[a{\sm}u][b{\sm}v]\in p 
&\text{Proposition~\ref{prop.sigma.iff}}
\\
\liff& x[b{\sm}v][a{\sm}u[b{\sm}v]]\in p
&\rulefont{\sigma\sigma},\ a\#v
\\
\liff& x\in p[u[b{\sm}v]\ms a][v\ms b] 
&\text{Proposition~\ref{prop.sigma.iff}}
\qedhere\end{tab2}
\end{itemize*}
\end{proof}

\begin{rmrk}
\label{rmrk.subsequent}
It is interesting to note a non-result of $\powamgis(\ns X)$.
Consider \rulefont{\sigma\#}; let us try to dualise it as we dualised \rulefont{\sigma\sigma} to \rulefont{\amgis\sigma} in Proposition~\ref{prop.amgis.2}.
\maketab{taba}{R{13em}@{\ }L{6em}@{}L{10em}}
\begin{taba}
\rulefont{\amgis\#}\qquad 
a\#p\limp p[u\ms a]=&p & \text{Bad axiom} 
\end{taba}
Suppose $a\#p$.
Then
\begin{taba}
x\in p[u\ms a]\liff & x[a\sm u]\in p &\text{Proposition~\ref{prop.sigma.iff}}
\\
\liff& \text{????}
\end{taba}
and now we are stuck: $a\#p$ does not imply that $a\#x$, so we cannot use \rulefont{\sigma\#} (e.g. if $p=\mathbb A$ then $b\in \mathbb A$ and $b\in\supp(b)$; more on this in \cite[Subsection~9.5]{gabbay:fountl}).
\end{rmrk}

\begin{rmrk}
It has been suggested that if $\sigma$-algebra algebraises substitution on syntax, then $\amgis$-algebra should algebraise pattern-matching on syntax.
This is not quite right; pattern-matching does not satisfy \rulefont{\amgis\sigma}.
Suppose a syntax of terms with pairing $(r_1,r_2)$ and one unary term-former $\tf f$.
Consider the term $(\tf f(c),\tf f(b))$ and interpret $[a\sm r]$ as substitution and $[r\ms a]$ as pattern-matching.
Then
$$
\begin{array}{l}
(\tf f(c),\tf f(b))[c\ms b][\tf f(b)\ms a]=(\tf f(b),\tf f(b))[\tf f(b)\ms a]=(a,a)
\quad\text{and}
\\
(\tf f(c),\tf f(b))[\tf f(b)[b\sm c]\ms a][c\ms b]=
(\tf f(c),\tf f(b))[\tf f(c)\ms a][c\ms b]=
(a,\tf f(b))[c\ms b]=
(a,\tf f(b))
\end{array}
$$
and $(a,a)\neq(a,\tf f(b))$. 

The natural examples 
of $\amgis$-algebras are the $\amgis$-powersets as constructed above in Definition~\ref{defn.powamgis}.
\end{rmrk}

\subsection{Exact amgis-algebra} 
\label{subsect.exact.amgis}

We now strengthen Proposition~\ref{prop.amgis.2} to Proposition~\ref{prop.powamgis.exact}, which states that $\powamgis(\ns X)$ is an \emph{exact} $\amgis$-algebra (Definition~\ref{defn.exact.amgis.algebra}).

Exactness will be useful later to interpret equality in $\powsigma$; see Subsection~\ref{subsect.powsigma.equality} and in particular Lemma~\ref{lemm.eq.sanity.check}.

To understand exactness, it is interesting to consider an easy property of $\sigma$-algebras:
\begin{lemm}
\label{lemm.garbage}
If $a\#x,y$ then $x[a\sm u]=y[a\sm u]$ implies $x=y$.
\end{lemm}
\begin{proof}
Immediate from \rulefont{\sigma\#}.
\end{proof}

We cannot hope to replicate the proof above directly for $\amgis$-algebras using an axiom \rulefont{\amgis\#} because, as noted in Remark~\ref{rmrk.subsequent}, we do \emph{not} want such an axiom. 

But $\amgis$ is a kind of dual to $\sigma$, and \emph{exactness} is a kind of dual to Lemma~\ref{lemm.garbage}. 
Recall the definition of $\amgis$-algebra from Definition~\ref{defn.bus.algebra} and the $\new$-quantifier from Definition~\ref{defn.New}.

\begin{frametxt}
\begin{defn}
\label{defn.exact.amgis.algebra}
Call an $\amgis$-algebra $\ns P$ \deffont{exact} when if $p,q{\in}|\ns P|$ and $u{\in}|\ns U|$ then 
$$
\New{c}p[u\ms c]=q[u\ms c]
\quad\text{implies}\quad
p=q.
$$
\end{defn}
\end{frametxt}

In words, $\ns P$ is exact when for every $p$, $q$, and $u$, if $p[u\ms c]$ and $q[u\ms c]$ are equal for \emph{most} (meaning `for all but finitely meany') $c$ then $p$ and $q$ are equal.

\maketab{tab6}{@{\hspace{1em}}R{10em}@{\ }L{12em}L{12em}}
\maketab{tab7}{@{\hspace{1em}}R{12em}@{\ }L{12em}}

\begin{prop}
\label{prop.powamgis.exact}
Suppose $\ns X$ is a $\sigma$-algebra over a termlike $\sigma$-algebra $\ns U$.
Then $\powamgis(\ns X)$ is an exact $\amgis$-algebra over $\ns U$.
\end{prop}
\begin{proof}
By Proposition~\ref{prop.amgis.2} $\powamgis(\ns X)$ is an $\amgis$-algebra over $\ns U$.
It remains to prove exactness.

Suppose $p,q{\in}|\powamgis(\ns X)|$ and $u{\in}|\ns U|$, and suppose 
$\New{c} p[u\ms c]=q[u\ms c]$. 
We now check that for all $x{\in}|\ns X|$ it is the case that $x\in p$ if and only if $x\in q$.
We reason as follows:
\begin{tab6}
x\in p\liff& \New{c} x[c\sm u]\in p
&\rulefont{\sigma\#},\ c\#x
\\
\liff& \New{c} x\in p[u\ms c]
&\text{Proposition~\ref{prop.sigma.iff}}
\\
\liff& \New{c} x\in q[u\ms c]
&\text{Assumption}
\\
\liff& \New{c} x[c\sm u]\in q 
&\text{Proposition~\ref{prop.sigma.iff}}
\\
\liff& x\in q
&\rulefont{\sigma\#},\ c\#x
\qedhere\end{tab6}
\end{proof}

\begin{xmpl}
The set of atoms $\mathbb A$ is a termlike $\sigma$-algebra over itself (so $a[a\sm c]=c$ and $a[b\sm c]=a$).
By Proposition~\ref{prop.amgis.2} the full powerset of $\mathbb A$ is an $\amgis$-algebra over $\mathbb A$, and by Proposition~\ref{prop.powamgis.exact} this $\amgis$-algebra is exact.

So suppose $p,q\subseteq\mathbb A$ and $\New{c}p[a\ms c]=q[a\ms c]$. 
This means that for all but finitely many atoms $c$ it is the case that for every $x{\in}\mathbb A$,\ $x[c\sm a]\in p$ if and only if $x[c\sm a]\in q$.
In particular for every $x$ there exists some $c\neq x$ such that $x\in p$ if and only if $x\in q$.
Thus, $p$ and $q$ are equal.
\end{xmpl}

Note that exactness is not an algebraic property (it has the form \emph{if} \dots \emph{then} rather than the form \emph{LHS = RHS}).
So although the class of all $\amgis$-algebras is algebraic, the class of \emph{exact} $\amgis$-algebras 
is not.\footnote{To be precise, the class of all $\amgis$-algebras is a \emph{variety} in the nominal algebraic sense of \cite{gabbay:nomahs}, whereas the class of exact $\amgis$-algebras is a \emph{quasivariety.}} 
This will not be a problem.

\subsection{Duality II: amgis to sigma} 
\label{subsect.amgis.to.sigma}

In Subsection~\ref{subsect.sigma.to.amgis} we showed how to build an $\amgis$-algebra out of a $\sigma$-algebra.
Dually, we can build a $\sigma$-algebra out of an $\amgis$-algebra; this is a little harder because, while we are free to define $\amgis$-algebra to suit ourselves, the notion of $\sigma$-algebra is already quite fixed by the behaviour of substitution, which it is intended to model.
The set of all finitely supported subsets is too large, so Definition~\ref{defn.powsigma} cuts this down with additional conditions~\ref{item.fresh.powsigma} and~\ref{item.alpha.powsigma}.
 
\subsubsection{The pointwise sigma-action on subsets of an amgis-algebra}

\begin{defn}
\label{defn.sub.sets}
Suppose $\ns P=(|\ns P|,\act,\ns U,\tf{amgis}_{\ns P})$ is an $\amgis$-algebra over a termlike $\sigma$-algebra $\ns U$. 
Give subsets $X\subseteq|\ns P|$ \deffont{pointwise} actions as follows: 
\begin{frameqn} 
\begin{array}{r@{\ }l@{\qquad}l}
\pi\act X=&\{\pi\act x \mid x\in X\}
\\
X[a{\sm}u]=&\{p{\in}|\ns P| \mid \New{c} p[u\ms c]\in (c\ a)\act X\}
&u{\in}|\ns U|
\end{array}
\end{frameqn}
\end{defn}

\begin{prop}
\label{prop.amgis.iff}
Suppose $\ns P$ is an $\amgis$-algebra over a termlike $\sigma$-algebra $\ns U$ and suppose $X\subseteq|\ns P|$.\footnote{We care most about the case that $X\in|\nompow(\ns P)|$---$X$ is finitely supported---but this result does not depend on that.}
Suppose $p{\in}|\ns P|$ and $u{\in}|\ns U|$ and $a\#u$.
Then:
\begin{enumerate*}
\item
$p\in X[a\sm u]$ if and only if $\New{c}p[u\ms c]\in (c\ a)\act X$.\footnote{Recall that $p$ need not have finite support here; see Remark~\ref{rmrk.p.not.finite.support.1}.}
\item
$p\in \pi\act X$ if and only if $\pi^\mone\act p\in X$.
\end{enumerate*}
\end{prop}
\begin{proof}
\begin{enumerate*}
\item
Direct from Definition~\ref{defn.sub.sets}.
\item
Direct from Theorem~\ref{thrm.equivar}.
\qedhere\end{enumerate*}
\end{proof}

\begin{rmrk}
Definitions~\ref{defn.p.action} and~\ref{defn.sub.sets} are not perfectly symmetric; Definition~\ref{defn.sub.sets} contains a $\new$-quantifier.
It is there specifically to make Lemma~\ref{lemm.sigma.alpha} true; in other words we guarantee axiom \rulefont{\sigma\alpha} from Figure~\ref{fig.nom.sigma} by construction:\footnote{The proof of Lemma~\ref{lemm.eq.aeq} will be in one line precisely because of this design for Definition~\ref{defn.sub.sets}.}
\end{rmrk}

\begin{lemm}[$\alpha$-equivalence]
\label{lemm.sigma.alpha}
Suppose $\ns P$ is an $\amgis$-algebra over a termlike $\sigma$-algebra $\ns U$ and suppose $X\subseteq|\ns P|$. 
Then if $b\#X$ then $X[a\sm u]=((b\ a)\act X)[b\sm u]$.
\end{lemm}
\begin{proof}
By part~1 of Proposition~\ref{prop.amgis.iff} $p\in X[a\sm u]$ if and only $\New{c}p[u\ms c]\in (c\ a)\act X$, and $p\in ((b\ a)\act X)[b\sm u]$ if and only if $\New{c}p[u\ms c]\in (c\ b)\act((b\ a)\act X)$.
By Corollary~\ref{corr.stuff} $(c\ a)\act X=(c\ b)\act ((b\ a)\act X)$ since $b\#X$.
The result follows.
\end{proof}

Lemma~\ref{lemm.sub.sub} is useful, amongst other things, in Lemma~\ref{lemm.pow.closed}.
On syntax it is known as the \emph{substitution lemma}, but here it is about an action on sets $X$, and the proof is different:
\begin{lemm}
\label{lemm.sub.sub}
Suppose $\ns P$ is an $\amgis$-algebra over a termlike $\sigma$-algebra $\ns U$ and suppose $X\subseteq|\ns P|$. 
Suppose $u,v{\in}|\ns U|$.
Then
$$
a\#v\quad\text{implies}\quad
X[a{\sm}u][b{\sm}v]=X[b{\sm}v][a{\sm}u[b{\sm}v]].
$$
\end{lemm}
\begin{proof}
We reason as follows, where we write $\pi=(a'\ a)\circ(b'\ b)$:
$$
\begin{array}[b]{r@{\ }l@{\quad}l}
p\in X[a\sm u][b\sm v]\liff &\New{a',b'}p[v\ms b'][(b'\ b)\act u\ms a']\in \pi\act X &\text{Proposition~\ref{prop.amgis.iff}} 
\\
\liff &\New{a',b'}p[((b'\ b)\act u)[b'\sm v]\ms a'][v\ms b']\in \pi\act X &\rulefont{\amgis\sigma}\ a'\#v
\\
\liff &\New{a',b'}p[u[b\sm v]\ms a'][v\ms b']\in \pi\act X &\rulefont{\sigma\alpha}\ b'\#u
\\
\liff &\New{a'}p[u[b\sm v]\ms a']\in ((a'\ a)\act X)[b\sm v] &\text{Proposition~\ref{prop.amgis.iff}}
\\
\liff &\New{a'}p[u[b\sm v]\ms a']\in ((a'\ a)\act X)[b\sm (a'\ a)\act v] &\text{Corollary~\ref{corr.stuff}}\ a',a\#v
\\
\liff &\New{a'}p[u[b\sm v]\ms a']\in (a'\ a)\act (X[b\sm v]) &\text{Theorem~\ref{thrm.equivar}}
\\
\liff &p\in X[b\sm v][a\sm u[b\sm v]] &\text{Proposition~\ref{prop.amgis.iff}}
\end{array}
\qedhere$$
\end{proof}

\subsubsection{The $\sigma$-powerset $\powsigma(\ns P)$}

Recall from Subsection~\ref{subsect.finsupp.pow} the \emph{finitely supported powerset} $\nompow(\ns X)$ of a nominal set $\ns X$.
\begin{defn}
\label{defn.powsigma}
Suppose $\ns P$ is an $\amgis$-algebra over a termlike $\sigma$-algebra $\ns U$.
Define the \deffont{$\sigma$-powerset} algebra $\powsigma(\ns P)$
by setting
\begin{itemize*}
\item
$|\powsigma(\ns P)|$ to be those $X{\in}|\nompow(\ns P)|$ (finitely supported subsets of $|\ns P|$; see Subsection~\ref{subsect.finsupp.pow}) with 
\item
the actions $\pi\act X$ and $X[a\sm u]$ from Definition~\ref{defn.sub.sets}, 
\end{itemize*}
satisfying conditions \ref{item.fresh.powsigma} and~\ref{item.alpha.powsigma} below, where $u{\in}|\ns U|$ and $p{\in}|\ns P|$:
\begin{frametxt}
\begin{enumerate*}
\item
\label{item.fresh.powsigma}
$\Forall{u}\New{a}\Forall{p}(p[u\ms a]\in X\liff p\in X)$.
\item
\label{item.alpha.powsigma}
$\Forall{a}\New{b}\Forall{p}(p[b\ms a]\in X\liff (b\ a)\act p\in X)$.
\end{enumerate*}
\end{frametxt}
\end{defn}

Lemma~\ref{lemm.X.sub.fresh.alpha} rephrases 
conditions~\ref{item.fresh.powsigma} and~\ref{item.alpha.powsigma} of Definition~\ref{defn.powsigma}, in a simpler language, albeit one which requires the $\sigma$-action on subsets of an $\amgis$-algebra from Definition~\ref{defn.sub.sets}:
\begin{lemm}
\label{lemm.X.sub.fresh.alpha}
Continuing the notation of Definition~\ref{defn.powsigma}, if $X{\in}|\powsigma(\ns P)|$ then
\begin{enumerate*}
\item
If $a\#X$ then $X[a\sm u]=X$.
\item
If $b\#X$ then $X[a\sm b]=(b\ a)\act X$.
\end{enumerate*}
\end{lemm}
\begin{proof}
\begin{enumerate*}
\item
Suppose $a\#X$.
By part~1 of Lemma~\ref{prop.amgis.iff} $p\in X[a\sm u]$ if and only if $\New{c}p[u\ms c]\in (c\ a)\act X$.
By Corollary~\ref{corr.stuff} $(c\ a)\act X=X$ and
by condition~\ref{item.fresh.powsigma} of Definition~\ref{defn.powsigma} $p[u\ms c]\in X$ if and only if $p\in X$,
so this is if and only if $\New{c}(p\in X)$, that is $p\in X$.
\item
We combine Proposition~\ref{prop.amgis.iff} with condition~\ref{item.alpha.powsigma} of Definition~\ref{defn.powsigma}, since $a\#b$.
\qedhere\end{enumerate*}
\end{proof}

\begin{corr}
\label{corr.amgis.id.sub}
Suppose $X{\in}|\powsigma(\ns P)|$.
Then $X[a\sm a]=X$.
\end{corr}
\begin{proof}
Suppose $b\#X$.
By Lemma~\ref{lemm.sigma.alpha} $X[a\sm a]=((b\ a)\act X)[b\sm a]$.
Note that by Proposition~\ref{prop.pi.supp} $a\#(b\ a)\act X$.
By part~2 of Lemma~\ref{lemm.X.sub.fresh.alpha} $((b\ a)\act X)[b\sm a]=(b\ a)\act((b\ a)\act X)=X$.
\end{proof}

\begin{lemm}
\label{lemm.pow.closed}
If $X{\in}|\powsigma(\ns P)|$ and $u{\in}|\ns U|$ then also $X[a\sm u]{\in}|\powsigma(\ns P)|$.

As a corollary, in Definition~\ref{defn.powsigma}, $|\powsigma(\ns P)|$ is closed under the $\sigma$-action from Definition~\ref{defn.sub.sets}.
\end{lemm}
\begin{proof}
By construction $X[a\sm u]\subseteq|\ns P|$, so we now check the properties listed in Definition~\ref{defn.powsigma}.

By assumption in Definition~\ref{defn.powsigma}, $X$ is finitely supported.
Finite support of $X[a\sm u]$ is from Theorem~\ref{thrm.no.increase.of.supp}.

We check the conditions of Definition~\ref{defn.powsigma} for $X[a\sm u]$:
\begin{enumerate*}
\item
\emph{For fresh $b$ (so $b\#u,X$),\ $X[a\sm u][b\sm v]=X[a\sm u]$.}\quad

We use Lemma~\ref{lemm.sigma.alpha} to assume without loss of generality that $a\#u$.
It suffices to reason as follows:
\begin{tab2}
X[a\sm u][b\sm v]=&X[b\sm v][a\sm u[b\sm v]] &\text{Lemma~\ref{lemm.sub.sub}},\ a\#v
\\
=&X[b\sm v][a\sm u] &\rulefont{\sigma\#},\ b\#u
\\
=&X[a\sm u] &\text{Part~1 of Lemma~\ref{lemm.X.sub.fresh.alpha}},\ b\#X
\end{tab2}
\item
\emph{For fresh $b'$ (so $b'\#u,v,X$) $X[a\sm u][b\sm b']=(b'\ b)\act(X[a\sm u])$.}\quad

It suffices to reason as follows:
\begin{tab2}
X[a\sm u][b\sm b']
=&X[b\sm b'][a\sm u[b\sm b']]
&\text{Lemma~\ref{lemm.sub.sub}},\ a\#b'
\\
=&((b'\ b)\act X)[a\sm (b'\ b)\act u]
&\text{Lemma~\ref{lemm.X.sub.fresh.alpha}},\ b'\#u,X
\\
=&(b'\ b)\act (X[a\sm u])
&\text{Part~2 of Theorem~\ref{thrm.equivar}}
\end{tab2}
\qedhere\end{enumerate*}
\end{proof}

\begin{prop}
\label{prop.pow.sub.algebra}
If $\ns P$ is an $\amgis$-algebra over a termlike $\sigma$-algebra $\ns U$ then $\powsigma(\ns P)$ (Definition~\ref{defn.powsigma}) is a $\sigma$-algebra over $\ns U$.
\end{prop}
\begin{proof}
By Lemma~\ref{lemm.pow.closed} the $\sigma$-action does indeed map to $|\powsigma(\ns P)|$.
By Theorem~\ref{thrm.equivar} so does the permutation action.
It remains to check validity of the axioms from Definition~\ref{defn.sub.algebra}.
\begin{itemize*}
\item
Axiom \rulefont{\sigma id} is Corollary~\ref{corr.amgis.id.sub}.
\item
Axiom \rulefont{\sigma\#} is part~1 of Lemma~\ref{lemm.X.sub.fresh.alpha}.
\item
Axiom \rulefont{\sigma\alpha} is Lemma~\ref{lemm.sigma.alpha}.
\item
Axiom \rulefont{\sigma\sigma} is Lemma~\ref{lemm.sub.sub}.
\qedhere\end{itemize*}
\end{proof}

\begin{xmpl}
\label{xmpl.approx}
Consider some set of terms considered as a termlike $\sigma$-algebra over themselves (Example~\ref{xmpl.fot}; so $s[a\sm t]$ is $s$ with $t$ substituted for $a$).
Write this $\ns{TRM}$.

By Propositions~\ref{prop.pow.sub.algebra} and~\ref{prop.amgis.2} $\powsigma(\powamgis(\ns{TRM}))$ is a $\sigma$-algebra over $\ns{TRM}$.\footnote{$\powamgis$ is the $\amgis$-powerset and is from Definition~\ref{defn.powamgis}.  $\powsigma$ is the \emph{$\sigma$-powerset} and is from~\ref{defn.powsigma}.}
It has Boolean structure given by sets intersection and complement, and by design it has a $\sigma$-algebra structure.

So perhaps $\powsigma(\powamgis(\ns{TRM}))$ is a model of first-order logic with equality; it is a $\sigma$-algebra so we might model quantification an infinite sets intersection.
How to interpret equality is less obvious \dots but we might be lucky.

In the rest of this paper we make this formal, prove it, and put the result in the context of the other models: in nominal posets, using maximally consistent sets (for the completeness result), Tarski-style valuation models, and Herbrand models.
If the reader holds on to the idea that 
\begin{quote}
\emph{to a first approximation this paper is about abstracting, axiomatising, and analysing the behaviour of $\powsigma(\powamgis(\ns{TRM}))$,} 
\end{quote}
then they should not go too far wrong. 
\end{xmpl}

\subsection{Brief interlude: simultaneous sigma- and amgis-actions} 

In Subsection~\ref{subsect.sigma.amgis} we only defined a $\sigma$-action $x[a\sm u]$ for a single atom at a time.
From Subsection~\ref{subsect.eq} and onwards it will be useful to consider a simultaneous $\sigma$-action.
In fact the axioms of Figure~\ref{fig.nom.sigma} give us the power of a simultaneous action, for $\sigma$ (but not $\amgis$, that is not $p[u\ms a]$, as we shall observe).

Suppose $\ns X$ is a $\sigma$-algebra 
over a termlike $\sigma$-algebra $\ns U$.

\begin{defn}
\label{defn.sim.sub}
Suppose $u_1,\dots,u_n{\in}|\ns U|$ and suppose $a_1,\dots,a_n$ are $n$ distinct atoms not in $\bigcup_i\supp(u_i)$.
Then for $x{\in}|\ns X|$ define
$$
\begin{array}{r@{\ }l}
x[a_1\sm u_1,\dots,a_n\sm u_n]=&x[a_1\sm u_1]\dots[a_n\sm u_n] .
\end{array}
$$
\end{defn}

We need to show that Definition~\ref{defn.sim.sub} does not depend on the order in which we take the $u_i$.
This follows using \rulefont{\sigma\sigma} and \rulefont{\amgis\sigma}, and \rulefont{\sigma\#}, because we assumed that each $a_i$ is not in $\bigcup_i\supp(u_i)$.

Now we extend this to the case where the $a_i$ are not necessarily fresh for the $u_i$:
\begin{defn}
\label{defn.sim.sub.alpha}
Suppose $u_1,\dots,u_n{\in}|\ns U|$ and suppose $a_1,\dots,a_n$ are any $n$ distinct atoms.
Then for $x{\in}|\ns X|$ define
$$
\begin{array}{r@{\ }l}
x[a_1\sm u_1,\dots,a_n\sm u_n]=&
(((a_1'\,a_1)\circ\dots\circ(a_n'\,a_n))\act x)[a_1'\sm u_1]\dots[a_n'\sm u_n] .
\end{array}
$$
where we choose $a_1',\dots,a_n'$ fresh (so not in $\{a_1,\dots,a_n\}\cup\supp(x)\cup\bigcup_i\supp(u_i)$).
\end{defn}

\begin{lemm}
\label{lemm.sim.sub.1}
The choice of fresh $a_i'$ and the order of the $u_i$ in Definition~\ref{defn.sim.sub.alpha}, do not matter.
\end{lemm}
\begin{proof}
By routine calculations using \rulefont{\sigma\alpha}, \rulefont{\sigma\sigma}, and \rulefont{\sigma\#}.
\end{proof}

It is natural to try to duplicate Definition~\ref{defn.sim.sub.alpha} and Lemma~\ref{lemm.sim.sub.1} for $\amgis$-algebras.
Suppose $\ns P$ is an $\amgis$-algebra over a termlike $\sigma$-algebra $\ns U$.
Lemma~\ref{lemm.amgis.swap} is a partial dual to Lemma~\ref{lemm.sim.sub.1} and has a simple proof:
\begin{lemm}
\label{lemm.amgis.swap}
Suppose $u,v{\in}|\ns U|$ and $p{\in}|\ns P|$.
Then 
$$
a\#v,b\#u\limp p[u\ms a][v\ms b]=p[v\ms b][u\ms a].
$$
\end{lemm}
\begin{proof}
We reason as follows:
\begin{tab6}
p[v\ms b][u\ms a]
=&
p[u[b\sm v]\ms a][v\ms b]
&\rulefont{\amgis\sigma},\ a\#v
\\
=&
p[u\ms a][v\ms b]
&\rulefont{\sigma\#},\ b\#u
\qedhere\end{tab6} 
\end{proof}
Lemma~\ref{lemm.amgis.swap} asserts that the atoms provided are `sufficiently fresh', then the order of the $\amgis$-action does not matter.
However, we cannot freshen atoms as we can for $\sigma$-action using \rulefont{\sigma\alpha} from Figure~\ref{fig.nom.sigma}---that is, there is no rule \rulefont{\amgis\alpha} of the form \hcancel{$b\#p\limp p[u\ms a]=(b\ a)\act (p[u\ms b])$}.
Lemma~\ref{lemm.amgis.alpha} is the closest we can get to such a result.
We mention it without proof: 
\begin{lemm}
\label{lemm.amgis.alpha}
Suppose $\ns P$ is a finitely-supported $\amgis$-algebra (so every $p{\in}|\ns P|$ has finite support) and $\ns X$ is a $\sigma$-algebra over a termlike $\sigma$-algebra $\ns U$ and suppose $u{\in}|\ns U|$.
Then:
\begin{enumerate*}
\item
If $X{\in}|\powsigma(\ns P)|$ and $b\#X,p,u$ and $p{\in}|\ns P|$ and $a\#p,u$ then $p[u\ms a]\in X$ if and only if $(b\ a)\act (p[u\ms b])\in X$.
\item
If $x{\in}|\ns X|$ and $b\#x$ and $p{\in}|\powamgis(\ns X)|$ then $x\in p[u\ms a]$ if and only if $x\in (b\ a)\act (p[u\ms b])$. 
\end{enumerate*}
\end{lemm}

\section{Nominal posets}
\label{sect.nom.pow}

\subsection{Nominal posets and fresh-finite limits}
\label{subsect.fresh-finite.limit}

\begin{defn}
\label{defn.nom.poset}
A \deffont{nominal poset} is a tuple $\mathcal L=(|\mathcal L|,\act,\leq)$ where
\begin{itemize*}
\item\  
$(|\mathcal L|,\act)$ is a nominal set, and 
\item\  
The relation $\leq\ \subseteq|\mathcal L|{\times}|\mathcal L|$ is an equivariant partial order.\footnote{So $x\leq y$ if and only if $\pi\act x\leq\pi\act y$.}
\end{itemize*}
Call $\mathcal L$ \deffont{finitely fresh-complete} or say it has \deffont{fresh-finite limits} when for every finite subset $X\subseteq|\mathcal L|$ and every finite set of atoms $A\subseteq\mathbb A$ the set of \deffont{$A$-fresh lower bounds}
$$
\{x'{\in}|\mathcal L| \mid A\cap\supp(x')=\varnothing\ \wedge\ \Forall{x{\in} X}x'{\leq} x\}
$$
has a $\leq$-greatest element $\freshwedge{A}X$, which we may call the \deffont{$A\#$limit} ($A$-fresh limit) or \deffont{$A\#$greatest lower bound} of $X$. 

Similarly call $\mathcal L$ \deffont{finitely fresh-cocomplete} or say it has \deffont{fresh-finite colimits} when for every finite $X$ and $A$ as above the set of $A$-fresh upper bounds 
$$
\{x'{\in}|\mathcal L| \mid A\cap\supp(x')=\varnothing\ \wedge\ \Forall{x{\in} X}x{\leq} x'\}
$$
has a $\leq$-least element $\bigvee^{\#A}X$, which we may call its \deffont{$A\#$colimit} or \deffont{$A\#$least upper bound} of $X$.
\end{defn}

\begin{xmpl}
Predicates of first-order logic quotiented by derivable logical equivalence and partially ordered by logical entailment, form a nominal poset; that is $[\phi]$ is the derivable equivalence class of $\phi$ and $[\phi]\leq[\psi]$ when $\phi\cent\psi$.
The permutation action is pointwise on variable symbols.

Then it is a fact that $[\phi\land\psi]$ is a limit for $\{[\phi],[\psi]\}$ and $[\tall a.\phi]$ is an $a\#$limit for $\{[\phi]\}$.
We sketch this particular example in a little more detail in Section~\ref{sect.herbrand}.
\end{xmpl}

\begin{nttn}
\label{nttn.lall}
Suppose $\mathcal L=(|\mathcal L|,\act,\leq)$ is nominal poset.
Suppose $X\subseteq|\mathcal L|$ and $A\subseteq\atoms$ are finite and suppose $x{\in}|\mathcal L|$.
\begin{itemize*}
\item
Call $\freshwedge{\varnothing}X$ a \deffont{limit} or \deffont{greatest lower bound} of $X$.
\item
Call $\freshwedge{\{a\}}\{x\}$ the $a\#$limit or $a\#$greatest lower bound of $x$ and write it $\freshwedge{a}x$.
Unpacking Definition~\ref{defn.nom.poset}, 
$$
\freshwedge{a}x\quad
\text{is the greatest element of}
\quad
\{ x'{\in}|\mathcal L| \mid x'{\leq} x\ \land\ a\#x'\} .
$$ 
\item
Call $\freshvee{\varnothing}X$ a \deffont{colimit} or \deffont{least upper bound} of $X$.
\item
Call $\freshvee{\{a\}}\{x\}$ the $a\#$colimit or $a\#$least upper bound of $x$ and write it $\freshvee{a}x$. 
\item
Write $\ltop$ for the greatest lower bound and $\lbot$ for the least upper bound of the empty set $\varnothing$. 
\item
Write $x\land y$ for the greatest lower bound and $x\lor y$ for the least upper bound of $\{x,y\}$. 
\end{itemize*}
\end{nttn}

\begin{rmrk}
So \emph{$A\#$(co)limits} generalise (co)limits; if we take $A=\varnothing$ then we get a limit (greatest lower bound) just as we are used to. 
\end{rmrk}

There is a convenient factoring of `finitely fresh-complete' into three constituent parts:
\begin{prop}
\label{prop.ffc.char}
Suppose $\mathcal L$ is a nominal poset.
Then $\mathcal L$ is finitely fresh-complete if and only if it has limits of the following three forms:
\begin{itemize*}
\item
$\ltop$, a greatest element (limit for the empty set $\varnothing$).
\item
$x\land y$, a limit for $\{x,y\}$.
\item
$\freshwedge{a}x$, an $a\#$limit for $\{x\}$.
\end{itemize*}
Similarly finitely fresh-cocomplete is equivalent to having $\lbot$, $x\lor y$, and $\freshvee{a}x$.
\end{prop}
\begin{proof}
The interesting part is the right-to-left implication where $X$ is non-empty.
Given finite non-empty $\{x_1,\dots,x_n\}\subseteq|\mathcal L|$ and $\{a_1,\dots,a_n\}\subseteq\mathbb A$, it is not hard to verify that $\freshwedge{a_1}\dots\freshwedge{a_n}(\dots(x_1\land x_2)\dots\land x_n)$ is an $A\#$limit for $X$.
\end{proof}

\begin{defn}
\label{defn.complement}
Suppose $\mathcal L$ is a partial order (it need not be nominal, though we will care most about the case that it is).
Call $x'{\in}|\mathcal L|$ a \deffont{complement} of $x{\in}|\mathcal L|$ when $x\land x'=\lbot$ and $x\lor x'=\ltop$.

If every $x{\in}|\mathcal L|$ has a complement say that $\mathcal L$ \deffont{is complemented}, and write the complement of $x$ as $\lneg x$.
\end{defn}

\begin{lemm}
\label{lemm.comp.unique}
Fresh-finite (co)limits, and complements, are unique if they exist, and $\lneg\lneg x=x$.
\end{lemm}
\begin{proof}
Using the fact that for a partial order, $x\leq y$ and $y\leq x$ imply $x=y$.
\end{proof}

\begin{corr}
\label{corr.supp.freshwedge}
Suppose $\mathcal L=(|\mathcal L|,\act,\leq)$ is a nominal poset.
\begin{enumerate*}
\item
Suppose $X{\subseteq}|\mathcal L|$ is finite and $A{\subseteq}\mathbb A$, and $\freshwedge{A} X$ exists. 

Then $\supp(\freshwedge{A} X){\subseteq} \bigcup\{\supp(x)\mid x{\in} X\}{\setminus} A$.
\item
Suppose $x{\in}|\mathcal L|$, and suppose $\lneg x$ exists.
Then $\supp(\lneg x)=\supp(x)$.
\end{enumerate*}
\end{corr}
\begin{proof}
\begin{enumerate*}
\item
By part~3 of Theorem~\ref{thrm.no.increase.of.supp} $\supp(\freshwedge{A}X)\subseteq\bigcup\{\supp(x)\mid x\in X\}\cup A$.
Since by assumption $A\cap\supp(\freshwedge{A}X)=\varnothing$, the result follows.
\item
By part~3 of Theorem~\ref{thrm.no.increase.of.supp} observing that the map $x\mapsto\lneg x$ is its own inverse.
\qedhere\end{enumerate*}
\end{proof}

Lemma~\ref{lemm.freshwedge.alpha} is $\alpha$-equivalence for fresh-finite limits:
\begin{lemm}
\label{lemm.freshwedge.alpha}
Suppose $\mathcal L$ is a nominal poset and $x{\in}|\mathcal L|$.

If $b\#x$ then $\freshwedge{a}x=\freshwedge{b}(b\ a)\act x$ and $\freshvee{a}x=\freshvee{b}(b\ a)\act x$, where $\freshwedge{a}x$ and $\freshvee{a}x$ exist.
\end{lemm}
\begin{proof}
By Corollary~\ref{corr.supp.freshwedge} $a,b\#\freshwedge{a}x$, so by part~1 of Corollary~\ref{corr.stuff} $\freshwedge{a}x=(b\ a)\act\freshwedge{a}x$.
By part~2 of Theorem~\ref{thrm.equivar} $(b\ a)\act\freshwedge{a}x=\freshwedge{b}(b\ a)\act x$.
The case of $\freshvee{a}x$ is similar.
\end{proof}

We will concentrate on limits from now on; the case of colimits is dual and in the presence of negation can be obtained directly from limits. 
A sequel to this subsection---which we postpoone until we need it later---is in Subsection~\ref{subsect.more.on.limits}.

\subsection{$\sigma$-algebra structure and fresh-finite limits}

\begin{defn}
\label{defn.fresh.continuous}
Suppose a finitely fresh-complete and finitely fresh-cocomplete nominal poset $\mathcal L=(|\mathcal L|,\act,\leq)$ is also a $\sigma$-algebra $(|\mathcal L|,\act,\ns U,\tf{sub}_{\ns U})$ over a termlike $\sigma$-algebra $\ns U$.
Call the $\sigma$-algebra structure \deffont{monotone} when for every $x,y{\in}|\mathcal L|$ and $u{\in}|\ns U|$
$$
x\leq y \text{\quad implies\quad} x[a\sm u]\leq y[a\sm u]
$$
and \deffont{compatible} when for every finite $X\subseteq|\ns L|$ and $A\subseteq\mathbb A$ and $u{\in}|\ns U|$ 
$$
\begin{array}{r@{\ }l@{\ }l}
(\freshwedge{A}X)[a\sm u] =& \freshwedge{A}\{x[a\sm u]\mid x\in X\} &\text{ provided }A\cap (\supp(u)\cup\{a\}) = \varnothing 
\\
(\freshvee{A}X)[a\sm u] =& \freshvee{A}\{x[a\sm u]\mid x\in X\} &\text{ provided }A\cap (\supp(u)\cup\{a\}) = \varnothing 
\\
(\lneg x)[a\sm u]=&\lneg(x[a\sm u]) .
\end{array}
$$ 
\end{defn}

\begin{rmrk}
As standard, negation converts greatest lower bounds to least upper bounds.
Thus we obtain $\freshvee{A}X$ as $\lneg\freshwedge{A}\{\lneg x\mid x\in X\}$.
So the second condition above follows from the first and third.
\end{rmrk}  

\begin{lemm}
\label{lemm.easy}
Given a finitely fresh-complete nominal poset with a compatible $\sigma$-action:
\begin{enumerate*}
\item
$(x\land y)[a\sm u]=x[a\sm u]\land (y[a\sm u])$ and similarly for $\lor$.
\item\label{easy.compatible.monotone}
If $x\leq y$ then $x[a\sm u]\leq y[a\sm u]$ (so a compatible $\sigma$-action is monotone).
\item
If $b\#a,u$ then $(\freshwedge{b}x)[a\sm u]=\freshwedge{b}(x[a\sm u])$ and similarly for $\freshvee{b}$.
\item
$\ltop[a\sm u]=\ltop$ and $\lbot[a\sm u]=\lbot$. 
\end{enumerate*}
\end{lemm}
\begin{proof}
Parts~1, 3, and~4 are special cases of compatibility.\footnote{Part~4 has an alternative one--line proof from Theorem~\ref{thrm.no.increase.of.supp} and \rulefont{\sigma\#}; we use this argument in Lemma~\ref{lemm.sub.commute} for the case of $\lbot$.}
Part~2 follows from part~1 using the fact that $x\leq y$ if and only if $x\land y=x$.
\end{proof}

\begin{lemm}
\label{lemm.fresh.glb.sub}
In a finitely fresh-complete nominal poset:
\begin{itemize*}
\item
If $z\leq x$ and $a\#z$ then $z\leq\freshwedge{a}x$.
\item
If the $\sigma$-action is monotone then
if $z\leq x$ and $a\#z$ then $z\leq x[a\sm u]$ for every $u$, so that $z$ is a lower bound for the (in general infinite) set $\{x[a\sm u]\mid u{\in}|\ns U|\}$.\footnote{The set $\{x[a\sm u]\mid u{\in}|\ns U|\}$ could be finite: either $a\#x$ so that by \rulefont{\sigma\#} $x[a\sm u]=x$ for all $u$; or $|\ns U|$ is finite, so that in particular $\tf{atm}_{\ns U}$ maps every $a$ to some constant element in $|\ns U|$.}
\end{itemize*}
\end{lemm}
\begin{proof}
The first part is direct from the definition of fresh-finite limit; $\freshwedge{a}x$ is by definition a least upper bound for all $z\leq x$ such that $a\#z$.

For the second part, since the $\sigma$-action is monotone $z[a\sm u]\leq x[a\sm u]$ for every $u{\in}|\ns U|$.
By \rulefont{\sigma\#} also $z=z[a\sm u]$.
It follows that $z\leq x[a\sm u]$ for every $u{\in}|\ns U|$.
\end{proof}

\begin{rmrk}
\label{rmrk.odd}
Clearly, we intend the $a\#$limit $\freshwedge{a}x$ to model the universal quantifier $\tall a.\phi$, just as the finite limit $x\land y$ models logical conjunction $\phi\land\psi$.

An odd thing about the quantifier rules is that \rulefont{\tall L} expresses $\Forall{a}\phi$ as a greatest lower bound of an infinite set $\{\phi[a\sm r]\mid \text{all }r\}$ whereas \rulefont{\tall R} expresses $\Forall{a}\phi$ as something else and more finite (these standard derivation rules feature in this paper in Figure~\ref{fig.FOL}).
Lemma~\ref{lemm.fresh.glb.sub} explains this dual nature of $\forall$ as an equality between 
\begin{itemize*}
\item
the limit of a fairly large set of elements $\{x[a\sm u]\mid u{\in}|\ns U|\}$---nice for an elimination rule, so we can eliminate for many $u$---and 
\item
the (fresh-)finite limit of the singleton set $\{x\}$---nice for an introduction rule, giving us but a single proof-obligation.
\end{itemize*} 
See also Proposition~\ref{prop.char.freshwedge} and Subsection~\ref{subsect.powsigma.quant}.
\end{rmrk}

Proposition~\ref{prop.char.freshwedge} proves an equality between the infinite greatest lower bound $\bigwedge_u x[a\sm u]$ and the (fresh-)finite greatest lower bound $\freshwedge{a}x$, provided one of these exists and the $\sigma$-action is monotone. 
That is, the limit of a diagram over infinitely many elements $x[a\sm u]$ is equivalent to the $a\#$limit of just $x$.
\begin{prop}
\label{prop.char.freshwedge}
Suppose $\mathcal L$ is a nominal poset with a monotone $\sigma$-action\footnote{We will only really care about the case that the $\sigma$-action satisfies the stronger property of being \emph{compatible} (Definition~\ref{defn.fresh.continuous}).  
However, compatibility includes conditions involving $\freshwedge{a}$, so we prefer to be more precise and only insist on what we need to make the result work, which is monotonicity.}
and $x{\in}|\mathcal L|$.
Then:
\begin{enumerate*}
\item
If $\freshwedge{a}x$ exists then so does $\bigwedge_{u{\in}|\ns U|} x[a\sm u]$ the limit for $\{x[a\sm u]\mid u{\in}|\ns U|\}$, and they are equal.
In symbols: 
$$
\freshwedge{a}x = \bigwedge_{u{\in}|\ns U|} x[a\sm u].
$$
\item
If $\bigwedge_u x[a\sm u]$ exists then so does $\freshwedge{a}x$, and they are equal.
\end{enumerate*}
\end{prop}
\begin{proof}
Suppose $\freshwedge{a}x$ exists.
\begin{itemize*}
\item
By Lemma~\ref{lemm.fresh.glb.sub} $\freshwedge{a}x$ is a lower bound for $\{x[a\sm u]\mid u{\in}|\ns U|\}$.
\item
Now suppose $z$ is any lower bound for $\{x[a\sm u]\mid u{\in}|\ns U|\}$.
So $z\leq x[a\sm u]$ for every $u{\in}|\ns U|$.
Note that we do not know \emph{a priori} that $a\#z$. 
Choose $b$ fresh (so $b\#z,x$) and take $u=b$.
Then $z\leq x[a\sm b]\stackrel{\text{L\ref{lemm.sub.alpha}}}{=} (b\ a)\act x$.
We assumed that $x$ has an $a\#$limit so by Theorem~\ref{thrm.equivar} also $(b\ a)\act x$ has a $b\#$limit (which by Lemma~\ref{lemm.freshwedge.alpha} is equal to the $a\#$limit of $x$).
Since $b\#z$ it follows that $z\leq\freshwedge{b}(b\ a)\act x\stackrel{\text{L\ref{lemm.freshwedge.alpha}}}{=}\freshwedge{a}x$.
\end{itemize*}
It follows that $\freshwedge{a}x = \bigwedge_u x[a\sm u]$.

Now suppose $\bigwedge_u x[a\sm u]$ exists.
By Lemma~\ref{lemm.a.fresh.bigset} and part~2 of Theorem~\ref{thrm.no.increase.of.supp} we have that $a\#\bigwedge_u x[a\sm u]$. 
\begin{itemize*}
\item
By assumption $\bigwedge_u x[a\sm u]\leq x[a\sm a]\stackrel{\rulefont{\sigma id}}{=} x$. 
Thus $\bigwedge_u x[a\sm u]$ is an $a\#$lower bound for $x$.
\item
Now suppose $z\leq x$ and $a\#z$; we need to show that $z\leq \bigwedge_u x[a\sm u]$.
This is direct from Lemma~\ref{lemm.fresh.glb.sub}.
\end{itemize*}
It follows that $\bigwedge_u x[a\sm u]=\freshwedge{a}x$.
\end{proof}

\begin{rmrk}
Note that Proposition~\ref{prop.char.freshwedge} does \emph{not} mean `finitely fresh-complete\ =\ complete', where \emph{being complete} means having limits for all sets (or rather, all finitely supported sets). 
Proposition~\ref{prop.char.freshwedge} only shows that `finitely fresh-complete' is the same as `complete for sets that can be expressed as a finite union of sets of the form $\{x[a\sm u]\mid u{\in}|\ns U|\}$'.
\end{rmrk}

We briefly mention another characterisation of $\freshwedge{a}x$ using a `smaller' conjunction, though we never use it:
\begin{prop}
\label{prop.char.freshwedge.names}
Suppose $\mathcal L$ is a nominal poset with a monotone $\sigma$-action
and $x{\in}|\mathcal L|$.
Then:
\begin{enumerate*}
\item
If $\freshwedge{a}x$ exists then so does $\bigwedge_{n} x[a\sm n]$ where $n$ ranges over all atoms, and they are equal.
\item
If $\bigwedge_n x[a\sm n]$ exists then so does $\freshwedge{a}x$, and they are equal.\footnote{Strictly speaking we should write $\bigwedge_n x[a\sm \tf{atm}_{\ns U}(n)]$.  See the notation in Definition~\ref{defn.term.sub.alg}.}
\end{enumerate*}
\end{prop}
\begin{proof}
As the proof of Proposition~\ref{prop.char.freshwedge}.
The important point is that by Lemma~\ref{lemm.a.fresh.bigset} and part~2 of Theorem~\ref{thrm.no.increase.of.supp} we have $a\#\bigwedge_n x[a\sm n]$. 
\end{proof}

\subsection{Equality}
\label{subsect.eq}

Suppose $\mathcal L$ is a finitely fresh-complete nominal poset\footnote{In fact we only care about $\land$ here, not $\freshwedge{a}$.} with a compatible $\sigma$-algebra structure over a termlike $\sigma$-algebra $\ns U$.

\begin{defn}
\label{defn.eq}
An \deffont{equality} is an element $(a{=^\lmathcal}b){\in}|\mathcal L|$ with $\supp(a{=^\lmathcal}b)\subseteq\{a,b\}$ and such that:
\begin{enumerate*}
\item
For every $u{\in}|\ns U|$,\ 
$$
(a{=^\lmathcal}b)[a\sm u,b\sm u]=\ltop.
$$
\item
For every $u,v{\in}|\ns U|$ and $z{\in}|\mathcal L|$,\ 
$$
(a{=^\lmathcal}b)[a\sm u,b\sm v]\land z[a\sm u]=(a{=^\lmathcal}b)[a\sm u,b\sm v]\land z[a\sm v].
$$
\end{enumerate*}
\end{defn}

It might interest the reader to look briefly ahead to Definition~\ref{defn.eq.powamgis}, where a very special and useful equality element will be constructed.

\begin{rmrk}
\label{rmrk.short}
The choice of $a$ and $b$ is arbitrary (simultaneous $\sigma$-action $[a\sm u,b\sm v]$ is from Definition~\ref{defn.sim.sub}).

Anticipating the notation of Definition~\ref{defn.extend.f.P}, if we write $(a{=^\lmathcal}b)[a\sm u,b\sm v]$ as $(u{=^\lmathcal}v)$ then properties~1 and~2 of Definition~\ref{defn.eq} can be rewritten in the following more economical form
$$
(u{=^\lmathcal}u)=\ltop
\qquad
(u{=^\lmathcal}v)\land z[a\sm u] = (u{=^\lmathcal}v)\land z[a\sm v]
$$
and this looks like a purely equational rendering of the sequent rules \rulefont{{\teq}R} and \rulefont{{\teq}L} of Figure~\ref{fig.FOL} (sequent rules for equality).
This is deliberate; see condition~\ref{item.eq.cont} of Theorem~\ref{thrm.eq}.
\end{rmrk}

In Definition~\ref{defn.eq} we talk about \emph{an} equality.
However---in the same spirit as Lemma~\ref{lemm.comp.unique}---if one such element exists then it is unique:
\begin{prop}
\label{prop.unique}
An equality in $\mathcal L$ is unique, up to choice of $a$ and $b$, if it exists.
\end{prop}
\begin{proof}
Consider two equalities $a{=^\lmathcal_1}b$ and $a{=^\lmathcal_2}b$.
Take $z=(a{=^\lmathcal_2}b)$ and $u=a$ and $v=b$ in condition~2 of Definition~\ref{defn.eq}.
Then 
using \rulefont{\sigma id},
$$
(a{=^\lmathcal_1}b)\land (a{=^\lmathcal_2}b)=(a{=^\lmathcal_1}b)\land (a{=^\lmathcal_2}b)[b\sm a] .
$$
Now by condition~1 for $(a{=^\lmathcal_2}b)$ we know that $(a{=^\lmathcal_2}b)[b\sm a]=\ltop$.
It follows that $(a{=^\lmathcal_1}b)\leq (a{=^\lmathcal_2}b)$.
By symmetry also $(a{=^\lmathcal_2}b)\leq (a{=^\lmathcal_1}b)$ and we are done.
\end{proof}

It is not hard to similarly prove that equality is reflexive, symmetric, and transitive, if it exists.

\subsection{Definition of a FOLeq algebra}

\begin{defn}
\label{defn.distrib}
Suppose $\mathcal L$ is a finitely fresh-complete and finitely fresh-cocomplete nominal poset (Definition~\ref{defn.nom.poset}).

Call $\mathcal L$ \deffont{distributive} when:
\begin{enumerate*}
\item
$x\lor (y\land z)=(x\lor y)\land(x\lor z)$.
\item
If $a\#x$ then $x\lor(\freshwedge{a}y)=\freshwedge{a}(x\lor y)$.
\end{enumerate*}
These two conditions can be unified into a single condition as follows, for finite $A\subseteq\mathbb A$ and $X\subseteq|\mathcal L|$:
$$
A{\cap}\supp(x){=}\varnothing\ \limp\ x\lor\freshwedge{A}X = \freshwedge{A}\{x{\lor} x'\mid x'{\in} X\}
$$ 
\end{defn}

\begin{rmrk}
Evidently, Definition~\ref{defn.distrib} generalises the usual notion of distributivity.
A dual version of part~1 of Definition~\ref{defn.distrib} is
$x\land (y\lor z)=(x\land y)\lor(x\land z)$ and by a standard argument \cite[Lemma~4.3]{priestley:intlo} the two are equivalent.

A dual version of part~2 of Definition~\ref{defn.distrib} is that if $a\#x$ then $x\land(\freshvee{a}y)=\freshvee{a}(x\land y)$.
This is \emph{not} equivalent to part~2 of Definition~\ref{defn.distrib} in general, but in the presence of complements it is, just by taking negation.
\end{rmrk}

\begin{frametxt}
\begin{defn}
\label{defn.FOLeq}
A \deffont{FOLeq algebra} over a termlike $\sigma$-algebra $\ns U$ is a nominal poset $\mathcal L$ with the following properties:
\begin{itemize*}
\item
$\mathcal L$ is finitely fresh-complete and finitely fresh-cocomplete.
\item
$\mathcal L$ is distributive.
\item
$\mathcal L$ is complemented.
\item
$\mathcal L$ has a compatible $\sigma$-algebra structure over $\ns U$.
\item
$\mathcal L$ has an equality. 
\end{itemize*}
If $\mathcal L$ is as above except that it does not have an equality, then we call $\mathcal L$ a \deffont{FOL algebra}. 
\end{defn}
\end{frametxt}

\begin{rmrk}
For the reader's convenience we break down Definition~\ref{defn.FOLeq} with precise references:
\begin{itemize*}
\item
\emph{Nominal poset} and \emph{finitely fresh-(co)complete} are Definition~\ref{defn.nom.poset}.
\item
\emph{Distributive} is Definition~\ref{defn.distrib}.
\item
\emph{Complements} are Definition~\ref{defn.complement}.
\item
\emph{$\sigma$-algebra structure} is Definition~\ref{defn.sub.algebra} and 
\item
\emph{Compatibility} is Definition~\ref{defn.fresh.continuous}.
\item
\emph{Having equality} is Definition~\ref{defn.eq}.
\end{itemize*}
As the name suggests, the notion of FOLeq algebra is an abstract nominal specification of what it is to be a model of first-order logic with equality. 
The connection between this definition and first-order logic is surely clear: 
\begin{itemize*}
\item
\emph{finitely fresh-(co)complete} gives us $\land$, $\lor$, $\freshwedge{a}$ and $\freshvee{a}$, which are nominal algebraic versions of conjunction, disjunction, and universal and existential quantification; 
\item
\emph{distributivity} ensures that fresh-limits and fresh-colimits interact sensibly---generalising the distributivity we expect of Boolean algebras; 
\item
being complemented gives us negation; 
\item
the compatible $\sigma$-algebra structure gives us a `substitution action'; and 
\item
equality gives us equality.
\end{itemize*}
We do not actually need to insist that a FOLeq algebra be finitely fresh-cocomplete, because we have complements.
However, we still need to insist on distributivity, and this axiom is easier to write down if we assume $\lor$ (no complicated nested negations).
We try to optimise for readability, even where this implies a little redundancy.

However, note that Definition~\ref{defn.FOLeq} is \emph{not} just a direct restatement of the axioms of first-order logic.
This may be hard to see because the reader will see freshness like `$a\#x$' and substitution like `$x[a\sm u]$' and these are properties that, usually, only are applicable to syntax.
But now, the things we are operating on $x$, $a$, and $u$ are abstract and need not be syntax of terms and predicates.
Soon, we shall see models built out of $\amgis$-algebras of Definition~\ref{defn.FOLeq} that are definitely not syntax-based.
It may be worth noting that in previous work we build some more $\sigma$-algebra models, quite differently, out of models of set theory \cite{gabbay:stusun}.

In short, Definition~\ref{defn.FOLeq} is an abstract, not a concrete, definition. 
\end{rmrk}

\begin{rmrk}
Calling a FOLeq algebra an \emph{algebra} might seem misleading.
To this author `algebra' suggests an algebraic (equational) treatment, rather than the poset style of Definition~\ref{defn.FOLeq}.
However, a nominal algebraic rendering of Definition~\ref{defn.FOLeq} is possible and is given in Subsection~\ref{subsect.foleq.alg}.
The interesting axioms are the ones that look like the left- and right-introduction rules for the universal quantifier; see \rulefont{\hnu{\lor}} and \rulefont{\hnu{\leq}} from Figure~\ref{fig.genhnu}. 

Here, we express properties of limits, instead of using nominal algebra axioms as we do in Subsection~\ref{subsect.foleq.alg}.
The content is the same.
\end{rmrk}

\section{Interpretation of first-order logic in a FOLeq algebra}
\label{sect.interp}

\subsection{Syntax and derivability of first-order logic}

\begin{defn}
\label{defn.terms.and.predicates}
A \deffont{signature} is a tuple $(\Sigma,\Pi,\ar)$ of disjoint sets of 
\begin{itemize*}
\item
\deffont{function symbols} $\tf f\in\Sigma$ and 
\item
\deffont{predicate symbols} $\tf P\in\Pi$ 
\end{itemize*}
to each of which is associated an \deffont{arity} $\ar(\tf f),\ar(\tf P)\in\{0,1,2,\dots\}$.

Given a signature, \deffont{terms} and \deffont{predicates} are defined inductively by: 
$$
\begin{array}{r@{\ }l}
r::=& a \mid \tf f(r_1,\dots,r_{\ar(\tf f)})
\\
\phi::=& \tbot \mid r\teq r \mid \tf P(r_1,\dots,r_{\ar(\tf P)}) \mid \phi\tand\phi\mid \tneg\phi \mid \tall a.\phi
\end{array}
$$
Above, $a$ ranges over atoms---atoms play the role of variable symbols in this syntax---and $\tf f$ ranges over elements of $\Sigma$, and $\tf P$ ranges over elements of $\Pi$.
\end{defn}

We take predicates up to $\alpha$-equivalence as standard, and define \deffont{free atoms} $\fa(r),\fa(\phi)$ and \deffont{capture-avoiding substitution} $r[a\sm s]$ and $\phi[a\sm s]$ also as standard.\footnote{The author is using nominal abstract syntax to do this \cite{gabbay:newaas-jv,gabbay:fountl}. A treatment of this with all definitions and a full proof of equivalence with `ordinary' syntax is in \cite[Section~5]{gabbay:fountl}.}

\begin{nttn}
\label{nttn.fol.sugar}
Since our logic is classical, we use the following standard abbreviations:
\begin{itemize*}
\item
Write $\phi\tor\psi$ for $\tneg((\tneg\phi)\tand(\tneg\psi))$.
\item
Write $\phi\timp\psi$ for $(\tneg\phi)\tor\psi$.
\item
Write $\phi\tiff\psi$ for $(\phi\timp\psi)\tand(\psi\timp\phi)$.
\end{itemize*}
\end{nttn}

\begin{nttn}
\label{nttn.phi.psi}
$\Phi$ and $\Psi$ will range over finite sets of predicates. 

We may write $\phi,\Phi$ for $\{\phi\}\cup\Phi$ and we may write $\fa(\Phi)$ for $\bigcup_{\phi{\in}\Phi}\fa(\phi)$.
\end{nttn}

The derivation rules of first-order logic are as standard:
\begin{defn}
\label{defn.sfss}
A \deffont{sequent} is a pair of finite sets of predicates $\Phi\cent\Psi$.
The \deffont{derivable sequents} are defined by the rules in Figure~\ref{fig.FOL}.
\end{defn}

\begin{figure}
$$
\begin{array}{c@{\qquad}c@{\qquad}c}
\begin{prooftree}
\phantom{h}
\justifies
\Phi,\phi\cent\phi,\Psi
\using\rulefont{Hyp}
\end{prooftree}
&
\begin{prooftree}
\phantom{h}
\justifies
\Phi,\tbot\cent\Psi
\using\rulefont{\tbot L}
\end{prooftree}
&
\begin{prooftree}
\Phi,r{\teq}r\cent\Psi
\justifies
\Phi\cent \Psi
\using\rulefont{{\teq}R}
\end{prooftree}
\\[4ex]
\begin{prooftree}
\Phi,\phi_1,\phi_2\cent\Psi
\justifies
\Phi,\phi_1\tand\phi_2\cent\Psi
\using\rulefont{{\tand}L}
\end{prooftree}
&
\begin{prooftree}
\Phi\cent\psi_1,\Psi
\quad
\Phi\cent\psi_2,\Psi
\justifies
\Phi\cent\psi_1\tand\psi_2,\Psi
\using\rulefont{{\tand}R}
\end{prooftree}
&
\begin{prooftree}
\Phi,r'{\teq}r,\phi[a\sm r']\cent \Phi
\justifies
\Phi,r'{\teq}r,\phi[a\sm r]\cent \Phi
\using\rulefont{{\teq}L}
\end{prooftree}
\\[4ex]
\begin{prooftree}
\Phi\cent\psi,\Psi
\justifies
\Phi,\tneg\psi\cent\Psi
\using\rulefont{\tneg L}
\end{prooftree}
&
\begin{prooftree}
\Phi,\phi\cent\Psi
\justifies
\Phi\cent\tneg\phi,\Psi
\using\rulefont{\tneg R}
\end{prooftree}
\\[4ex]
\begin{prooftree}
\Phi,\phi[a\sm r]\cent\Psi
\justifies
\Phi,\tall a.\phi\cent\Psi
\using\rulefont{\tall L}
\end{prooftree}
&
\begin{prooftree}
\Phi\cent\psi,\Psi\ \ (a\not\in\fa(\Phi\cup\Psi))
\justifies
\Phi\cent\tall a.\psi,\Psi
\using\rulefont{\tall R}
\end{prooftree}
\end{array}
$$
\caption{Derivation rules of first-order logic with equality (FOLeq)}
\label{fig.FOL}
\end{figure}

\subsection{Sound interpretation in a FOLeq algebra}

Definition~\ref{defn.multiple.otimes} elaborates and specialises Subsection~\ref{subsect.otimes} (tensor product):
\begin{defn}
\label{defn.multiple.otimes}
Define $\otimes^n\mathbb A$ to be the set of $n$-tuples of distinct atoms with the pointwise action $\pi\act(a_1,\dots,a_n)=(\pi(a_1),\dots,\pi(a_n))$.
\end{defn}
 
\begin{defn}
\label{defn.interp}
Suppose $\mathcal L$ is a FOLeq algebra over $\ns U$ (from Definition~\ref{defn.FOLeq}; so $\mathcal L$ is a finitely fresh-complete complemented nominal poset, with a compatible $\sigma$-algebra structure over $\ns U$, and with equality $(a{=^\lmathcal}b)$). 
We may write $\mathcal L$ also for $(|\mathcal L|,\act)$ the underlying nominal set of the FOLeq algebra.

An \deffont{interpretation} $\interp I$ of a signature $(\Sigma,\Pi,\ar)$ over $\mathcal L$ is an assignment 
\begin{itemize*}
\item
to each $\tf f\in\Sigma$ with arity $n$, an equivariant function $\tf f^{\iden}:\otimes^n\mathbb A\Func \ns U$,\item
to each $\tf P\in\Pi$ with arity $n$, an equivariant function $\tf P^{\iden}:\otimes^n\mathbb A\Func\mathcal L$.
\end{itemize*}
\end{defn}
(It is not hard to check from Subsection~\ref{subsect.eq} that the equality $a{=^\lmathcal}b$ of $\mathcal L$ also gives us an equivariant function $\otimes^2\mathbb A\Func\mathcal L$, so equality could have been presented as one of the $\tf P$ in the signature, but we took it to be primitive to $\mathcal L$ instead.
The treatments of equality and the $\tf P$ flow together in Definition~\ref{defn.extend.f.P}.) 

\begin{rmrk}
By Definition~\ref{defn.equivariant},\ $\tf f^{\iden}$ equivariant means $\tf f^{\iden}(\pi(a_1),\dots,\pi(a_n))=\pi\act\tf f^\iden(a_1,\dots,a_n)$.
Because of this, the information we need to calculate the output of $\tf f^\iden$ for any input, is contained in its output for any one list of $n$ distinct atoms.
Similarly for $\tf P^\iden$.
This can be viewed as a generalisation of the $\new$ quantifier and has a surprisingly attractive general theory; see the \emph{abstractive functions} of \cite{gabbay:genmn}. 

In \cite{gabbay:genmn} we considered generalisations of nominal sets to other permutation groups, whereas in this paper we are specifically interested in the monoidal case of a $\sigma$-action (the simultaneous $\sigma$-actions form a monoid).
Or, to put this more simply: in this paper we can substitute for atoms and not just permute them.
This leads us to Definition~\ref{defn.extend.f.P}, for which we first need a little notation: 
\end{rmrk}

\begin{nttn}
If $\ns U$ is a nominal set write $\ns U^n$ for $\overbrace{\ns U\times\dots\times\ns U}^{\text{$n$ times}}$.
\end{nttn}

\begin{defn}
\label{defn.extend.f.P}
Extend $\tf f^\iden$ and $\tf P^\iden$ and $=^\lmathcal$ to functions $\tf f^\iden:\ns U^n\Func\ns U$ and $\tf P^\iden:\ns U^n\Func\mathcal L$ and $=^\iden:\ns U^2\Func\mathcal L$ as follows
$$
\begin{array}{r@{\ }l}
\tf f^\iden(u_1,\dots,u_n)=&\tf f^\iden(a_1,\dots,a_n)[a_1\sm u_1,\dots,a_n\sm u_n]
\\
\tf P^\iden(u_1,\dots,u_n)=&\tf P^\iden(a_1,\dots,a_n)[a_1\sm u_1,\dots,a_n\sm u_n]
\\
{=^\iden}(u_1,u_2)=&(a{=^\lmathcal}b)[a\sm u_1,b\sm u_2]
\end{array}
$$ 
where we choose $a_1,\dots,a_n$ to be $n$ distinct atoms fresh for $\bigcup_1^n\supp(u_i)$.
We may write ${=^\iden}(u_1,u_2)$ infix as $u_1{=^\iden}u_2$.
\end{defn}

\begin{lemm}
\label{lemm.sim.sigma}
Definition~\ref{defn.extend.f.P} is independent of the choice of fresh atoms $a_1,\dots,a_n$ and $a,b$.
\end{lemm}
\begin{proof}
A proof by direct calculations is not hard, or we can see this directly from \cite[Theorem~27]{gabbay:genmn} since by Lemma~\ref{lemm.fresh.sub} $\tf f^\iden$ and $\tf P^\iden$ are \emph{purely abstractive} (terminology of Theorem~27).
\end{proof}

\begin{defn}
\label{defn.interp.I}
Extend the interpretation $\interp I$ of Definition~\ref{defn.extend.f.P} to terms and predicates as shown in Figure~\ref{fig.extend.interp}.
\end{defn}
Much as for $\tf f^\iden$ and $\tf P^\iden$ above, we technically need to check that Definition~\ref{defn.interp.I} does not depend on the choice of bound atom $a$ in $\tall a.\phi$.
This follows by easy calculations using the fact that by assumption, $a\not\in\supp(\freshwedge{a}\model{\phi})$. 

\begin{figure}
$$
\begin{array}{r@{\ }l@{\qquad}r@{\ }l}
\model{a}=&\tf{atm}_{\ns U}(a)
&
\model{\tbot}=&\lbot
\\
\model{\tf f(r_1,\dots,r_n)}=&\tf f^\iden(\model{r_1},\dots,\model{r_n})
&
\model{r{\teq}s}=&\model{r}{=^\iden}\model{s}
\\
&&
\model{\tf P(r_1,\dots,r_n)}=&\tf P^\iden(\model{r_1},\dots,\model{r_n})
\\
&&
\model{\phi\tand\psi}=&\model{\phi}\land\model{\psi}
\\
&&
\model{\tneg\phi}=&\lneg\model{\phi}
\\
&&
\model{\tall a.\phi}=&\freshwedge{a}\model{\phi}
\end{array}
$$
\caption{The interpretation of terms and predicates}
\label{fig.extend.interp}
\end{figure}

\begin{lemm}
\label{lemm.supp.lemma}
$\supp(\model{r})\subseteq\fa(r)$ and $\supp(\model{\phi})\subseteq\fa(\phi)$.
\end{lemm}
\begin{proof}
From Theorem~\ref{thrm.no.increase.of.supp} using the fact that the support of the abstract syntax of $r$ is equal to $\fa(r)$ and the support of the (nominal, since we have a binder $\tall$) abstract syntax of $\phi$ is equal to $\fa(\phi)$ (see \cite[Section~5]{gabbay:fountl} for a detailed treatment of the nominal abstract syntax of the untyped $\lambda$-calculus). 
\end{proof}

\begin{lemm}
$\model{\ttop}=\ltop{\in}|\mathcal L|$.
\end{lemm}
\begin{proof}
Unpacking Definition~\ref{defn.interp.I} and recalling that $\ttop$ is sugar for $\tneg\tbot$.
\end{proof}

\begin{lemm}
\label{lemm.f.P.sub}
\begin{enumerate*}
\item
$\tf f^\iden(u_1,\dots,u_n)[a\sm u]=\tf f^\iden(u_1[a\sm u],\dots,u_n[a\sm u])$.
\item
$(u_1{=^\iden}u_2)[a\sm u]=(u_1[a\sm u]{=^\iden}u_2[a\sm u])$.
\item
$\tf P^\iden(u_1,\dots,u_n)[a\sm u]=\tf P^\iden(u_1[a\sm u],\dots,u_n[a\sm u])$.
\end{enumerate*}
\end{lemm}
\begin{proof}
By definition $\tf f^\iden(u_1,\dots,u_n)[a\sm u]$ is equal to $\tf f^\iden(a_1,\dots,a_n)[a_1\sm u_1]\dots[a_n\sm u_n][a\sm u]$ where we are free to choose $a_1,\dots,a_n$ fresh.
We use Lemma~\ref{lemm.sub.sub} to rewrite this as
$$
\tf f^\iden(a_1,\dots,a_n)[a\sm u][a_1\sm u_1[a\sm u]]\dots[a_n\sm u_n[a\sm u]] .
$$ 
By Theorem~\ref{thrm.no.increase.of.supp} $a\#\tf f^\iden(a_1,\dots,a_n)$ so by \rulefont{\sigma\#} we can garbage--collect the innermost $[a\sm u]$.
The result follows.

The cases of $=^\iden$ and $\tf P^\iden$ are similar.
\end{proof}

\begin{lemm}
\label{lemm.sub.commute}
\begin{itemize*}
\item
$\model{t[a\sm r]}=\model{t}[a\sm\model{r}]$.
\item
$\model{\phi[a\sm r]}=\model{\phi}[a\sm \model{r}]$.
\end{itemize*}
\end{lemm}
\begin{proof}
By a routine induction on $t$ and $\phi$.
\begin{itemize*}
\item
\emph{The case of $a$.}\quad
$a[a\sm r]=r$ and by \rulefont{\sigma a} $\tf{atm}_{\ns U}(a)[a\sm \model{r}]=\model{r}$.
\item
\emph{The cases of $\tf f(r_1,\dots,r_n)$, $\tf P(r_1,\dots,r_n)$, and $r_1\teq r_2$.}\quad
This is Lemma~\ref{lemm.f.P.sub}.
\item
\emph{The case of $\tbot$.}\quad
By Theorem~\ref{thrm.no.increase.of.supp} $\supp(\tbot)=\varnothing$.
By \rulefont{\sigma\#} $\tbot[a\sm x]=\tbot$.
\item
\emph{The cases of $\phi\tand\psi$ and $\tall a.\phi$.}\quad
From parts~1 and~3 of Lemma~\ref{lemm.easy}.
\item
\emph{The case of $\tneg\phi$.}\quad
Direct from our assumption that the $\sigma$-algebra structure is compatible (Definition~\ref{defn.fresh.continuous}).
\end{itemize*}
\end{proof}

\begin{defn}
Extend the interpretation $\interp I$ of Definition~\ref{defn.interp.I} to sequents by defining
$$
\model{\Phi\cent\Psi} \quad\text{is the assertion}\quad \bigwedge \{\model{\phi} \mid \phi{\in}\Phi\} \leq \bigvee \{\model{\psi} \mid \psi\in\Psi\} .
$$
\end{defn}

\begin{frametxt}
\begin{thrm}[Soundness]
\label{thrm.fol.sound}
If $\Phi\cent\Psi$ is derivable then $\model{\Phi\cent\Psi}$ is true.
\end{thrm}
\end{frametxt}
\begin{proof}
It suffices to check validity of the rules in Figure~\ref{fig.FOL}.
\begin{itemize*}
\item
\rulefont{\tbot L} is valid because $\model{\tbot}=\lbot$ in Definition~\ref{defn.interp.I}.
\item
\rulefont{{\teq}R} is valid from condition~1 of Definition~\ref{defn.eq}.
\item
\rulefont{{\teq}L} is valid from condition~2 of Definition~\ref{defn.eq}.
\item
\rulefont{{\tand}L}, \rulefont{{\tand}R}, \rulefont{\tneg L}, and \rulefont{\tneg R} are by standard facts of Boolean algebras (bounded distributive complemented lattices). 
\item
\rulefont{\tall R} and \rulefont{\tall L} are from Lemma~\ref{lemm.fresh.glb.sub} and Lemma~\ref{lemm.supp.lemma}.
\qedhere\end{itemize*} 
\end{proof}

\section{The $\sigma$-powerset as a FOLeq algebra}
\label{sect.sigma.foleq}

Suppose $\ns P$ is an $\amgis$-algebra over a termlike $\sigma$-algebra $\ns U$.
We know by Proposition~\ref{prop.pow.sub.algebra} that $\powsigma(\ns P)$ from Definition~\ref{defn.powsigma} is a $\sigma$-algebra.
It is also very easy to see that $\powsigma(\ns P)$ with the pointwise action $\pi\act X=\{\pi\act p\mid p\in X\}$ (Definition~\ref{defn.sub.sets}), ordered by subset inclusion $X\subseteq Y$, is a nominal poset (Definition~\ref{defn.nom.poset}).

We want now to show that $\powsigma(\ns P)$ is a FOLeq algebra (Theorem~\ref{thrm.powsigma.FOLeq}).
That is, we want to show that $\powsigma(\ns P)$
\begin{itemize*}
\item
is finitely fresh-complete (has $\freshwedge{A}X$),
\item
is complemented ($|\ns P|\setminus X$), 
\item
has an equality, 
\item
and the $\sigma$-action is compatible (Definition~\ref{defn.fresh.continuous}).
\end{itemize*} 
It will be obvious that complement is given by sets complement.
Equality is non-evident but very natural and is discussed in Subsection~\ref{subsect.powsigma.equality}.
For $\freshwedge{A}X$ we use Proposition~\ref{prop.ffc.char} to split the problem into parts.
We consider this next.

\subsection{Intersection, complement, top element}

Suppose $\ns P=(|\ns P|,\act,\ns U,\tf{amgis}_{\ns P})$ is an $\amgis$-algebra over a termlike $\sigma$-algebra $\ns U$. 
Recall the definitions of $\nompow(\ns P)$ from Subsection~\ref{subsect.finsupp.pow} (finitely-supported powerset)
and of $\powsigma(\ns P)$ from Definition~\ref{defn.powsigma} (the $\sigma$-powerset).

Lemmas~\ref{lemm.sub.bigcap} and~\ref{lemm.pow.nu.closed} are perhaps not entirely obvious, but all the cases follow the same pattern of pointwise calculations on sets.
These cover the technically simplest cases; in Subsection~\ref{subsect.powsigma.quant} we move on to quantification.

\begin{lemm}
\label{lemm.sub.bigcap}
Suppose $\ns P$ is an $\amgis$-algebra over a termlike $\sigma$-algebra $\ns U$.
Suppose $u{\in}|\ns U|$ and $Y,Y'{\in}|\nompow(\ns P)|$ and suppose 
$\mathcal X\subseteq|\nompow(\ns P)|$ is strictly finitely supported (Definition~\ref{defn.strictpow}).
Then:
\begin{enumerate*}
\item
$(|\ns P|{\setminus} Y)[a\sm u]=|\ns P|\setminus(Y[a\sm u])$.

In words: $\sigma$ commutes with sets complement.
\item\label{sub.bigcap.sfs}
$(\bigcap_{X{\in}\mathcal X} X)[a\sm u]=\bigcap_{X{\in}\mathcal X}(X[a\sm u])$
and
$(\bigcup_{X{\in}\mathcal X} X)[a\sm u]=\bigcup_{X{\in}\mathcal X}(X[a\sm u])$.

In words: $\sigma$ commutes with strictly finitely supported sets intersections and unions.

In particular, $\sigma$ commutes with \emph{finite} intersection and union.
\item
If $Y\subseteq Y'$ then $Y[a\sm u]\subseteq Y'[a\sm u]$.

In words: $\sigma$ is monotone.
\end{enumerate*}
\end{lemm}
\begin{proof}
We reason as follows:
$$
\begin{array}{r@{\ }l@{\qquad}l}
p\in (|\ns P|{\setminus}X)[a\sm u]
\liff&
\New{c}p[u\ms c]\in (c\ a)\act(|\ns P|{\setminus}X)
&\text{Proposition~\ref{prop.amgis.iff}}
\\
\liff&
\New{c}p[u\ms c]\in |\ns P|{\setminus}(c\ a)\act X
&\text{Theorem~\ref{thrm.equivar}}
\\
\liff&
\New{c}p[u\ms c]\not\in (c\ a)\act X
&\text{Fact}
\\
\liff&
p\not\in X[a\sm u]
&\text{Proposition~\ref{prop.amgis.iff}}
\\
\liff&
p\in |\ns P|{\setminus}(X[a\sm u])
&\text{Fact}
\\[1ex]
p\in (\bigcap_{X{\in}\mathcal X} X)[a\sm u]
\liff&
\New{c}p[u\ms c]\in \bigcap_{X{\in}\mathcal X}(c\ a)\act X
&\text{Prop~\ref{prop.amgis.iff}, Thm~\ref{thrm.equivar}}
\\
\liff&
\New{c}\Forall{X{\in}\mathcal X}p[u\ms c]\in (c\ a)\act X
&\text{Fact}
\\[1.5ex]
p\in \bigcap_{X{\in}\mathcal X}(X[a\sm u])
\liff&
\Forall{X{\in}\mathcal X}p\in X[a\sm u]
&\text{Fact}
\\
\liff&
\Forall{X{\in}\mathcal X}\New{c}p[u\ms c]\in (c\ a)\act X
&\text{Proposition~\ref{prop.amgis.iff}}
\end{array}
$$
We note that by Lemma~\ref{lemm.strict.support},\ $c\#\mathcal X$ if and only if $c\#X$ for every $X\in\mathcal X$.
This allows us to swap the $\forall$ and the $\new$ quantifier, and the result follows.

Part~3 follows using part~2 and the fact that $Y[a\sm u]\subseteq Y'[a\sm u]$ if and only if $Y[a\sm u]\cap Y'[a\sm u]=Y[a\sm u]$.
\end{proof}

\begin{lemm}
\label{lemm.pow.nu.closed}
\begin{itemize*}
\item
If ${X,Y{\in}|\nompow(\ns P)|}$ then
$X\cap Y{\in}|\nompow(\ns P)|$
and \\
if ${X,Y{\in}|\powsigma(\ns P)|}$ then
${X\cap Y{\in}|\powsigma(\ns P)|}$.
\item
If ${X{\in}|\nompow(\ns P)|}$ then $|\ns P|{\setminus}X{\in}|\nompow(\ns P)|$, and 
\\
if ${X{\in}|\powsigma(\ns P)|}$ then $|\ns P|{\setminus}X{\in}|\powsigma(\ns P)|$. 
\item
$|\ns P|{\in}|\nompow(\ns P)|$ and 
$|\ns P|{\in}|\powsigma(\ns P)|$.
\end{itemize*}
\end{lemm}
\begin{proof}
Finite support follows from Theorem~\ref{thrm.no.increase.of.supp}.
All the results involving finitely-supported powerset $\nompow(\text{-})$ follow immediately.

It remains to check conditions~\ref{item.fresh.powsigma} and~\ref{item.alpha.powsigma} of Definition~\ref{defn.powsigma} for the $\sigma$-powerset $\powsigma(\text{-})$.
We do this just for $X\cap Y$:
\begin{enumerate*}
\item
\emph{If $a$ is fresh (so $a\#X,Y,u$) then $(X\cap Y)[a\sm u]=X\cap Y$.}
\\ We reason as follows:
\begin{tab3}
(X\cap Y)[a{\sm}u]=&(X[a{\sm}b])\cap(Y[a{\sm}u])
&\text{Lemma~\ref{lemm.sub.bigcap}}
\\
=&X\cap Y
&\text{Cond~\ref{item.fresh.powsigma} of Def~\ref{defn.powsigma}}
\end{tab3}
\item
\emph{If $b$ is fresh (so $b\#X,Y$) then $(X\cap Y)[a\sm b]=(b\ a)\act (X\cap Y)$.}
\\ We reason as follows:
\begin{tab3}
(X\cap Y)[a{\sm}b]=&(X[a{\sm}b])\cap(Y[a{\sm}b])
&\text{Lemma~\ref{lemm.sub.bigcap}}
\\
=&((b\ a)\act X)\cap((b\ a)\act Y)
&\text{Cond~\ref{item.alpha.powsigma} of Def~\ref{defn.powsigma}}
\\
=&(b\ a)\act(X\cap Y)
&\text{Theorem~\ref{thrm.equivar}}
\qedhere\end{tab3}
\end{enumerate*}
\end{proof}

\begin{prop}
\label{prop.powsigma.bool}
$\varnothing$ is the bottom element and $|\ns P|$ is the top element in $\nompow(\ns P)$ and $\powsigma(\ns P)$ ordered by subset inclusion.
$X\cap Y$ is the greatest lower bound for $\{X,Y\}$.
$|\ns P|{\setminus}X$ is the complement of $X$.
Thus $\nompow(\ns P)$ and $\powsigma(\ns P)$ ordered by subset inclusion form Boolean algebras, and by construction $\powsigma(\ns P)$ is a subalgebra of $\nompow(\ns P)$.
\end{prop}
\begin{proof}
By standard sets calculations.
\end{proof}

\subsection{Quantification}
\label{subsect.quantification}
\label{subsect.powsigma.quant}

In this subsection we explore what concrete sets operation corresponds to the fresh-finite limits $\freshwedge{a}$ of Notation~\ref{nttn.lall}; in other words we explore quantification.

Suppose $\ns P=(|\ns P|,\act,\ns U,\tf{amgis}_{\ns P})$ is an $\amgis$-algebra over a termlike $\sigma$-algebra $\ns U$. 
Recall the definitions of $\nompow(\ns P)$ from Subsection~\ref{subsect.finsupp.pow} (finitely-supported powerset)
and of $\powsigma(\ns P)$ from Definition~\ref{defn.powsigma} (the $\sigma$-powerset).

\begin{defn}
\label{defn.nu.U}
If $X\in|\nompow(\ns P)|$ then define
\begin{frameqn}
\freshcap{a}X=
\bigcap\{ X[a\sm u]\mid u{\in}|\ns U|\} .
\end{frameqn}
\end{defn}

We work towards proving $\freshcap{a}X\in|\nompow(\ns P)|$; this is Theorem~\ref{thrm.all.closed}.
We need the key technical Lemma~\ref{lemm.technical} for Proposition~\ref{prop.all.sub.commute}:

\begin{lemm}
\label{lemm.technical}
Suppose $p\in|\ns P|$ and $X{\in}\nompow(\ns P)$ and $v{\in}|\ns U|$ and suppose $a\#v$.
Then
$$
\begin{array}{l}
\New{b'}\Forall{u{\in}|\ns U|}\New{a'} p[v\ms b'][u\ms a']\in (b'\,b)\act(a'\,a)\act X
\quad\text{if and only if}
\\
\Forall{u{\in}|\ns U|}\New{b'}\New{a'} p[u\ms a'][v\ms b']\in (b'\,b)\act(a'\,a)\act X .
\end{array}
$$
\end{lemm}
\begin{proof}
We prove two implications:
\begin{itemize}
\item
\emph{The up-down implication.}
Assume $\New{b'}\Forall{u{\in}|\ns U|}\New{a'} p[v\ms b'][u\ms a']\in (b'\,b)\act(a'\,a)\act X$.

Choose $u{\in}|\ns U|$.
Choose fresh $b'$ and $a'$ (so $b',a'\#X,v,u$).
Then by assumption (since $b'\#X,v$ and $a'\#X,v,u$) $p[v\ms b'][u\ms a']\in (b'\,b)\act(a'\,a)\act X$
so that by \rulefont{\amgis\sigma} of Figure~\ref{fig.amgis} (since $a'\#v$) $p[u[b'\sm v]\ms a'][v\ms b']\in (b'\,b)\act(a'\,a)\act X$.
Now by \rulefont{\sigma\#} $u[b'\sm v]=u$ (since $b'\#u$).
Therefore $p[u\ms a'][v\ms b']\in (b'\,b)\act(a'\,a)\act X$. 
\item
\emph{The down-up implication.}
Assume $\Forall{u{\in}|\ns U|}\New{b'}\New{a'} p[u\ms a'][v\ms b']\in (b'\,b)\act(a'\,a)\act X$.

Choose fresh $b'$ (so $b'\#X,v$), choose $u{\in}|\ns U|$ (for which $b$ need not necessarily be fresh), and choose fresh $a'$ (so $a'\#X,v,u$).
By Lemma~\ref{lemm.fresh.sub} $b'\#u[b'\sm v]$ (since $b'\#v$).
Therefore $p[u[b'\sm v]\ms a'][v\ms b']\in (b'\,b)\act(a'\,a)\act X$ and by \rulefont{\amgis\sigma} of Figure~\ref{fig.amgis} (since $a'\#v$) $p[v\ms b'][u\ms a']\in (b'\,b)\act(a'\,a)\act X$.
\qedhere
\end{itemize}
\end{proof}

We cannot use Lemma~\ref{lemm.sub.bigcap}(\ref{sub.bigcap.sfs}) to derive Proposition~\ref{prop.all.sub.commute} because $\{X[a\sm u]\mid u{\in}|\ns U|\}$ is not necessarily strictly finitely supported.
The result still holds, by a proof using Lemma~\ref{lemm.technical}:
\begin{prop}
\label{prop.all.sub.commute}
Suppose $X{\in}|\nompow(\ns P)|$ and $v{\in}|\ns U|$ and $a\#v$.
Then
$$
(\freshcap{a}X)[b\sm v]=\freshcap{a}(X[b\sm v]) .
$$
\end{prop}
\begin{proof}
Consider $p{\in}|\ns P|$.
We reason as follows:
$$
\begin{array}[b]{@{\hspace{-0em}}r@{\ }l@{\quad}l}
p\in (\freshcap{a}&X)[b\sm v]
\\
\liff&
\New{b'}p[v\ms b']\in \freshcap{a}(b'\,b)\act X
&\text{Prop~\ref{prop.amgis.iff},\ Thrm~\ref{thrm.equivar}}
\\
\liff&
\New{b'}p[v\ms b']\in \bigcap_{u{\in}|\ns U|} ((b'\ b)\act X)[a\sm u]
&\text{Definition~\ref{defn.nu.U}}
\\
\liff&
\New{b'}\Forall{u{\in}|\ns U|}\New{a'} p[v\ms b'][u\ms a']\in (b'\,b)\act(a'\,a)\act X
&\text{Proposition~\ref{prop.amgis.iff}}
\\
\liff&
\Forall{u{\in}|\ns U|}\New{b'}\New{a'} p[u\ms a'][v\ms b']\in (b'\,b)\act(a'\,a)\act X
&\text{Lemma~\ref{lemm.technical}}
\\
\liff&
\Forall{u{\in}|\ns U|}\New{a'} p[u\ms a']\in ((a'\,a)\act X)[b\sm v]
&\text{Proposition~\ref{prop.amgis.iff}}
\\
\liff&
\Forall{u{\in}|\ns U|}\New{a'} p[u\ms a']\in (a'\,a)\act (X[b\sm v])
&\text{T\ref{thrm.equivar},\ C\ref{corr.stuff}},\,a',a\#v
\\
\liff&
\Forall{u{\in}|\ns U|} p\in X[b\sm v][a\sm u]
&\text{Proposition~\ref{prop.amgis.iff}}
\\
\liff&
p\in \freshcap{a}(X[b\sm v])
&\text{Definition~\ref{defn.nu.U}}
\end{array}
\qedhere$$
\end{proof}

\begin{lemm}
\label{lemm.all.alpha}
Suppose $X\in|\nompow(\ns P)|$. 
Then 
$$
b\#X
\quad\text{implies}\quad
\freshcap{a}X=\freshcap{b}(b\ a)\act X .
$$
As a corollary, 
$a\#\freshcap{a}X$ and 
$\supp(\freshcap{a}X)\subseteq\supp(X){\setminus}\{a\}$.
\end{lemm}
\begin{proof}
The corollary follows by part~3 of Corollary~\ref{corr.stuff} and by Theorem~\ref{thrm.no.increase.of.supp}.
For the first part, we reason as follows:
\begin{tab2}
\freshcap{a}X=&\bigcap\{X[a{\sm}u]\mid u{\in}|\ns U|\} 
&\text{Definition~\ref{defn.nu.U}}
\\
=&\bigcap\{((b\ a)\act X)[b{\sm}u]\mid u{\in}|\ns U|\}
&\text{Lemma~\ref{lemm.sigma.alpha}}
\\
=&\freshcap{b}(b\ a)\act X
&\text{Definition~\ref{defn.nu.U}}
\qedhere\end{tab2}
\end{proof}

Recall $\powsigma(\ns P)$ from Definition~\ref{defn.powsigma}:
\begin{thrm}
\label{thrm.all.closed}
If $X{\in}|\powsigma(\ns P)|$ then $\freshcap{a}X{\in}|\powsigma(\ns P)|$.
\end{thrm}
\begin{proof}
Finite support of $\freshcap{a}X$ is by Theorem~\ref{thrm.no.increase.of.supp}. 
It remains to check conditions~\ref{item.fresh.powsigma} and~\ref{item.alpha.powsigma} of Definition~\ref{defn.powsigma}:
\begin{enumerate*}
\item
Suppose $b$ is fresh (so $b\#X$) and suppose $v{\in}|\ns U|$.
Using Lemma~\ref{lemm.all.alpha} suppose without loss of generality that $a\#v$.
Then we reason as follows:
\begin{tab3}
(\freshcap{a}X)[b{\sm}v]=&\freshcap{a}(X[b{\sm}v])
&\text{Proposition~\ref{prop.all.sub.commute}}
\\
=&\freshcap{a}X
&\text{Cond~\ref{item.fresh.powsigma} of Def~\ref{defn.powsigma}},\ b\#X
\end{tab3}
\item
Suppose $b'$ is fresh (so $b'\#X$).
Then we reason as follows:
\begin{tab3}
(\freshcap{a}X)[b{\sm}b']=&\freshcap{a}(X[b{\sm}b'])
&\text{Proposition~\ref{prop.all.sub.commute}}
\\
=&\freshcap{a}((b'\ b)\act X)
&\text{Cond~\ref{item.alpha.powsigma} of Def~\ref{defn.powsigma}},\ b'\#X
\\
=&(b'\ b)\act (\freshcap{a}X)
&\text{Theorem~\ref{thrm.equivar}}
\qedhere\end{tab3}
\end{enumerate*}
\end{proof}

\begin{corr}
\label{corr.freshcap.freshwedge}
If $X{\in}|\powsigma(\ns P)|$ then $\freshcap{a}X$ from Definition~\ref{defn.nu.U} is equal to $\freshwedge{a}X$ the fresh-finite limit in the lattice $\powsigma(\ns P)$ ordered by subset inclusion.
\end{corr}
\begin{proof}
By Theorem~\ref{thrm.all.closed} $\freshcap{a}X\in|\powsigma(\ns P)|$.
We use Proposition~\ref{prop.char.freshwedge}. 
\end{proof}

\begin{rmrk}
The proofs leading up to Corollary~\ref{corr.freshcap.freshwedge} are somewhat subtle, and it is useful to illustrate how they work by tracing through how a similar `proof' fails: 

By Proposition~\ref{prop.char.freshwedge.names} $\freshwedge{a}X=\bigwedge_{n{\in}\mathbb A} X[a\sm n]$.
However it does \emph{not} follow from this that $\freshcap{a}X=\bigcap_{n{\in}\mathbb A}X[a\sm n]$.
This is because Theorem~\ref{thrm.all.closed} does not hold of $\bigcap_{n{\in}\mathbb A}X[a\sm n]$---it is not in $\powsigma(\ns P)$---and this is because Proposition~\ref{prop.all.sub.commute} fails, which is because Lemma~\ref{lemm.technical} fails.

That is, if we replace every $\forall u{\in}|\ns U|$ in Lemma~\ref{lemm.technical} with $\forall n{\in}\mathbb A$, then Lemma~\ref{lemm.technical} no longer works (and this is because $\mathbb A$ is not closed under applying substitutions $[a\sm u]$, whereas $\ns U$ is).
\end{rmrk}

\subsection{Equality in the $\sigma$-powerset}
\label{subsect.powsigma.equality}

We have seen how to interpret conjunction, union, negation, truth, false, and quantification in the $\sigma$-powerset of an $\amgis$-algebra.
Now we show that we can also interpret equality.

Suppose $\ns P=(|\ns P|,\act,\ns U,\tf{amgis}_{\ns P})$ is an exact $\amgis$-algebra (Definition~\ref{defn.exact.amgis.algebra}) over a termlike $\sigma$-algebra $\ns U$. 

\begin{defn}
\label{defn.eq.powamgis}
If $u,v{\in}|\ns U|$ then define $(u{=^{\ns P}}v)\subseteq|\ns P|$ by
\begin{frameqn}
(u{=^{\ns P}}v)\ =\ \{p{\in}|\ns P| \mid \New{c}(p[u\ms c]=p[v\ms c])\} .
\end{frameqn}
\end{defn}

\begin{rmrk}
\label{rmrk.leibnitz}
Some motivation for Definition~\ref{defn.eq.powamgis}:
Suppose $\ns X$ is a $\sigma$-algebra over $\ns U$ (Definition~\ref{defn.sub.algebra}) and suppose just for this remark that $\ns P=\powamgis(\ns X)$ (Definition~\ref{defn.powamgis}).
Recall from Proposition~\ref{prop.powamgis.exact} that $\ns P$ is indeed exact.
 
Suppose $p\in(u{=^{\ns P}}v)$, meaning by Definition~\ref{defn.eq.powamgis} that $\New{c}(p[u\ms c]=p[v\ms c])$.
Using Proposition~\ref{prop.sigma.iff} $p[u\ms c]=p[v\ms c]$ when $x[c\sm u]\in p\liff x[c\sm v]\in p$ for every $x{\in}|\ns X|$.
Thus
$$
p\in(u{=^{\ns P}}v)\quad \liff\quad \New{c}\Forall{x{\in}\ns X}(x[c\sm u]\in p\liff x[c\sm v]\in p).
$$
Intuitively this means $p$ cannot discern any difference between $u$ and $v$. 
This is in the spirit of \emph{Leibnitz equality}; two things that equal when they are indiscernible. 

More discussion is in Appendix~\ref{sect.more.equality}.
\end{rmrk}

\begin{rmrk}
It is not necessarily the case that $\supp(u){\cup}\supp(v)\subseteq\supp(u{=^{\ns P}}v)$.
For instance, if $\ns P$ has the \emph{trivial} $\amgis$-action that $p[u\ms a]=p$ for all $u$ and $a$ (for instance, take $|\ns P|=\{\ast\}$ the one-element set) then $(a{=^{\ns P}}b)=|\ns P|$, which has support $\varnothing$.
\end{rmrk}

Lemma~\ref{lemm.eq.sanity.check} is where we really use the exactness property of $\ns P$ (Definition~\ref{defn.exact.amgis.algebra}):
\begin{lemm}
\label{lemm.eq.sanity.check}
$(u_1{=^{\ns P}}u_2)[a\sm v] =((u_1[a\sm v]){=^{\ns P}}(u_2[a\sm v]))$.
\end{lemm}
\begin{proof}
We reason as follows, where we write $u_1'=(a'\ a)\act u_1$ and $u_2'=(a'\ a)\act u_2$ for brevity:
$$
\begin{array}[b]{@{\hspace{-.0ex}}r@{\ }l@{\ \ }l}
p{\in} (u_1{=^{\ns P}}u_2)[a\sm v]
\liff&
\New{a'}\,p[v\ms a']\in (u_1'{=^{\ns P}}u_2')
&\text{Prop~\ref{prop.amgis.iff}}
\\
\liff&
\New{a',c}\, p[v\ms a'][u_1'\ms c]=p[v\ms a'][u_2'\ms c]
&\text{Def~\ref{defn.eq.powamgis}}
\\
\liff&
\New{a',c}\, p[u_1'[a'\sm v]\ms c][v\ms a']=p[u_2'[a'\sm v]\ms c][v\ms a']
&\rulefont{\amgis\sigma},\ c\#v
\\
\liff&
\New{a',c}\, p[u_1'[a'\sm v]\ms c]=p[u_2'[a'\sm v]\ms c]
&\text{Exactness}
\\
\liff&
\New{c}\, p[u_1[a\sm v]\ms c]=p[u_2[a\sm v]\ms c]
&\rulefont{\sigma\alpha},\ a'\#u_1,u_2
\\
\liff&
p\in ((u_1[a\sm v]){=^{\ns P}}u_2[a\sm v])
&\text{Def~\ref{defn.eq.powamgis}}
\end{array}
\qedhere$$
\end{proof}

\begin{prop}
\label{prop.eq.closed}
Suppose $u,v{\in}|\ns U|$.
Then:
\begin{enumerate*}
\item
If $a\#u,v$ then $(u{=^{\ns P}}v)[a\sm w]=(u{=^{\ns P}}v)$.
\item
If $b\#u,v$ then $(u{=^{\ns P}}v)[a\sm b]=(a\ b)\act (u{=^{\ns P}}v)$.
\end{enumerate*}
As a corollary, $(u{=^{\ns P}}v){\in}|\powsigma(\ns P)|$.
\end{prop}
\begin{proof}
The first part is from Lemma~\ref{lemm.eq.sanity.check} and \rulefont{\sigma\#}.
The second part is from Lemmas~\ref{lemm.eq.sanity.check} and~\ref{lemm.sub.alpha} and Theorem~\ref{thrm.equivar}.

For the corollary, finite support is from Theorem~\ref{thrm.no.increase.of.supp}
and conditions~\ref{item.fresh.powsigma} and~\ref{item.alpha.powsigma} of Definition~\ref{defn.powsigma} are just the first two parts of this result. 
\end{proof}

\begin{lemm}
\label{lemm.eq.aeq}
Suppose $u,v,w{\in}|\ns U|$ and suppose $b\#u,v$. 
Then 
$$
(u{=^{\ns P}}v)[a\sm w]=((b\ a)\act u{=^{\ns P}}(b\ a)\act v)[b\sm v].
$$
\end{lemm}
\begin{proof}
From Lemma~\ref{lemm.eq.sanity.check} and \rulefont{\sigma\alpha} (a direct proof by Lemma~\ref{lemm.sigma.alpha} is also possible).
\end{proof}

\begin{corr}
\label{corr.when.in.equality.P}
$p\in(u{=^{\ns P}}v)$ if and only if $p\in(a{=^{\ns P}}b)[a\sm u,b\sm v]$.
\end{corr}
\begin{proof}
Using Lemma~\ref{lemm.eq.aeq} we may assume without loss of generality that $a,b\#u,v$ so that $(a{=^{\ns P}}b)[a\sm u,b\sm v]=(a{=^{\ns P}}b)[a\sm u][b\sm v]$.
It follows by Lemma~\ref{lemm.eq.sanity.check} that $(a{=^{\ns P}}b)[a\sm u,b\sm v]=(u{=^{\ns P}}v)$.
\end{proof}

A converse to Proposition~\ref{prop.xeqx} is valid for a specific model; see Proposition~\ref{prop.xeqx.converse}.
\begin{prop}
\label{prop.xeqx}
Suppose $X{\in}|\powsigma(\ns P)|$ and $u,v{\in}|\ns U|$.
Then:
\begin{enumerate*}
\item
$(u{=^{\ns P}}u)=|\ns P|$.
\item
$(u{=^{\ns P}}v)\cap (X[a\sm u])=(u{=^{\ns P}}v)\cap(X[a\sm v])$.
\end{enumerate*}
\end{prop}
\begin{proof}
Part~1 follows from the fact that $p[u\ms c]=p[u\ms c]$ for every $p$ and every $c$.

For part~2, using condition~\ref{item.alpha.powsigma} of Definition~\ref{defn.powsigma} to rename if necessary, suppose $a\#u,v$.
By Definition~\ref{defn.eq.powamgis} $p\in(u{=^{\ns P}}v)$ means $\New{a}p[u\ms a]=p[v\ms a]$.
Using Proposition~\ref{prop.amgis.iff} and part~2 of Proposition~\ref{prop.eq.closed} and part~2 of Lemma~\ref{lemm.sub.bigcap} we have 
$$
\begin{array}[b]{r@{\ }l}
p\in (u{=^{\ns P}}v)\cap (X[a\sm u])
\liff& \New{a}\bigl(p[u\ms a]=p[v\ms a]\land p[u\ms a]\in X\bigr)
\\
\liff& \New{a}\bigl(p[u\ms a]=p[v\ms a]\land p[v\ms a]\in X\bigr)
\\
\liff&
p\in (u{=^{\ns P}}v)\cap (X[a\sm v]) .
\end{array}
\qedhere$$
\end{proof}

\begin{thrm}
\label{thrm.eq.correct}
If $\ns P$ is an exact $\amgis$-algebra then $(a{=^\ns P}b)$ from Definition~\ref{defn.eq.powamgis} is an equality in $\powsigma(\ns P)$ in the sense of Definition~\ref{defn.eq}.
\end{thrm}
\begin{proof} 
By Proposition~\ref{prop.eq.closed} $(a{=^{\ns P}}b)$ is in $\powsigma(\ns P)$, and $(a{=^{\ns P}}b)$ is an equality by Corollary~\ref{corr.when.in.equality.P} and  Proposition~\ref{prop.xeqx}.
\end{proof}

\subsection{Interpreting first-order logic in the $\sigma$-powerset}

We can now prove Theorem~\ref{thrm.powsigma.FOLeq}:
\begin{thrm}
\label{thrm.powsigma.FOLeq}
Suppose $\ns P$ is an $\amgis$-algebra over a termlike $\sigma$-algebra $\ns U$.
Then $\powsigma(\ns P)$ is a FOL algebra if we define the following:
$$
\begin{array}{r@{\ }l}
\model{\tbot}=&\varnothing
\\
\model{\phi\tand\psi}=&\model{\phi}\cap\model{\psi}
\\
\model{\tneg\phi}=&|\ns P|\setminus\model{\phi}
\\
\model{\tall a.\phi}=& 
\freshcap{a}\model{\phi} \qquad \text{(Def~\ref{defn.nu.U})}
\end{array}
$$
If in addition $\ns P$ is exact (Definition~\ref{defn.exact.amgis.algebra})
then $\powsigma(\ns P)$ is a FOLeq algebra if we define the following:
$$
\model{a{\teq}b}=
\{p\mid \New{c}(p[a\ms c]=p[b\ms c])\}
$$
\end{thrm}
\begin{proof}
\begin{itemize*}
\item
The Boolean structure ($\tbot$, $\tand$, and $\tneg$) is treated in Proposition~\ref{prop.powsigma.bool}.
\item
The $\sigma$-action ($[a\sm u]$) is compatible by Lemma~\ref{lemm.sub.bigcap} and Proposition~\ref{prop.all.sub.commute}.
\item
We observe by Lemma~\ref{lemm.easy}(\ref{easy.compatible.monotone}) that the $\sigma$-action is monotone, and by Theorem~\ref{thrm.all.closed} that $\freshcap{a}X\in|\powsigma(\ns P)|$.
We combine these observations with Definition~\ref{defn.nu.U} and Proposition~\ref{prop.char.freshwedge} to conclude that the fresh-finite limit $\freshwedge{a}X$ exists in $\powsigma(\ns P)$.
\item
Equality ($\teq$) is treated in Theorem~\ref{thrm.eq.correct}.
\qedhere\end{itemize*}
\end{proof}

\begin{rmrk}
\label{rmrk.sum.up}
For the reader's convenience we assemble the maths so far by unpacking what Definitions~\ref{defn.interp} and~\ref{defn.interp.I} mean:
\begin{itemize*}
\item \emph{Definition~\ref{defn.interp}.}\quad
Given a signature $(\Sigma,\Pi,\ar)$ an interpretation maps each $\tf f\in\Sigma$ with arity $n$ to an equivariant function $\tf f^{\iden}:\otimes^n\mathbb A\Func \ns U$ and each $\tf P\in\Pi$ with arity $n$ to an equivariant function $\tf P^{\iden}:\otimes^n\mathbb A\Func \powsigma(\ns P)$ ($\otimes^n$ is from Definition~\ref{defn.multiple.otimes}).
\item \emph{Definition~\ref{defn.interp.I}.}\quad
The interpretation extends as follows:
$$
\begin{array}{r@{\ }l@{\qquad}l}
\model{a}=&\tf{atm}_{\ns U}(a)
&
\text{Def~\ref{defn.term.sub.alg}}
\\
\model{\tf f(r_1,\dots,r_n)}=&\tf f^\iden(a_1,\dots,a_n)[a_1\sm\model{r_1},\dots,a_n\sm\model{r_n}]
&\text{Defs~\ref{defn.extend.f.P} \&~\ref{defn.interp.I}}
\\
\model{\tbot}=&\varnothing
&\text{Prop~\ref{prop.powsigma.bool}}
\\
\model{r{\teq}s}=&
\{p\mid \New{c}(p[\model{r}\ms c]=p[\model{s}\ms c])\}
&\text{Def~\ref{defn.eq.powamgis}}
\\
\model{\tf P(r_1,\dots,r_n)}=&\tf P^\iden(a_1,\dots,a_n)[a_1\sm\model{r_1},\dots,a_n\sm\model{r_n}]
&\text{Defs~\ref{defn.extend.f.P} \&~\ref{defn.interp.I}}
\\
\model{\phi\tand\psi}=&\model{\phi}\cap\model{\psi}
&\text{Prop~\ref{prop.powsigma.bool}}
\\
\model{\tneg\phi}=&|\ns P|\setminus\model{\phi}
&\text{Prop~\ref{prop.powsigma.bool}}
\\
\model{\tall a.\phi}=&\freshcap{a}\model{\phi}=\bigcap_{u{\in}|\ns U|}\model{\phi}[a\sm u]
&\text{Def~\ref{defn.nu.U}}
\end{array}
$$
Recall also that $\tf{atm}_{\ns U}(a)$ comes from Definition~\ref{defn.term.sub.alg}, and where $\ns U$ is assumed we will usually write it just as $a$.
\item
\emph{Theorem~\ref{thrm.fol.sound}.}\quad
From Theorem~\ref{thrm.powsigma.FOLeq} we conclude that if $\phi_1,\dots,\phi_n\cent\psi_1,\dots,\psi_m$ then with the definitions above, $\bigcap_i\model{\phi_i}\subseteq\bigcup_j\model{\psi_j}$.
\end{itemize*}
\end{rmrk}

\section{Completeness}
\label{sect.completeness}

We now set about proving Theorem~\ref{thrm.completeness}, which states that if $\model{\Phi\cent\Psi}$ is true in every interpretation in every FOLeq algebra, then indeed $\Phi\cent\Psi$ is derivable in first-order logic.
The proof follows the general outline that the reader familiar with such proofs might expect: we build a notion of \emph{point} out of maximal consistent sets of predicates (Definition~\ref{defn.points}), map $\phi$ to a set of points (Definition~\ref{defn.syntax.interp}), show that every consistent $\phi$ is contained in some point, note that $\phi$ maps precisely to the set of points containing $\phi$, and then use that to prove completeness.

The main results are Theorems~\ref{thrm.maxfilt.zorn}, \ref{thrm.model.iden} and~\ref{thrm.completeness}.

The constructions are subtle:
\begin{enumerate}
\item
We have more structure than usual: points form not just a set but an $\amgis$-algebra, and sets of points form not just a set but a $\sigma$-algebra, and not just a lattice but a lattice with fresh-finite limits and equality. 
This means more properties to verify, and so more proofs. 
\item
We are used to seeing valuations being used to build the final semantics; not so here.
The semantics, as discussed in the Introduction, is absolute.
This requires a certain change of perspective.
\item
We cannot directly use Zorn's lemma and need to use a carefully designed increasing chain of filter-ideal pairs, and indeed, the definitions of filter and ideal also require careful design.
\end{enumerate}
Detailed exposition follows below; see in particular Remarks~\ref{rmrk.p.not.finite.support},\ \ref{rmrk.ramification}, and~\ref{rmrk.asymmetric}, and also Remark~\ref{rmrk.why.remarkable}.

An argument is possible that the proofs that follow are the \emph{real} proof, and the usual development is a projection of that proof to a Zermelo-Fraenkel sets (ZF) universe.
This is reasonable, since ZF is consistent if and only if FM is, but the ZF construction is not optimal, since FOL syntax interrogates the models just for fresh-finite limits and equality, that is, for $\ttop$, $\tand$, $\tall a$, and $\teq$---which is precisely the structure that FOLeq algebras provide using nominal sets.

\subsection{Filters and points}

\subsubsection{Filters}

Recall \emph{predicates} $\phi$ from Definition~\ref{defn.terms.and.predicates} and the \emph{entailment} relation $\Phi\cent\phi$ from Figure~\ref{fig.FOL}:
\begin{frametxt}
\begin{defn}
\label{defn.filter}
A \deffont{filter} is a nonempty set $p$ of predicates (which need not have finite support) such that: 
\begin{enumerate*}
\item
\label{filter.proper}
$\tbot\not\cin p$ (we call $p$ \deffont{consistent}).
\item
\label{filter.up}
If $\phi\cin p$ and $\phi\cent\phi'$ then $\phi'\cin p$ (we call $p$ \deffont{deductively closed}).
\item
\label{filter.and}
If $\phi\cin p$ and $\phi'\cin p$ then $\phi\tand \phi'\cin p$.
\item
\label{filter.new}
If $\New{b}(b\ a)\act \phi\cin p$ then $\tall a.\phi\cin p$.
\end{enumerate*}
\end{defn}
\end{frametxt}

\begin{rmrk}
\label{rmrk.p.not.finite.support}
$p$ above need not have finite support.
This means that in Definition~\ref{defn.filter} it is \emph{not} the case that $\New{b}(b\ a)\act\phi\cin p$ if and only if $\neg\New{b}(b\ a)\act\phi\not\cin p$. 
The `standard' equivalent decompositions of $\new$ into `$\forall$+freshness' and `$\exists$+freshness' (Theorem~2.17 of \cite{gabbay:stodfo}, Theorem~6.5 of \cite{gabbay:fountl}, or Theorem~9.4.6 of \cite{gabbay:thesis}) might not work with respect to $p$.
However, $x$ in Definition~\ref{defn.filter} is still assumed to have finite support, and that gives us all we need. 

This is why $\amgis$-algebras (Definition~\ref{defn.bus.algebra}) are not assumed to be nominal sets, but are only assumed to be sets with a permutation action---and why $\sigma$-algebras (Definition~\ref{defn.sub.algebra}) are assumed to be nominal sets.

More discussion of the support of $p$ in Remark~\ref{rmrk.ramification}.
\end{rmrk}

Ideals are dual to filters; Definition~\ref{defn.ideal} is standard:
\begin{defn}
\label{defn.ideal}
An \deffont{ideal} is a nonempty set $Z$ of predicates (which need not have finite support) such that:
\begin{enumerate*}
\item\label{ideal.top}
$\ttop\not\in Z$.
\item\label{ideal.down}
If $\psi\in Z$ and $\psi'\leq \psi$ then $\psi'\in Z$ (we call $Z$ \deffont{down-closed}).
\item\label{ideal.or}
If $\psi\in Z$ and $\psi'\in Z$ then $\psi\tor \psi'\in Z$.
\end{enumerate*}
\end{defn}

\begin{rmrk}
Definition~\ref{defn.ideal} is not a perfect dual to Definition~\ref{defn.filter}: we do not have a fourth condition corresponding to condition~\ref{filter.new} of Definition~\ref{defn.filter}.

This is deliberate and will be important; see Remark~\ref{rmrk.asymmetric}.
\end{rmrk}

\begin{defn}
\label{defn.uparrow}
If $\phi$ is a predicate then define $\phi{\uparrow}$ and $\phi{\downarrow}$ by
$$
\phi{\uparrow} = \{\xi \mid \phi\cent\xi\} 
\qquad\text{and}\qquad
\phi{\downarrow} = \{\xi \mid \xi\cent\phi\}  .
$$
\end{defn}

\begin{lemm}
\label{lemm.uparrow.filter}
If $\phi\not\cent\tbot$ then $\phi{\uparrow}$ from Definition~\ref{defn.uparrow} is a filter and $\phi{\downarrow}$ is an ideal.
\end{lemm}
\begin{proof}
It is routine to verify conditions~\ref{filter.proper} to~\ref{filter.and} of Definition~\ref{defn.filter}.
We now consider condition~\ref{filter.new}.
Suppose $\phi\cent (b\ a)\act \xi$ for all but finitely many $b$.
We take one particular $b\#\phi,\xi$ and conclude using \rulefont{\tall R} that $\phi\cent\tall a.\xi$.
The case of $\phi{\downarrow}$ is no harder. 
\end{proof}

\subsubsection{On universal quantification in filters}

Lemma~\ref{lemm.freshness} is needed for Proposition~\ref{prop.these.are.equivalent}.
It is a fact of syntax related to Lemma~\ref{lemm.sub.alpha}.\footnote{It would be easy to hide this with a bit of hand-waving.  However, in the context of nominal techniques, we give it its own result and indicate the issues involved.}
\begin{lemm}
\label{lemm.freshness}
Suppose $\phi$ is a predicate and suppose $b$ is fresh (so $b\#\phi$).
Then $\phi[a\sm b]$ is $\alpha$-equivalent to $(b\ a)\act\phi$.
\end{lemm}
\begin{proof}[Sketch proof]
The substitution $[a\sm b]$ in $\phi[a\sm b]$ is the standard capture-avoiding substitution on the syntax of the predicate $\phi$, from Definition~\ref{defn.terms.and.predicates}.
The permutation $(b\ a)\act\phi$ is the permutation acting on the syntax of $\phi$.
But when we write `$b$ fresh', what exactly do we mean; fresh for what?  
For $\phi$? 
For the $\alpha$-equivalence class of $\phi$?  
Or are we using the $\new$-quantifier?
In fact, all of these are equivalent. 
The interested reader can find the key technical devices for a slick proof in \cite[Section~5]{gabbay:fountl}, notably in Lemma~5.16. 
\end{proof}

\begin{prop}
\label{prop.these.are.equivalent}
Suppose $p$ is a filter and $a$ is an atom.
Then:
\begin{enumerate}
\item
The following conditions are equivalent (below, $n$ ranges over all atoms, including $a$):
$$
\tall a.\phi\cin p
\quad\liff\quad
\Forall{u}\phi[a\sm u]\cin p
\quad\liff\quad
\Forall{n{\in}\mathbb A}\phi[a\sm n]\cin p
\quad\liff\quad
\New{b}(b\ a)\act \phi\cin p
$$
\item
If furthermore $p$ is finitely supported and $a\#p$ then the following conditions are equivalent:
$$
\tall a.\phi\cin p\quad\liff\quad \phi\cin p 
$$
\end{enumerate}
\end{prop}
\begin{proof}
First we prove part~1.
Suppose $\tall a.\phi\cin p$.
By \rulefont{\tall L} from Figure~\ref{fig.FOL} and condition~\ref{filter.up} of Definition~\ref{defn.filter} also $\phi[a\sm u]\cin p$ for every term $u$.
It follows in particular that $\phi[a\sm n]\cin p$ for every atom (i.e. variable symbol) $n$.
It follows that $\phi[a\sm b]\in p$ for all $b\#\phi$ so by Lemma~\ref{lemm.freshness} also $\New{b}(b\ a)\act \phi\cin p$.
Finally if $\New{b}(b\ a)\act \phi\cin p$ then by condition~\ref{filter.new} of Definition~\ref{defn.filter} also $\tall a.\phi\cin p$.

Part~2 follows from part~1 using Theorem~\ref{thrm.new.equiv}.
\end{proof}

It will be useful to have Notation~\ref{nttn.ultrafilter}:
\begin{nttn}
\label{nttn.ultrafilter}
Call a filter $p$ an \deffont{ultrafilter} when for every $\phi$, precisely one of $\phi\cin p$ and $\tneg\phi\in p$ holds.
\end{nttn}

\begin{rmrk}
\label{rmrk.ramification}
Proposition~\ref{prop.these.are.equivalent} has some interesting ramifications which relate back to Remark~\ref{rmrk.p.not.finite.support}.

We will shortly see how to construct ultrafilters (which are also maximally consistent sets of predicates) in Subsections~\ref{subsect.prime.ultra} and~\ref{subsect.zorn}, and 
organise them into an $\amgis$-algebra of \emph{points} in Subsection~\ref{subsect.amgis.on.filters} (see in particular Proposition~\ref{prop.points.amgis}).
We will use this to prove completeness in Subsection~\ref{subsect.proof.completeness}.

Proposition~\ref{prop.these.are.equivalent} lets us make some interesting observations about what these ultrafilters look like.
It is not hard to prove that if an ultrafilter $p$ has finite support then by Proposition~\ref{prop.these.are.equivalent} for any $a\#p$ and $\phi$,  
$$
\tall a.\phi\in p\ \ \liff\ \ \phi\in p\ \ \liff\ \ \texi a.\phi\in p
$$
where $\texi a.\phi$ is shorthand for $\tneg(\tall a.\tneg\phi)$,
and so for any $a$ (even if it is not fresh for $p$) and $\phi$, 
$$
((\tall a.\phi)\ \tiff\ \texi a.\phi)\in p . 
$$
So a finitely supported ultrafilter `believes' that $\tall=\texi$, and so `believes' that quantification is trivial.  
Quantification motivates this paper, so a finitely supported ultrafilter is for us degenerate and uninteresting. 

In that sense Proposition~\ref{prop.these.are.equivalent} implies that ultrafilters must have infinite support---at least, the ones we care about must have this property.

Ultrafilters are dual to predicates in the sense of duality theory;
the dual notion to `must have finite support' (like predicates and open sets) is evidently `must \emph{not} have finite support' (like ultrafilters).\footnote{Note that as discussed above, the story is more subtle than that.  Ultrafilters can have finite support; it is just that the finitely supported ultrafilters represent worlds in which quantification is not interesting.   More work could be done to make this discussion fully formal.} 

Informally this seems reasonable---but to some extent, Proposition~\ref{prop.these.are.equivalent} makes it formal.
\end{rmrk}

\subsubsection{Growing filters}

It will be useful to make larger filters out of smaller filters:
\begin{defn}
\label{defn.tand.bar}
Suppose $p$ is a set of predicates and $\psi$ is a predicate.
Then define:
$$
\begin{array}{r@{\ }l}
p{+}\psi&=\{\xi \mid \phi\tand \psi\cent\xi,\  \phi\in p\}
\end{array}
$$
\end{defn}

\begin{lemm}
\label{lemm.technical.contradiction}
Suppose $p$ is a finitely supported filter and suppose $\psi$ is a predicate.
Then:
\begin{itemize*}
\item
$p\subseteq p{+}\psi$.
\item
$\psi\in p{+}\psi$.
\item
$p{+}\psi$ is closed under conditions~\ref{filter.up} to~\ref{filter.new} of Definition~\ref{defn.filter}.
\end{itemize*}
As a corollary, if $Z$ is an ideal and $p{+}y\cap Z=\varnothing$ then $p{+}y$ is a filter.
\end{lemm}
\begin{proof}
The corollary follows from the body of this result because since from condition~\ref{ideal.down} of Definition~\ref{defn.ideal} $\tbot\in Z$, so $\ttop\not\in p$.

We now consider the body of this result.
It is clear from the construction that $p\subseteq p{+}\psi$ and $\psi\in p{+}\psi$.
We now check that $p{+}\psi$ satisfies conditions~\ref{filter.up} to~\ref{filter.new} of Definition~\ref{defn.filter}:
\begin{enumerate*}
\setcounter{enumi}{1}
\item
\emph{If $\xi\in p{+}\psi$ and $\xi\cent \xi'$ then $\psi'\in p{+}\psi$.}\quad
By construction.
\item
\emph{If $\xi\in p{+}\psi$ and $\xi'\in p{+}\psi$ then $\xi{\tand}\xi'\in p{+}\psi$.}\quad
Suppose $\phi\tand \psi\cent\xi$ and $\phi'\tand \psi\cent\xi'$ for $\phi,\phi'\in p$.
Then by condition~\ref{filter.and} of Definition~\ref{defn.filter} $\phi\tand \phi'\in p$, and it is a fact that $(\phi{\tand}\phi')\tand \psi\cent \xi{\tand}\xi'$.
\item
\emph{If $\New{b}((b\ a)\act \xi\cin p{+}\psi)$ then $\tall a.\xi\in p{+}\psi$.}\quad
Suppose for cofinitely many $b$ there exists a $\phi_b\cin p$ such that $\phi_b{\tand}\psi\cent (b\ a)\act \xi$.
Then there certainly exists some $b$ such that $b\#\psi,\xi,p$ and $\phi_b{\tand}\psi\cent (b\ a)\act \xi$.
Note by Proposition~\ref{prop.these.are.equivalent}(2) that also $\tall b.\phi_b\cin p$.

We apply $\tall b$ to both sides and we conclude that 
$$
(\tall b.\phi_b){\tand}\psi\cent \tall b.(b\ a)\act \xi = \tall a.\xi.
$$
By Lemma~\ref{lemm.freshwedge.alpha} and condition~\ref{filter.up} of Definition~\ref{defn.filter} we conclude that $\tall a.\xi\in p{+}\psi$ as required.
\end{enumerate*}
\end{proof}

\subsubsection{Growing ideals}

It will be useful to make larger ideals out of smaller ideals:
\begin{defn}
\label{defn.tor.bar}
Suppose $Z$ and $Y$ are sets of predicates.
Then define $Z{+}Y$ by
\begin{frameqn}
Z{+}Y=\{\xi'\mid \Exists{\xi{\in}Z,n{\in}\mathbb N,\psi_1,\dots,\psi_n{\in}Y}\xi'\leq \xi\tor \psi_1\tor\dots\tor \psi_n\} .
\end{frameqn}
\end{defn}

Lemma~\ref{lemm.ideal.plus} is a version of Lemma~\ref{lemm.technical.contradiction} for ideals.
It is the simpler result, because Definition~\ref{defn.ideal} has nothing corresponding to condition~\ref{filter.up} of Definition~\ref{defn.filter}:
\begin{lemm}
\label{lemm.ideal.plus}
Suppose $Z$ is an ideal and $Y$ is some set of predicates.
Then:
\begin{itemize*}
\item
$Z\subseteq Z{+}Y$ and $Y\subseteq Z{+}Y$.
\item
$Z{+}Y$ is closed under conditions~2 and~3 of Definition~\ref{defn.ideal} (so that if $\ttop\not\in Z{+}Y$ then it is an ideal).
\end{itemize*}
\end{lemm}
\begin{proof}
By routine calculations.
\end{proof}

\subsubsection{Prime filters and ultrafilters}
\label{subsect.prime.ultra}

\begin{defn}
\label{defn.prime.filter}
\begin{itemize*}
\item
Call a filter $p$ \deffont{prime} when $\phi_1{\tor}\phi_2\in p$ implies either $\phi_1\in p$ or $\phi_2\in p$.
\item
Suppose $p$ is a filter and $Z$ is an ideal.
Call $p$ \deffont{maximal with respect to $Z$} when $p{\cap}Z=\varnothing$ and for every filter $p'$ with $p'{\cap}Z=\varnothing$, if $p\subseteq p'$ then $p=p'$.
\item
Call $p$ \deffont{maximal} when it is maximal with respect to the ideal $\{\tbot\}$.
\end{itemize*}
We will use the terms \emph{prime filter} and \deffont{point} synonymously henceforth (see also Definition~\ref{defn.points}).
\end{defn}

Lemma~\ref{lemm.prime.ultra} is standard, but we still check carefully that being `nominal' does not interfere with the classical propositional structure; it all works.
We need the result for Theorem~\ref{thrm.model.iden}:
\begin{lemm}
\label{lemm.prime.ultra}
A filter $p$ is prime if and only if it is an ultrafilter (Notation~\ref{nttn.ultrafilter}).
\end{lemm}
\begin{proof}
Suppose $p$ is prime.
We note that $(\xi\tor\tneg\xi)\liff\ttop$ so by condition~\ref{filter.up} of Definition~\ref{defn.filter} and non-emptiness of $p$, $\xi\tor\tneg\xi\cin p$ (we use the conditions of Definition~\ref{defn.filter} silently henceforth).
By primeness of $p$, at least one of $\xi\cin p$ and $\tneg\xi\cin p$ holds.
Similarly $(\xi\tand\tneg\xi)\liff\tbot$ so that $(\xi\tand\tneg\xi)\not\in p$.
It follows that precisely one of $\xi\cin p$ and $\tneg\xi\cin p$ holds, so that $p$ is an ultrafilter.

Conversely suppose $p$ is an ultrafilter and suppose $\xi\tor\xi'\in p$ and yet $\xi\not\cin p$ and $\xi'\not\cin p$.
Then $\tneg\xi\tand\tneg\xi'\cin p$, so that $\tbot\cin p$, a contradiction.
\end{proof}

\begin{prop}
\label{prop.max.2}
Suppose $p$ is a filter and $Z$ is an ideal. 
If $p$ is a maximal filter with respect to $Z$ then it is prime.
\end{prop}
\begin{proof}
Suppose $\psi_1\tor \psi_2\cin p$ and $\psi_1,\psi_2\not\cin p$.
By Lemma~\ref{lemm.technical.contradiction} and maximality we have that $p{+}\psi_1\cap Z\neq\varnothing$ and $p{+}\psi_2\cap Z\neq\varnothing$.
It follows that there exist $\phi_1,\phi_2\cin p$ with $\phi_1{\tand}\psi_1\cin Z$ and $\phi_2{\tand}\psi_2\cin Z$.
By condition~\ref{ideal.or} of Definition~\ref{defn.ideal}
$$
(\phi_1{\tand}\psi_1)\tor(\phi_2{\tand}\psi_2)\cin Z .
$$
Now we rearrange the left-hand side to deduce that
$$
\zeta=(\phi_1{\tor}\phi_2)\tand (\phi_1{\tor}\psi_2) \tand (\psi_1{\tor}\phi_2) \tand (\psi_1{\tor}\psi_2)\cin Z .
$$
We now note that $\phi_1{\tor}\phi_2\cin p$ (since $\phi_1\cin p$, and indeed also $\phi_2\cin p$) and $\phi_1{\tor}\psi_2\cin p$ (since $\phi_1\cin p$) and $\psi_1{\tor}\phi_2\cin p$ (since $\phi_2\cin p$) and $\psi_1{\tor}\psi_2\cin p$ by assumption.
But then $\zeta\cin p$, so that by condition~\ref{filter.up} of Definition~\ref{defn.filter} $\zeta\in p{\cap}Z$, a contradiction.
\end{proof}

Lemma~\ref{lemm.can.extend.Z} is a technical result which will be useful for proving Theorem~\ref{thrm.maxfilt.zorn}:
\begin{lemm}
\label{lemm.can.extend.Z}
Suppose
\begin{itemize*}
\item
$p$ is a finitely supported filter and $Z$ is an ideal, and suppose
\item
$\tbot\in p{+}\tall a.\psi$ and $p{\cap}Z=\varnothing$.
\end{itemize*}
Write $Y=\{(b\ a)\act \psi\mid b\in \mathbb A\setminus(\supp(p){\cup}\supp(\psi)\cup\{a\})\}$.
Then:
\begin{enumerate*}
\item
$p\cap (Z{+}Y)=\varnothing$.
\item
As a corollary, $Z{+}Y$ is an ideal (and by part~1 is disjoint from $p$).
\end{enumerate*}
\end{lemm}
\begin{proof}
Suppose $\tbot\in p{+}\tall a.\psi$ and $p{\cap}Z=\varnothing$ and $p\cap (Z{+}Y)\neq \varnothing$.
We note the following:
\begin{itemize*}
\item
Since $p{\cap}(Z{+}Y)\neq\varnothing$, there exist $b_1,\dots,b_n\#p,\psi$ and $\xi{\in}Z$
with $\xi\tor(b_1\ a)\act \psi\tor\dots\tor(b_n\ a)\act \psi\in p$.
\item
Since $\tbot\in p{+}\tall a.\psi$ there exists $\phi{\in}p$ with $\cent\phi\tand\tall a.\psi\,\tiff\, \tbot$.

By Proposition~\ref{prop.these.are.equivalent}(2) (since $b_1,\dots,b_n\#p$) $\tall b_1\dots\tall b_n.\phi\in p$, and we see that we may assume without loss of generality that $b_1,\dots,b_n\#\phi$.
\item
Since $p{\cap}Z=\varnothing$ we have $\Forall{\phi'{\in}p,\xi'{\in}Z}\,(\phi'\tand \xi'{=}\tbot)$.
\end{itemize*}
By condition~\ref{filter.and} of Definition~\ref{defn.filter}
$$
\phi\tand(\xi\tor(b_1\ a)\act \psi\tor\dots\tor(b_n\ a)\act \psi)\in p
$$
and therefore by distributivity (Definition~\ref{defn.distrib})
$$
(\phi\tand \xi)\tor(\phi\tand (b_1\ a)\act \psi)\tor\dots\tor(\phi\tand (b_n\ a)\act \psi)
\stackrel{\phi\tand \xi=\tbot}=
(\phi\tand (b_1\ a)\act \psi)\tor\dots\tor(\phi\tand (b_n\ a)\act \psi)
\in p .
$$
By assumption $b_1,\dots,b_n\#p$ so by Proposition~\ref{prop.these.are.equivalent}(2) we deduce that
$$
\tall b_1\dots b_n.\bigl((\phi\tand (b_1\ a)\act \psi)\tor\dots\tor(\phi\tand (b_n\ a)\act \psi)\bigr)
\in p .
$$
Recall that by assumption $b_1,\dots,b_n\#\phi,\psi$; by properties of first-order logic (including $\alpha$-equivalence and distributivity) we conclude that
$$
\phi\tand \tall a.\psi\in p
\quad\text{and so that}\quad \tbot\in p .
$$
This contradicts condition~\ref{filter.proper} of Definition~\ref{defn.filter}.

For the corollary, it follows from condition~\ref{filter.up} of Definition~\ref{defn.filter} that $\ttop\in p$ so that $\ttop\not\in Z{+}Y$.
We use Lemma~\ref{lemm.ideal.plus}.
\end{proof}

\subsubsection{The Zorn argument}
\label{subsect.zorn}

\begin{lemm}
\label{lemm.ascending.chain.ideals}
If $Z_1\subseteq Z_2\subseteq Z_3\subseteq\dots$ is an ascending chain of ideals then $\bigcup_i Z_i$ is an ideal.
\end{lemm}
\begin{proof}
By standard calculations on the three conditions of Definition~\ref{defn.ideal}.
\end{proof}

A version of Lemma~\ref{lemm.ascending.chain.ideals} for filters does not hold, because condition~\ref{filter.new} of Definition~\ref{defn.filter} is not closed under ascending chains of filters.
We can still make a Zorn-like argument (for a carefully selected chain) to prove the existence of maximal filters: 
\begin{thrm}
\label{thrm.maxfilt.zorn}
Suppose $p$ is a finitely supported filter and $Z$ is a finitely supported ideal and suppose $p\cap Z=\varnothing$.
Then there exists a prime filter $q$ with $p\subseteq q$ and $q\cap Z=\varnothing$.

As corollaries:
\begin{itemize*}
\item
If $\phi\not\cent\psi$ then there exists a prime filter $q$ such that $\phi\cin q$ and $\psi\not\cin q$.
\item
If $\phi\not\cent\tbot$ then there exists a prime filter $q$ such that $\phi\cin q$.  
\end{itemize*}
\end{thrm}
\begin{proof}
The corollaries follow by considering $\phi{\uparrow}$ and $\psi{\downarrow}$ (see Definition~\ref{defn.uparrow} and Lemma~\ref{lemm.uparrow.filter}).
We now consider the main part of this result.

Write $\f{Predicates}$ for the set of all predicates from Definition~\ref{defn.terms.and.predicates}.
Enumerate $\mathbb A\times\f{Predicates}$ as a list of pairs $(a_i,\phi_i)_{i{\in}\mathbb N}$; by assumption in Definition~\ref{defn.atoms} atoms are countable, and by construction so are predicates, so we can do this.

We define a sequence of (by Theorem~\ref{thrm.no.increase.of.supp}) finitely supported disjoint filter-ideal pairs $(p_i,Z_i)_{i{\in}\mathbb N}$ as follows:
\begin{enumerate*}
\item
$(p_0,Z_0)=(p,Z)$.
By assumption $p\cap Z=\varnothing$.
\item
Suppose $i{\geq}1$ and $p_{i\minus 1}{+}\tall a_i.\phi_i\cap Z_{i\minus 1}=\varnothing$.
Then take
$$
(p_i,Z_i)=(p_{i\minus 1}{+}\tall a_i.\phi_i,Z_{i\minus 1}).
$$
It follows from Lemma~\ref{lemm.technical.contradiction} that $p_i$ is a filter and from Theorem~\ref{thrm.no.increase.of.supp} that it is finitely supported.
\item
Suppose $i{\geq}1$ and $p_{i\minus 1}{+}\tall a_i.\phi_i\cap Z_{i\minus 1}\neq\varnothing$.
Then take
$$
\begin{array}{r@{\ }l}
Y=&\{(b\ a_i)\act \phi_i\mid b{\in}\mathbb A\setminus(\supp(p_{i\minus 1}){\cup}\supp(\phi_i){\cup}\{a_i\})\}
\quad\text{and} 
\\
(p_i,Z_i)=&(p_{i\minus 1},Z_{i\minus 1}{+}Y).
\end{array}
$$
It follows from Lemma~\ref{lemm.can.extend.Z} that $Z_i$ is an ideal and $p_i\cap Z_i=\varnothing$.
\end{enumerate*}
We can note the following:
\begin{itemize}
\item
\emph{We note that $p\subseteq\bigcup_i p_i$ and $Z\subseteq\bigcup_i Z_i$.}\quad
Since we took $(p_0,Z_0)=(p,Z)$ and assumed $p\cap Z=\varnothing$.
\item
\emph{We note that $\bigcup_i p_i$ is disjoint from $\bigcup_i Z_i$ and thus from $Z$.}\quad
By construction.
\item
\emph{We note that $\bigcup_i Z_i$ is an ideal.}\quad
This is Lemma~\ref{lemm.ascending.chain.ideals}.
\item
\emph{We note that $\bigcup_i p_i$ is a filter.}\quad
Conditions~\ref{filter.proper} to~\ref{filter.and} of Definition~\ref{defn.filter} are routine.
To check condition~\ref{filter.new}, suppose $\New{b}(b\ a)\act \phi\in\bigcup_i p_i$; we need to prove $\tall a.\phi\in\bigcup_i p_i$.
So let $j$ be that index in the enumeration above such that $(a,\phi)=(a_j,\phi_j)$.
At stage $j$, when we built $(p_j,Z_j)$, there were two possibilities:
\begin{itemize*}
\item
If $p_{j\minus 1}{+}\tall a.\phi\cap Z=\varnothing$ then we must have put $\tall a.\phi$ into $p_j$, so that $\tall a.\phi\in\bigcup_i p_i$ and we are done.
\item
If $p_j{+}\tall a.\phi\cap Z\neq\varnothing$ then we must have put $(b\ a)\act \phi$ into $Z_j$ for cofinitely many $b$, so that $\New{b}(b\ a)\act \phi\in\bigcup_i Z_i$.

But this is impossible because we assumed that cofinitely many $(b\ a)\act \phi$ were in $\bigcup_i p_i$, which is disjoint from $\bigcup_i Z_i$.
\end{itemize*}
\item
\emph{We note that $\bigcup_i p_i$ is a maximal filter disjoint from the ideal $\bigcup_i Z_i$.}\quad
Consider any $\phi$ and choose some fresh $a$ (so $a\#\phi$).
Note that $\cent\tall a.\phi\liff \phi$.

Let $j$ be that index in the enumeration above such that $(a,\phi)=(a_j,\phi_j)$.
It follows from the structure of the algorithm above that precisely one of $\phi\in p_j$ or $\phi\in Z_j$ will hold.
Maximality follows.
\end{itemize}
Therefore by Proposition~\ref{prop.max.2} $\bigcup_i p_i$ is a prime filter, and by construction $p\subseteq\bigcup_i p_i$ and $(\bigcup_i p_i)\cap Z=\varnothing$.
\end{proof}

\begin{rmrk}
\label{rmrk.asymmetric}
The proof of Theorem~\ref{thrm.maxfilt.zorn} is unusual in that we consider a chain of filter-ideal pairs instead of a chain of filters.

The reasons for this may not be obvious.
Why bother?  
When we add $Y$ to $Z_{i\minus 1}$, why not add $Y'=\{\tneg\psi\mid\psi\in Y\}$ to $p_{i\minus 1}$ instead?
After all, we have negation.

However if we do this then the set we might write as $p{+}Y'$ would fail to be a filter, because it would contain $(b\ a)\act\tneg\phi$ for cofinitely many $b$, but not $\tall a.\tneg\phi$.
Furthermore we cannot just add $\tall a.\tneg\phi$ because we may not \emph{want} this to be in $p$---what we want is subtly different: that $(b\ a)\act\phi$ not be added for cofinitely many $b$ `by accident' at later stages (possibly infinitely many later stages).

Thus the asymmetry between Definitions~\ref{defn.filter} and~\ref{defn.ideal}, that Definition~\ref{defn.filter} has a condition for $\tall$ (condition~\ref{filter.new}) whereas Definition~\ref{defn.ideal} does not, is a useful part of the design.

Perhaps we should write about \emph{strong} filters and ideals which satisfy a $\tall$/$\texi$-condition, and \emph{weak} ones which need not.
In that terminology, we can note the interesting fact that our proof of the existence of maximal strong filters / ideals requires us to consider also weak ideals / filters.
\end{rmrk}

\subsection{The amgis-action on (prime) filters}
\label{subsect.amgis.on.filters}

Recall from Definition~\ref{defn.p.action} the pointwise $\amgis$-action $p[u\ms a]=\{x \mid x[a\sm u]\in p\}$. 
In this subsection we check that this preserves the property of being a (prime) filter (Definitions~\ref{defn.filter} and~\ref{defn.prime.filter}) of predicates (so $p$ is a set of predicates and $u$ in the $\amgis$-action is a term from Definition~\ref{defn.terms.and.predicates}).

The work of this subsection happens in Lemma~\ref{lemm.bus.filter}; Proposition~\ref{prop.points.amgis} then puts the result in a some nice packaging.

\begin{lemm}
\label{lemm.bus.filter}
If $p$ is a filter then so is $p[u\ms a]$.
Furthermore, if $p$ is prime then so is $p[u\ms a]$.
\end{lemm}
\begin{proof}
We check the conditions in Definition~\ref{defn.filter}.
We use Proposition~\ref{prop.sigma.iff} without comment:
\begin{itemize*}
\item
\emph{$\tbot\not\in p[u\ms a]$.}\quad
Since it is a fact of syntax that $\tbot[a\sm u]=\tbot$.
We use condition~\ref{filter.proper} of Definition~\ref{defn.filter}.
\item
\emph{If $\phi{\cin}p[u\ms a]$ and $\phi\cent\phi'$ then $\phi'{\cin}p[u\ms a]$.}\quad
If $\phi\cent\phi'$ then it is a fact of first-order logic derivability that also $\phi[a\sm u]\cent\phi'[a\sm u]$.
We use condition~\ref{filter.up} of Definition~\ref{defn.filter}.
\item
\emph{If $\phi{\cin}p[u\ms a]$ and $\phi'{\cin}p[u\ms a]$ then $\phi{\tand}\phi'\in p[u\ms a]$.}\quad
It is a fact of syntax that $(\phi{\tand}\phi')[a\sm u]=\phi[a\sm u]{\tand}(\phi'[a\sm u])$.
We use condition~\ref{filter.and} of Definition~\ref{defn.filter}.
\item
\emph{If $\New{b'}((b'\ b)\act \phi\cin p[u\ms a])$ then $\tall b.\phi\cin p[u\ms a]$.}\quad
Choose some fresh $c$ (so $c\#x,u$).
By Corollary~\ref{corr.stuff} $(b'\ c)\act(c\ b)\act\phi=(b'\ b)\act \phi$ and by \rulefont{\tall\alpha} also $\tall b.\phi=\tall c.(c\ b)\act \phi$.
Thus, we may assume without loss of generality that $b\#u$. 

Now suppose $((b'\ b)\act \phi)[a\sm u]\cin p$ for all but finitely many $b'$; so suppose $b'\#u$.
By Corollary~\ref{corr.stuff} $(b'\ b)\act u=u$, so that $(b'\ b)\act(\phi[a\sm u])\in p$ for all but finitely many $b'$.
By condition~\ref{filter.new} of Definition~\ref{defn.filter} $\tall b.(\phi[a\sm u])\cin p$ and since $b\#u$ it follows that $(\tall b.\phi)[a\sm u]\cin p$.
\end{itemize*}
Now suppose $p$ is prime and suppose $(\psi_1{\tor}\psi_2)[a\sm u]\cin p$.
Then $\psi_1[a\sm u]{\tor}(\psi_2[a\sm u])\cin p$, so that either $\psi_1[a\sm u]\cin p$ or $\psi_2[a\sm u]\cin p$.
\end{proof}

Recall from Definition~\ref{defn.prime.filter} that we called prime filters \emph{points}:
\begin{defn}
\label{defn.points}
Write $\Points$ for the $\amgis$-algebra determined by prime filters and the pointwise actions from Definition~\ref{defn.p.action}.
That is:
\begin{frametxt}
\begin{itemize*}
\item
$|\Points|=\{p \mid p\text{ is a prime filter}\}$.
\item
$\pi\act p=\{\pi\act \phi\mid \phi\cin p\}$ and $p[u\ms a]=\{\phi \mid \phi[a\sm u]\cin p\}$. 
\end{itemize*}
\end{frametxt}
\end{defn}

\begin{prop}
\label{prop.points.amgis}
$\Points$ is indeed an $\amgis$-algebra.
Furthermore, $\Points$ is exact (Definition~\ref{defn.exact.amgis.algebra}).
\end{prop}
\begin{proof}
The first part is just Lemma~\ref{lemm.bus.filter} combined with Proposition~\ref{prop.amgis.2}.

We now prove exactness.
Suppose $\New{c}p[u\ms c]=q[u\ms c]$.

Consider some predicate $\phi$.
Note that cofinitely many (all but finitely many) atoms $c$ satisfy all of the following properties, since cofinitely many atoms satisfy each of them, and a finite intersection of cofinite sets is cofinite:
\begin{itemize*}
\item
$p[u\ms c]=q[u\ms c]$ (by assumption), 
\item
$c\#\phi$, so that also $\phi=\phi[c\sm u]$ (a fact of syntax),
\item
$\phi[c\sm u]\cin p\liff \phi\cin p[u\ms c]$ (Proposition~\ref{prop.sigma.iff}),
and 
\item
$\phi[c\sm u]\cin q \liff \phi\cin q[u\ms c]$ (Proposition~\ref{prop.sigma.iff}).
\end{itemize*}
So there exists at least one $c$ satisfying all the above.
Then 
$$
\phi\cin p
\liff
\phi[c\sm u]\cin p
\liff
\phi\cin p[u\ms c]
\liff
\phi\cin q[u\ms c]
\liff
\phi[c\sm u]\cin q
\liff
\phi\cin q .
\qedhere$$
\end{proof}

\subsection{Proof of completeness}
\label{subsect.proof.completeness}

Recall from Definition~\ref{defn.points} the set of \emph{points} $\Points$, which by Proposition~\ref{prop.points.amgis} is an exact $\amgis$-algebra.
Recall from Definition~\ref{defn.powsigma} the definition of $\sigma$-powerset of an $\amgis$-algebra, for example $\powsigma(\Points)$, and recall that by Theorem~\ref{thrm.powsigma.FOLeq} it is a FOLeq algebra.

Recall from Example~\ref{xmpl.fot} and Definition~\ref{defn.terms.and.predicates} the syntax of terms $r$ with real substitution $r[a\sm t]$.  Then:
\begin{defn}
\label{defn.syntax.interp}
Define an interpretation $\interp I$ (Definition~\ref{defn.interp}) as follows:
\begin{itemize*}
\item
Take $\ns U$ to be the syntax of terms with real substitution, and $\tf f^\iden(a_1,\dots,a_n)=\tf f(a_1,\dots,a_n)$.
\item
Take $\mathcal L$ to be $\powsigma(\Points)$ considered as a FOLeq algebra.
\item
Take $\tf P^\iden(a_1,\dots,a_n)=\{p\in\Points \mid \tf P(a_1,\dots,a_n)\cin p\}$.
\end{itemize*}
\end{defn}

As we shall see, Definition~\ref{defn.syntax.interp} is what we need to build a complete model of first-order logic with equality.
To check this, it only remains to connect Definition~\ref{defn.syntax.interp} to the sets machinery we have built so far, and verify what comes out the other end.

\begin{lemm}
\label{lemm.model.syntax.r}
If $r$ is a term (Definition~\ref{defn.terms.and.predicates}) then $\model{r}=r$.
\end{lemm}
\begin{proof}
An easy consequence of taking $\ns U$ to be syntax with substitution, and of Definition~\ref{defn.extend.f.P}.
\end{proof}

\begin{rmrk}
\label{rmrk.after}
We do not have to specify the interpretation of equality $\teq$ because it is fixed by Definition~\ref{defn.eq.powamgis}, 
just as the intepretations of $\tbot$, $\tneg$, $\tand$, and $\tall$ are fixed.
This was observed in Theorem~\ref{thrm.powsigma.FOLeq}.

We briefly revisit Remark~\ref{rmrk.sum.up} and simplify it using Lemma~\ref{lemm.model.syntax.r}, to sum up how the definitions instantiate to our case of $\powsigma(\Points)$:
$$
\begin{array}{r@{\ }l}
\model{\tbot}=&\varnothing
\\
\model{r{\teq}s}=&
\{p{\in}\Points\mid \New{c}(p[r\ms c]=p[s\ms c])\}
\\
\model{\phi\tand\psi}=&\model{\phi}\cap\model{\psi}
=\{p{\in}\Points\mid p{\in}\model{\phi}\land p{\in}\model{\psi}\}
\\
\model{\tneg\phi}=&\Points\setminus\model{\phi}
=\{p{\in}\Points\mid p{\not\in}\model{\phi}\}
\\
\model{\tall a.\phi}=&\freshcap{a}\model{\phi}=\bigcap_{r}\model{\phi}[a\sm r]
\end{array}
$$
\end{rmrk}

\maketab{tab5}{@{\hspace{-1em}}R{10em}@{\ }L{18em}L{12em}}

\begin{lemm}
\label{lemm.P.p.iff}
\begin{enumerate*}
\item
$r_1{=^\iden}r_2=\{p\in\Points \mid (r_1{\teq}r_2)\cin p\}$.
\item
$\tf P^\iden(r_1,\dots,r_n)=\{p\in\Points \mid \tf P(r_1,\dots,r_n)\cin p\}$.
\end{enumerate*}
\end{lemm}
\begin{proof}
For part~1, we prove two implications:
\begin{itemize*}
\item
Suppose $(r_1{\teq}r_2)\cin p$.
We will show that $p\in (r_1{=^\iden}r_2)$, that is, that $\New{c}p[r_1\ms c]=p[r_2\ms c]$.

We note by Proposition~\ref{prop.sigma.iff} that for any atom $c$, 
$$
\phi\cin p[r_1\ms c] \liff \phi[c\sm r_1]\cin p
\quad\text{and}\quad
\phi\cin p[r_2\ms c] \liff \phi[c\sm r_2]\cin p.
$$
By assumption $(r_1{\teq}r_2)\cin p$ so by properties of first-order logic and condition~\ref{filter.up} of Definition~\ref{defn.filter}
$$
\phi[c\sm r_1]\cin p\quad\text{if and only if}\quad \phi[c\sm r_2]\cin p. 
$$
It follows by Proposition~\ref{prop.sigma.iff} that $\phi\cin p[r_1\ms c]\cent\phi\liff \phi\cin p[r_2\ms c]$.

Since $\phi$ and $c$ were arbitrary, we conclude in particular that $\New{c}p[r_1\ms c]=p[r_2\ms c]$.
\item
Suppose $p\in (r_1{=^\iden}r_2)$, so $\New{c}p[r_1\ms c]=p[r_2\ms c]$.
We will show that $(r_1{\teq}r_2)\cin p$.

Consider some $c$ such that $c\not\in\fa(r_1)\cup\fa(r_2)$ and such that $p[r_1\ms c]=p[r_2\ms c]$.
By condition~\ref{filter.up} of Definition~\ref{defn.filter},\ $(r_1{\teq}r_1)\cin p$.
Also since $c\not\in\fa(r_1)$ we have that $(r_1{\teq}r_1)=(r_1{\teq}c)[c\sm r_1]$, so that $(r_1{\teq}c)[c\sm r_1]\cin p$, and by Proposition~\ref{prop.sigma.iff} $(r_1{\teq}c)\cin p[r_1\ms c]$.
Therefore $(r_1{\teq}c)\cin p[r_2\ms c]$ and again by Proposition~\ref{prop.sigma.iff} and our assumption that $c\not\in\fa(r_2)$ we conclude that $(r_1{\teq}r_2)\cin p$ as required.
\end{itemize*}

For part~2 we reason as follows:
$$
\begin{array}[b]{r@{\ }l@{\qquad}l}
\tf P(r_1,\dots,r_n)\cin p
\liff&
\New{a_1,\dots,a_n}
\tf P(a_1,\dots,a_n)[a_1\sm r_1,\dots,a_n\sm r_n]\cin p
&\text{Fact of syntax}
\\
\liff&
\New{a_1,\dots,a_n}
\tf P(a_1,\dots,a_n)\cin p[r_1\ms a_1]\dots[r_n\ms a_n]
&\text{Prop~\ref{prop.sigma.iff}}
\\
\liff&
\New{a_1,\dots,a_n}
p[r_1\ms a_1]\dots[r_n\ms a_n]\in \tf P^\iden(a_1,\dots,a_n)
&\text{Def~\ref{defn.syntax.interp}} 
\\
\liff&
\New{a_1,\dots,a_n}
p\in \tf P^\iden(a_1,\dots,a_n)[a_1\sm r_1,\dots,a_n\sm r_n]
&\text{Prop~\ref{prop.amgis.iff}}
\\
\liff&
p\in \tf P^\iden(r_1,\dots,r_n)
&\text{Def~\ref{defn.extend.f.P}}
\end{array}
\qedhere$$
\end{proof}

We should briefly check that Definition~\ref{defn.syntax.interp} is well-defined, in the sense that $\tf P^\iden(a_1,\dots,a_n)$ from clause~3 is indeed an element of $\mathcal L$ from clause~2.
This is routine:
\begin{lemm}
Continuing the notation of Definition~\ref{defn.syntax.interp},\ $\tf P^\iden(a_1,\dots,a_n)$ which is equal to $\{p\in\Points\mid \tf P(a_1,\dots,a_n)\cin p\}$, is an element of $\powsigma(\Points)$.
\end{lemm}
\begin{proof}
Finite support is from Theorem~\ref{thrm.no.increase.of.supp} since the support of $\tf P^\iden(a_1,\dots,a_n)$ is bounded by $\{a_1,\dots,a_n\}$, which is finite.
It remains to check conditions~\ref{item.fresh.powsigma} and~\ref{item.alpha.powsigma} of Definition~\ref{defn.powsigma}, namely:
\begin{itemize*}
\item
For fresh $a$ (so $a\not\in\{a_1,\dots,a_n\}$) and any term $r$, $\tf P^\iden(a_1,\dots,a_n)[a\sm r]=\tf P^\iden(a_1,\dots,a_n)$.
\item
For fresh $b$ (so $b\not\in\{a_1,\dots,a_n\}$) and any atom $m$ (so possibly $m\in\{a_1,\dots,a_n\}$), $\tf P^\iden(a_1,\dots,a_n)[m\sm r]=(b\ m)\act \tf P^\iden(a_1,\dots,a_n)$.
\end{itemize*}
For the first condition, by Proposition~\ref{prop.amgis.iff}
$p\in\tf P^\iden(a_1,\dots,a_n)[a\sm r]$ if and only if 
$\New{a'}p[r\ms a']\in\tf P^\iden(a_1,\dots,a_n)$.
By Definition~\ref{defn.syntax.interp} this is if and only if $\New{a'}\tf P(a_1,\dots,a_n)\cin p[r\ms a']$.
By Proposition~\ref{prop.sigma.iff} this is if and only if 
$\New{a'}\tf P(a_1,\dots,a_n)[a'\sm r]\cin p$.
It is now a fact of syntax that $\tf P(a_1,\dots,a_n)[a\sm r]= \tf P(a_1,\dots,a_n)$.

The second condition follows much as the first, using the fact of syntax that 
$\tf P(a_1,\dots,a_n)[m\sm b]=(b\ m)\act\tf P(a_1,\dots,a_n)$. 
\end{proof}

\begin{thrm}
\label{thrm.model.iden}
$\model{\phi}=\{p\in\Points\mid \phi\cin p\}$.
\end{thrm}
\begin{proof}
By induction on $\phi$:
\begin{itemize*}
\item
\emph{The case of $\tbot$.}\quad
By assumption in Definition~\ref{defn.filter} $p\in\Points$ is consistent (meaning $\tbot\not\cin p$).
Furthermore, $\model{\tbot}=\varnothing$.
\item
\emph{The cases of $r_1{\teq}r_2$ and $\tf P(r_1,\dots,r_n)$.}\quad
This is Lemma~\ref{lemm.P.p.iff}.
\item
\emph{The case of $\phi\tand\psi$.}\quad
From the inductive hypothesis and conditions~\ref{filter.up} and~\ref{filter.and} of Definition~\ref{defn.filter}.
\item
\emph{The case of $\tneg\phi$.}\quad
By assumption $p$ is prime (Definition~\ref{defn.prime.filter}), and it follows by Lemma~\ref{lemm.prime.ultra} that $\tneg\phi\cin p$ if and only if $\phi\cin p$.
Also by assumption $\model{\tneg\phi}=\Points{\setminus}\model{\phi}$.
\item
\emph{The case of $\tall a.\phi$.}\quad
Combining Remark~\ref{rmrk.after} with Proposition~\ref{prop.amgis.iff} 
$p\in\model{\tall a.\phi}$ means that $\New{c}p[r\ms c]\in(c\ a)\act \model{\phi}$ for every $r$.
By inductive hypothesis this is if and only if $\New{c}(c\ a)\act\phi\cin p[r\ms c]$ for every $r$, which by Proposition~\ref{prop.sigma.iff} and $\alpha$-conversion is if and only if $\phi[a\sm r]\cin p$ for every $r$.
We use Proposition~\ref{prop.these.are.equivalent}(1).
\qedhere\end{itemize*}
\end{proof}

By Proposition~\ref{prop.xeqx} $(u{=^\iden}u)=\Points$.
In Proposition~\ref{prop.xeqx.converse} we show that specifically for the interpretation of Definition~\ref{defn.syntax.interp}, a converse holds:
\begin{prop}
\label{prop.xeqx.converse}
If $(r_1{=^\iden}r_2)=\Points$ then $r_1=r_2$ (that is, $r_1$ and $r_2$ are syntactically identical terms). 
\end{prop}
\begin{proof}
Suppose $(r_1{=^\iden}r_2)=\Points$ and $r_1\neq r_2$.
It is a fact of first-order logic that $\tneg(r_1\teq r_2)\not\cent\tbot$ (easily proved using cut-elimination and the syntax-directed nature of the derivation rules).
By Theorem~\ref{thrm.maxfilt.zorn} there exists a point $p$ such that $\tneg(r_1\teq r_2)\cin p$.
Now by assumption $p\in(r_1{=^\iden}r_2)$, so by Theorem~\ref{thrm.model.iden} $(r_1{=^\iden}r_2)\cin p$.
Thus $\tbot\cin p$, a contradiction.
\end{proof}

\begin{rmrk}
\label{rmrk.why.remarkable}
It is worth pausing to note, informally but mathematically, why Propositions~\ref{prop.xeqx} and~\ref{prop.xeqx.converse} are remarkable.

We constructed $\powsigma(\Points)$ by taking the $\sigma$-powerset of points.
The `nominal' aspects of our construction in this paper allows us to interpret $\ttop$, $\tand$, and \emph{also} $\tall$ as three aspects of a single unifying notion of \emph{fresh-finite limits}---but let us set that aside for now.

In addition to the above, Propositions~\ref{prop.xeqx} and~\ref{prop.xeqx.converse} make formal an idea that even if the underlying notion of predicate did not contain an equality---in the sense of going back to Definition~\ref{defn.terms.and.predicates} and erasing equality $=$ from predicates, and going to the first-order logic sequent rules in Figure~\ref{fig.FOL} and erasing \rulefont{{\teq}R} and \rulefont{{\teq}L}---then even so, $\powsigma(\Points)$ would \emph{still} have an equality.

In summary: $\powsigma(\Points)$ 
`generates' a sets-based notion of equality, just as it `generates' sets-based notions of conjunction and universal quantification as fresh-finite limits, and negation as sets complement. 
\end{rmrk}

\begin{corr}
\label{corr.completeness}
If $\phi\not\cent\psi$ then $\model{\phi}\not\subseteq\model{\psi}$. 
\end{corr}
\begin{proof}
Suppose $\phi\not\cent\psi$.
By Theorem~\ref{thrm.maxfilt.zorn}
there exists a point $q$ with $\psi\not\cin q$ and $\phi{\uparrow}{\subseteq}q$ so that $\phi\cin q$.
By Theorem~\ref{thrm.model.iden} it follows that $q\in\model{\phi}$ and $q\not\in\model{\psi}$, and the result follows.
\end{proof}

Theorem~\ref{thrm.completeness} is a converse to Theorem~\ref{thrm.fol.sound}:
\begin{thrm}[Completeness]
\label{thrm.completeness}
If $\model{\Phi\cent\Psi}$ is true for every interpretation in every FOLeq algebra, then $\Phi\cent\Psi$. 
\end{thrm}
\begin{proof}
From Corollary~\ref{corr.completeness}. 
\end{proof}

\begin{corr}
$\model{\phi}=\Points$ if and only if $\cent\phi$ in FOLeq. 
\end{corr}

\section{Tarski models}
\label{sect.complete}

We recall the usual definition of model for first-oder logic, which goes back to \cite{tarski:semct} and is based on valuations.
We will show how to lift this to a nominal model.

The main definition is Definition~\ref{defn.idenot.beta} and the main proofs are in Corollary~\ref{corr.termlike.sub.alg} and Proposition~\ref{prop.standard.nom.bool}.

The reader will probably be familiar with the valuation-based models and may ask: why bother taking something familiar and translating it to something \emph{less} familiar?
Usually, we expect to see it the other way around---the unfamiliar translated to the familiar.  

However, we will argue in Remark~\ref{rmrk.argue} that the nominal semantics is more natural.

\subsection{Tarski-style model of first-order logic}
\label{subsect.standard}

We briefly sketch the standard model of first-order classical logic, with valuations and without atoms.
This model is not intended to be sophisticated; we will just need that one exists.

\begin{nttn}
To avoid the confusion between sets and nominal sets, we may write \emph{ordinary set} for the former.
\end{nttn}

\begin{defn}
\label{defn.varsigma}
Suppose $X$ is an ordinary set.
Write $\atoms\Func X$ for the set of functions from atoms to $X$.
Let $\varsigma$ range over elements of $\atoms\Func X$ and call these \deffont{valuations} (to $X$).

If $x\in X$ then define $\varsigma[a\ssm x]$ to be the function on atoms given by:
$$
\begin{array}{r@{\ }l}
(\varsigma[a\ssm x])(a)=&x
\\
(\varsigma[a\ssm x])(b)=&\varsigma(b)\quad\text{ for all other }b
\end{array}
$$
\end{defn}

\begin{defn}
\label{defn.varsigma.action}
Define a permutation action on $\varsigma\in\atoms\Func X$ by
$$
(\pi\act\varsigma)(a)=\varsigma(\pi^\mone(a)) .
$$
\end{defn}

\begin{rmrk}
\label{rmrk.not.nominal}
$\atoms\Func X$ is a set with a permutation action (Definition~\ref{defn.fin.supp}).
It is not in general a nominal set because it has elements without finite support.\footnote{The finitely supported $\varsigma$ are such that there exists a finite $A\subseteq\mathbb A$ such that for all $a,b\not\in A$, $\varsigma(a)=\varsigma(b)$; it is not worth our while to impose this restriction, though it would do no harm to do so.} 
For our purposes that will not be a problem.

The action from Definition~\ref{defn.varsigma.action} is a special case of the \emph{conjugation action} $(\pi\act\varsigma)(a)=\pi\act\varsigma(\pi^\mone(a))$ (already mentioned in Subsection~\ref{subsect.full.function.space}), where we use the trivial action $\pi\act x=x$ for every $x\in X$.
This is standard; for a specifically `nominal' discussion see \cite[Definition~2.4.2]{gabbay:nomtnl}.
\end{rmrk}

Recall from Definition~\ref{defn.terms.and.predicates} the notion of \emph{signature} $(\Sigma,\Pi,\ar)$ of \emph{function symbols} $\tf f\in\Sigma$ and \emph{predicate symbols} $\tf P\in\Pi$ with their associated \emph{arities} $\ar(\tf f),\ar(\tf P)$.
For the rest of this subsection fix some signature.
\begin{defn}
\label{defn.model.standard}
An \deffont{ordinary model} $\mathcal N=(|\mathcal N|,\text{-}^\nden)$ is a tuple such that:
\begin{itemize*}
\item
$|\mathcal N|$ is some non-empty (ordinary) \deffont{underlying set}.
\item
$\text{-}^\nden$ assigns to each term-former $\tf f$ a function $\tf f^\nden:|\mathcal N|^{\ar(\tf f)}\Func |\mathcal N|$ and to each predicate-former $\tf P$ a function $\tf P^\nden:|\mathcal N|^{\ar(\tf P)}\Func \{\lbot,\ltop\}$. 
\end{itemize*}
\end{defn}

\begin{defn}
$\{\lbot,\ltop\}$ is a complete Boolean algebra.
We use standard definitions, such as $\land$, $\lneg$, $\bigwedge$, and $\bigvee$ without comment
\end{defn}

Recall the syntax of first-order logic from Definition~\ref{defn.terms.and.predicates}.
\begin{defn}
\label{defn.interp.tarski}
Define a (standard) \deffont{interpretation} function $\ndenot{\varsigma}{\text{-}}$ mapping terms and predicates to elements of $|\mathcal N|$ and $\{\lbot,\ltop\}$ (considered as a Boolean algebra) respectively, as follows:
$$
\begin{array}{r@{\ }l@{\qquad}r@{\ }l}
\ndenot{\varsigma}{a}=&\varsigma(a)
&
\ndenot{\varsigma}{\tf f(r_1,\dots,r_n)}=&\tf f^\nden(\ndenot{\varsigma}{r_1},\dots,\ndenot{\varsigma}{r_n})
\\
\ndenot{\varsigma}{\tbot}=&\lbot
&
\ndenot{\varsigma}{\tf P(r_1,\dots,r_n)}=&\tf P^\nden(\ndenot{\varsigma}{r_1},\dots,\ndenot{\varsigma}{r_n})
\\
\ndenot{\varsigma}{\phi'\tand\phi}=&\ndenot{\varsigma}{\phi'}\land\ndenot{\varsigma}{\phi}
&
\ndenot{\varsigma}{\tall a.\phi}=&\bigwedge_{x{\in}|\mathcal N|}\ndenot{\varsigma[a\ssm x]}{\phi}
\\
\ndenot{\varsigma}{\tneg \phi}=&\lneg\ndenot{\varsigma}{\phi}
\end{array}
$$
\end{defn}

Theorem~\ref{thrm.fol.standard} expresses the usual soundness and completeness result for first-order logic; for details and proofs see e.g. \cite[Subsection~1.5]{dalen:logs}:
\begin{thrm}
\label{thrm.fol.standard}
$\Phi\cent\Psi$ is derivable if and only if for every ordinary model $\mathcal N$ and every valuation $\varsigma$ to $|\mathcal N|$ it is the case that $\bigwedge_{\phi\in\Phi}\ndenot{\varsigma}{\phi}=\ltop$ implies $\bigvee_{\hspace{-.6ex}\psi\in\Psi}\ndenot{\varsigma}{\psi}=\ltop$.
\end{thrm}

\subsection{Lifting to a FOLeq algebra}
\label{subsect.nom.model.from.standard}

We now show how to `lift' a model over ordinary sets to a nominal model (Proposition~\ref{prop.standard.nom.bool}).
We then deduce completeness for nominal Boolean algebras (Corollary~\ref{corr.nom.complete}).

\subsubsection{Lifting to a sigma-algebra}

\begin{defn}
\label{defn.idenot.beta}
Given ordinary sets $X$ and $Y$ define 
a termlike $\sigma$-algebra $\Tarski{X,X}$, 
and define a $\sigma$-algebra $\Tarski{X,Y}$ over $\Tarski{X,X}$ by:
$$
\begin{array}{r@{\ }l}
\Tarski{X,X}=&(|\Tarski{X,X}|,\act,\tf{sub}_{\Tarski{X,Y}},\tf{atm}_{\Tarski{X,X}})
\\
\Tarski{X,Y}=&(|\Tarski{X,Y}|,\act,\tf{sub}_{\Tarski{X,Y}})
\end{array}
$$
as follows:
\begin{itemize*}
\item
$|\Tarski{X,Y}|$ is the set of functions $f$ from $\atoms\Func X$ to $Y$ such that there exists a finite set $A_f\subseteq\mathbb A$ such that 
\begin{equation}
\label{Af}
\Forall{a{\in}A_f}\varsigma(a)=\varsigma'(a)
\quad\text{implies}\quad
f(\varsigma)=f(\varsigma')
\end{equation} 
(if $X=Y$ then we obtain $|\Tarski{X,X}|$). 
\item
The permutation action $\act$ is defined by 
$$
(\pi\act f)(\varsigma)=f(\pi^\mone\act\varsigma).
$$
\item
If $f\in\atoms\Func Y$ and $u\in\atoms\Func X$ define 
$$
(f[a\sm u])(\varsigma)=f(\varsigma[a\ssm u(\varsigma)]).
$$
\item
If $X=Y$ then define $\tf{atm}_{\Tarski{X,X}}(a)(\varsigma)=\varsigma(a)$.
\end{itemize*}
\end{defn}

\begin{nttn}
$\tf{atm}_{\Tarski{X,X}}(a)$ is a bit of a mouthful so we may just write it as $a$.
It will always be clear whether we mean `$a\in\atoms$' or `$a\in\Tarski{X,X}$'.
\end{nttn}

\begin{lemm}
\label{lemm.idenot.beta.pnom}
$(|\Tarski{X,Y}|,\act)$ from Definition~\ref{defn.idenot.beta} is indeed a nominal set (Definition~\ref{defn.nominal.set}).
\end{lemm}
\begin{proof}
It is routine to verify that the permutation action is a group action.
It remains to check finite support.

Suppose $\pi(a)=a$ for every $a\in A_f$ where $A_f$ is the finite set of atoms whose existence for each $f$ is assumed in Definition~\ref{defn.idenot.beta}.
By Definition~\ref{defn.idenot.beta} $(\pi\act f)(\varsigma)=f(\pi^\mone\act\varsigma)$.
By Definition~\ref{defn.varsigma.action} $(\pi^\mone\act\varsigma)(a)=\varsigma(\pi(a))$.
Now by assumption $\varsigma(a)=\varsigma(\pi(a))$ for every $a\in A_f$.
Therefore, $(\pi^\mone\act\varsigma)(a)=\varsigma(a)$ for every $a\in A_f$, and so $f(\varsigma)=f(\pi^\mone\act\varsigma)$, and so $(\pi\act f)(\varsigma)=f(\varsigma)$. 
Thus, $f$ has finite support (and is supported by $A_f$).
\end{proof}

\begin{lemm}
\label{lemm.supp'.supp}
Suppose $f\in |\Tarski{X,Y}|$.
Then $\Forall{a{\in}\supp(f)}\varsigma(a){=}\varsigma'(a)$ implies $f(\varsigma)=f(\varsigma')$. 
\end{lemm}
\begin{proof}
It suffices to show that if $a\in A_f{\setminus}\supp(f)$ and $x\in X$ then $f(\varsigma[a\ssm x])=f(\varsigma)$.
Choose fresh $b$ (so $b\not\in A_f$).
By part~1 of Corollary~\ref{corr.stuff} $(b\ a)\act f=f$, since $a,b\not\in\supp(f)$.
We reason as follows:
$$
f(\varsigma[a\ssm x])\stackrel{b{\not\in}A_f}{=} f(\varsigma[a\ssm x][b\ssm\varsigma(a)])\stackrel{(b\,a)\act f{=}f}{=}f(\varsigma[b\ssm x])\stackrel{b{\not\in}A_f}{=} f(\varsigma)
\qedhere$$ 
\end{proof}

\begin{rmrk}
Lemma~\ref{lemm.supp'.supp} does not follow from Lemma~\ref{lemm.idenot.beta.pnom}: to see this,
take $X{=}Y{=}\{0,1\}$ and $f(\varsigma)=\f{min}\{\varsigma(a)\mid a{\in}\mathbb A\}$.
Then by Theorem~\ref{thrm.no.increase.of.supp} $\supp(f){=}\varnothing$ so that $\Forall{a{\in}\supp(f)}\varsigma(a){=}\varsigma'(a)$ for any $\varsigma$ and $\varsigma'$, yet $f(\varsigma)\neq f(\varsigma')$ where $\varsigma=\lam{a{\in}\mathbb A}0$ and $\varsigma'=\lam{a{\in}\mathbb A}1$.

What makes Lemma~\ref{lemm.supp'.supp} work is the interaction of support with condition~\eqref{Af} of Definition~\ref{defn.idenot.beta}.
\end{rmrk}

\begin{corr}
\label{corr.termlike.sub.alg}
$\Tarski{X,X}$ is indeed a termlike $\sigma$-algebra, and $\Tarski{X,Y}$ is a $\sigma$-algebra over $\Tarski{X,X}$.
\end{corr}
\begin{proof}
By Definitions~\ref{defn.term.sub.alg} and~\ref{defn.idenot.beta} and see that we need to check equivariance and the \rulefont{\sigma\ast} axioms 
from Figure~\ref{fig.nom.sigma}.
This is routine:
\begin{itemize*}
\item
\emph{Rule \rulefont{\sigma a} (only for $\Tarski{X,X}$).}\quad
Unpacking definitions, $a[a\sm u](\varsigma)=\varsigma[a\ssm u(\varsigma)](a)=u(\varsigma)$.
\item
\emph{Rule \rulefont{\sigma id}.}\quad
Again we unpack definitions: $f[a\sm a](\varsigma)=f(\varsigma[a\ssm\varsigma(a)])=f(\varsigma)$.
\item
\emph{Rule \rulefont{\sigma\#}.}\quad
Using Lemma~\ref{lemm.supp'.supp}.
\item 
\emph{Rule \rulefont{\sigma\alpha}.}\quad
Suppose $f\in \Tarski{X,Y}$ and $b\#f$, and $u\in \Tarski{X,X}$.
Then $f[a\sm u](\varsigma)=f(\varsigma[a\ssm u(\varsigma)])$.
Also, unpacking Definitions~\ref{defn.varsigma.action} and~\ref{defn.idenot.beta} $((b\ a)\act f)[b\sm u]=f(\varsigma[a\ssm u(\varsigma)][b\ssm \varsigma(a)])$.
We use Lemma~\ref{lemm.supp'.supp}. 
\item
\emph{Rule \rulefont{\sigma\sigma}.}\quad
Suppose $f\in \Tarski{X,Y}$ and $u,v\in\Tarski{X,X}$ and $a\#v$.
We just unpack definitions:
$$
\begin{array}{r@{\ }l@{\qquad}l}
f[a\sm u][b\sm v](\varsigma)=&f[a\sm u](\varsigma[b\ssm v(\varsigma)])
\\
=&f(\varsigma[a\ssm u(\varsigma[b\ssm v(\varsigma)])][b\ssm v(\varsigma)])
\\[2ex]
f[b\sm v][a\sm u[b\sm v]](\varsigma)=&
f[b\sm v](\varsigma[a\ssm u(\varsigma[b\ssm v(\varsigma)])])
\\
=&
f(\varsigma[a\ssm u(\varsigma[b\ssm v(\varsigma)])][b\ssm v(\varsigma[a\ssm u(\varsigma[b\ssm v(\varsigma)])])])
\\
=&
f(\varsigma[a\ssm u(\varsigma[b\ssm v(\varsigma)])][b\ssm v(\varsigma)])
\end{array}
$$
The final step is valid using Lemma~\ref{lemm.supp'.supp} since $a\#v$.
(There are two symbols here: $\sm$ and $\ssm$.  $f[a\sm u]$ is `$f$ with $u$ substituted for $a$' from Definition~\ref{defn.idenot.beta}, and $\varsigma[a\ssm u(\varsigma)]$ is `$\varsigma$ with $a$ maps to $u(\varsigma)$' from Definition~\ref{defn.varsigma}.)
\qedhere\end{itemize*}
\end{proof}

\subsubsection{Interlude: more on limits in nominal posets}
\label{subsect.more.on.limits}

For this subsection, fix a nominal poset $\mathcal L$ (Definition~\ref{defn.nom.poset}).

It is useful to continue and extend the maths from Subsection~\ref{subsect.fresh-finite.limit}.
Recall from Notation~\ref{nttn.fix} the definition of $\fix$: 
\begin{defn}
\label{defn.fresh.orbit}
Following \cite{gabbay:stusun,gabbay:frenrs} define $x\ii{a}$ by
$$
\begin{array}{r@{\ }l}
x\ii{a}=&\{\pi\act x \mid \pi\in\fix(\supp(x){\setminus}\{a\})\}  .
\\
=&\{x\}\cup\{(b\ a)\act x\mid b\#x\}
\end{array}
$$
Write $\bigwedge x\ii{a}$ for the $\leq$-greatest lower bound of $x\ii{a}$, if this exists.
\end{defn}

Lemma~\ref{lemm.fresh.ii} is related to Lemma~\ref{lemm.a.fresh.bigset} and Proposition~\ref{prop.freshwedgeo} is related to Propositions~\ref{prop.char.freshwedge} and~\ref{prop.char.freshwedge.names}:
\begin{lemm}
\label{lemm.fresh.ii}
$a\#x\ii{a}$.
\end{lemm}
\begin{proof}
We note by the pointwise action that $\pi'\act x\ii{a}=\{(\pi'\circ\pi)\act x\mid \pi\in\fix(\supp(x){\setminus}\{a\})\}$.
We choose a fresh $b$ (so $b\#x$, that is, $b{\not\in}\supp(x)$) and use Corollary~\ref{corr.stuff}(3) and routine calculations. 
\end{proof}

\begin{prop}
\label{prop.freshwedgeo}
Suppose $a{\in}\mathbb A$ and $x{\in}|\mathcal L|$.
Then:
\begin{enumerate}
\item
If $\freshwedge{a}x$ exists then so does $\bigwedge x\ii{a}$, and they are equal. 
\item 
Suppose $\mathcal L$ has a monotone $\sigma$-action (Definition~\ref{defn.fresh.continuous}).
Then if $\bigwedge x\ii{a}$ exists, then so does $\freshwedge{a}x$, and they are equal. 
\end{enumerate}
\end{prop}
\begin{proof}
\begin{enumerate}
\item
By Definition~\ref{defn.nom.poset} $\freshwedge{a}x\leq x$ and $a\#\freshwedge{a}x$.
It follows by equivariance of $\leq$ and Corollary~\ref{corr.stuff} that $\freshwedge{a}x\leq \pi\act x$ for every $\pi\in\fix(\supp(x){\setminus}\{a\})$.
Therefore $\freshwedge{a}x$ is a lower bound for $x\ii{a}$.

Now consider $z$ some other lower bound for $x\ii{a}$, so that $\Forall{\pi{\in}\fix(\supp(x){\setminus}\{a\})}z\leq\pi\act x$.
Choose fresh $b$ (so $b\#x,z$); by Theorem~\ref{thrm.equivar} $\freshwedge{b}(b\ a)\act x$ exists (and by Lemma~\ref{lemm.freshwedge.alpha} it is equal to $\freshwedge{a}x$).
It follows by equivariance of $\leq$ and Corollary~\ref{corr.stuff} (since $a,b\#z$) that $z\leq (b\ a)\act x$ so $z\leq\freshwedge{b}(b\ a)\act x\stackrel{\text{L\ref{lemm.freshwedge.alpha}}}=\freshwedge{a}x$.
\item
By Lemma~\ref{lemm.fresh.ii} and Theorem~\ref{thrm.no.increase.of.supp} $a\#\bigwedge x\ii{a}$ so $\bigwedge x\ii{a}$ is an $a\#$lower bound for $x$.

Consider any other $z$ such that $z\leq x$ and $a\#z$.
By \rulefont{\sigma\#} $z[a\sm n]=z$ for every $n{\in}\mathbb A$.
By monotonicity $z=z[a\sm n]\leq x[a\sm n]$ for every $n{\in}\mathbb A$.
It follows using Lemma~\ref{lemm.sub.alpha} that $z\leq \bigwedge x\ii{a}$. 
\qedhere\end{enumerate}
\end{proof}

\begin{rmrk}
We have seen multiple characterisations of quantification: 
\begin{itemize*}
\item
Fresh-finite limits from Notation~\ref{nttn.lall}.
\item
Limits of permutation orbits, above in Proposition~\ref{prop.freshwedgeo}.
\item
Limits of substitution instances (if there is a monotone $\sigma$-action), from Propositions~\ref{prop.char.freshwedge} and~\ref{prop.char.freshwedge.names}.
\end{itemize*}
Note of Proposition~\ref{prop.freshwedgeo} that we only need atom-for-atom substitution (what this author calls a \emph{renaming action}) in the proof.
\end{rmrk}

\begin{defn}
\label{defn.nom.complete}
Call a nominal poset $\mathcal L=(|\mathcal L|,\act,\leq)$ \deffont{nominally complete} when every finitely supported subset $\mathcal X\subseteq|\mathcal L|$ has a greatest lower bound.

In another terminology: $\mathcal L$ is nominally complete when it has limits of finitely supported diagrams. 
\end{defn}

\begin{prop}
\label{prop.nom.complete}
If $\mathcal L$ is nominally complete and has a monotone $\sigma$-action, then it is finitely fresh-complete (Definition~\ref{defn.nom.poset}).

In other words: having finitely supported limits and substitution implies having fresh-finite limits.\footnote{The reverse implication does not hold; having finitely supported limits and substitution is a far stronger condition.}
\end{prop}
\begin{proof}
We use Proposition~\ref{prop.ffc.char}.
By Theorem~\ref{thrm.no.increase.of.supp} $\varnothing{\subseteq}|\mathcal L|$ and $\{x,y\}{\subseteq}|\mathcal L|$ and $x\ii{a}$ are finitely supported.
By Proposition~\ref{prop.freshwedgeo} the $a\#$limit $\freshwedge{a}x$ exists in $\mathcal L$ if and only if the limit of $x\ii{a}$ does. 
The result follows.
\end{proof}

\subsubsection{Lifting the logical structure}

\begin{defn}
\label{defn.leq.nsN}
Suppose $f',f\in |\Tarski{X,\{\lbot,\ltop\}}|$. 
Write $f'\leq f$ when $\Forall{\varsigma}f'(\varsigma)\leq f(\varsigma)$.
\end{defn}

\begin{lemm}
\label{lemm.monotone}
With the partial order $f'\leq f$ of Definition~\ref{defn.leq.nsN}, the $\sigma$-action from Definition~\ref{defn.idenot.beta} is monotone. 
\end{lemm}
\begin{proof}
Suppose $f'\leq f$ and $u\in|\Tarski{X,X}|$ and $\varsigma\in\atoms\Func X$.
We reason as follows:
$$
\begin{array}[b]{r@{\ }l@{\quad}l}
(f'[a\sm u])(\varsigma)=&f'(\varsigma[a\ssm u(\varsigma)])
&\text{Definition~\ref{defn.idenot.beta}}
\\
\leq&f(\varsigma[a\ssm u(\varsigma)]
&\text{Definition~\ref{defn.leq.nsN}}\ f'\leq f
\\
=&(f[a\sm u])(\varsigma)
&\text{Definition~\ref{defn.idenot.beta}}
\end{array}
\qedhere$$ 
\end{proof}

\begin{prop}
\label{prop.standard.nom.bool}
$(|\Tarski{X,\{\lbot,\ltop\}}|,\act,\leq)$ is nominally complete, complemented, and has a compatible $\sigma$-action and an equality.

As a corollary, $(|\Tarski{X,\{\lbot,\ltop\}}|,\act,\leq)$ is also a FOLeq algebra (Definition~\ref{defn.FOLeq}).
\end{prop}
\begin{proof}
The greatest lower bound of a finitely supported set $\mathcal X\subseteq|\Tarski{X,\{\lbot,\ltop\}}|$, and the complement of $f{\in}|\Tarski{X,\{\lbot,\ltop\}}|$, are defined pointwise by
$$
(\bigwedge \mathcal X)(\varsigma)=
\bigwedge_{f{\in}\mathcal X} f(\varsigma) 
\qquad\text{and}\qquad
(\lneg f)(\varsigma)=\lneg(f(\varsigma)) .
$$ 
Equality is defined by
$$
(u_1{=^{\Tarski{X,\{\lbot,\ltop\}}}}u_2)(\varsigma)=\begin{cases}\ltop & u_1(\varsigma)=u_2(\varsigma) \\ \lbot & u_1(\varsigma)\neq u_2(\varsigma) . \end{cases}
$$
Checking that the $\sigma$-action is compatible and that the definitions above are correct, is routine.

By Lemma~\ref{lemm.easy}(\ref{easy.compatible.monotone}) the $\sigma$-action is monotone.
The corollary follows by Proposition~\ref{prop.nom.complete}. 
\end{proof}

\begin{rmrk}
\label{rmrk.argue}
By Proposition~\ref{prop.standard.nom.bool} every Tarski-style model is a FOLeq algebra.
However, not every FOLeq algebra is a Tarski-style model, because by Proposition~\ref{prop.standard.nom.bool} Tarski-style models have limits for all, possibly infinite, finitely supported subsets.

In FOLeq algebras 
the notion of completeness used is \emph{fresh-finite} completeness (Subsection~\ref{subsect.fresh-finite.limit}); a concept which is natural to express in a nominal universe. 
This makes FOLeq algebras complete enough to interpret $\tand$ and $\tall$, and no more.

Thus, Tarski-style models have more limits than first-order logic requires.
This is not detectable from inside first-order logic, since first-order logic is sound and complete for the Tarski-style models and this is why Tarski-style models suffice to model first-order logic in a ZF universe.

However, completeness is not the only issue: if it were, we would only ever need the Herbrand model of syntax quotiented by derivable equivalence (Section~\ref{sect.herbrand}).
The issue is to capture in abstract semantic terms exactly that structure necessary to interpret first-order logic.
Neither the Tarski models nor indeed the Herbrand models quite do this; but FOLeq algebras do.
In this sense, Tarski-style models are less natural than the nominal semantics proposed in this paper.
\end{rmrk}

Before we are done, we need to build an interpretation of term-formers and predicate-formers in the sense of Definition~\ref{defn.interp}, using the interpretation in `ordinary' sets from Definition~\ref{defn.model.standard}.
This is not hard: 
\begin{defn}
\label{defn.functional.model}
Suppose $\mathcal N$ is an ordinary model (Definition~\ref{defn.model.standard}), so that by Proposition~\ref{prop.standard.nom.bool} we have that $\Tarski{|\mathcal N|,\{\lbot,\ltop\}}$ is a FOLeq algebra.

Define an interpretation $\ns I$ over  $\Tarski{|\mathcal N|,\{\lbot,\ltop\}}$ in the sense of Definition~\ref{defn.interp} as follows:
$$
\begin{array}{r@{\ }l}
\tf f^\iden(a_1,\dots,a_{\ar(\tf f)})(\varsigma)=&\tf f^\nden(\varsigma(a_1),\dots,\varsigma(a_{\ar(\tf f)}))
\\
\tf P^\iden(a_1,\dots,a_{\ar(\tf P)})(\varsigma)=&\tf P^\nden(\varsigma(a_1),\dots,\varsigma(a_{\ar(\tf P)}))
\end{array}
$$
where the $a_i$ are distinct atoms (the $a_i$ above left are short for $\tf{atm}_{\Tarski{X,X}}(a_i)$ from Definition~\ref{defn.idenot.beta}).
\end{defn}

We obtain an interpretation of first-order logic immediately from Definition~\ref{defn.interp.I} and soundness immediately from Theorem~\ref{thrm.fol.sound}.
If we like, we can also leverage the completeness theorem with respect to ordinary models:
\begin{corr}
\label{corr.nom.complete}
If $\Phi\not\cent\Psi$, then there exists a termlike $\sigma$-algebra $\ns U$ and a FOLeq algebra $\ns M$ over $\ns U$ such that $\bigwedge_{\phi\in\Phi}\idenot{\phi}\not\leq\bigvee_{\psi\in\Psi}\idenot{\psi}$.
\end{corr}
\begin{proof}
By completeness of first-order logic there exists some ordinary model $\mathcal N$ (Definition~\ref{defn.model.standard}) and valuation $\varsigma$ to $\mathcal N$ such that 
$\bigwedge_{\phi\in\Phi}\ndenot{\varsigma}{\phi}=\ltop$ and $\bigvee_{\hspace{-.6ex}\psi\in\Psi}\ndenot{\varsigma}{\psi}=\lbot$.

So consider $\Tarski{|\mathcal N|,\{\lbot,\ltop\}}$ and consider the interpretation $\ns I$ from Definition~\ref{defn.functional.model}.
It follows from Definition~\ref{defn.leq.nsN} that 
$\bigwedge_{\phi\in\Phi}\idenot{\phi}\not\leq\bigvee_{\hspace{-.6ex}\psi\in\Psi}\idenot{\psi}$.
\end{proof}

\section{Herbrand models (Lindenbaum-Tarski algebras)}
\label{sect.herbrand}

We conclude, briefly, by observing that predicate syntax quotiented by derivable equality is a FOLeq algebra.
This construction is variously called a \emph{Herbrand}, \emph{Lindenbaum}, or \emph{Lindenbaum-Tarski} construction, algebra, or model---or just `syntax quotiented by derivable equivalence'.
 
There is not intended to be much new in this observation; just to show how syntax quotiented by derivable equivalence fits with the idea of a FOLeq algebra.

If there is any subtlety it is that we build our model from possibly open syntax, and the familiar syntactic notion of variable will be handled in our nominal abstract machinery as a special case of the background Fraenkel-Mostowski notion of atom.
It all works perfectly, which is part of the point. 

Recall the syntax of terms and predicates from Definition~\ref{defn.terms.and.predicates}, and the notion of logical entailment from Figure~\ref{fig.FOL}.

\begin{defn}
Define a \deffont{logical equivalence} relation $\sim$ on predicates by 
$$
\phi\sim\phi'
\quad\text{when}\quad
\phi\cent\phi'\ \land \phi'\cent\phi .
$$
Write $[\phi]_\sim$ for the $\sim$-equivalence class of $\phi$. 

If $x$ and $x'$ are $\sim$-equivalence classes then impose a partial order by $x\leq x'$ when $\phi\in x$ and $\phi'\in x'$ and $\phi\cent\phi'$.
It is a fact that this is well-defined (does not depend on the choice of $\phi$ and $\phi'$).
\end{defn}

Give predicates the natural permutation action where $\pi$ acts on $\phi$ by acting on the atoms in $\phi$.
Theorem~\ref{thrm.equivar} tells us the following:
\begin{lemm}
If $\phi\sim\phi'$ then $\pi\act\phi\sim\pi\act\phi'$.
\end{lemm}

\begin{defn}
Give $\sim$-equivalence classes the \deffont{pointwise} permutation action by $\pi\act[\phi]_\sim=[\pi\act\phi]_\sim$.
\end{defn}

\begin{defn}
Give $\sim$-equivalence classes a \deffont{pointwise} $\sigma$-action by $[\phi]_\sim[a\sm u]=[\phi[a\sm u]]_\sim$.
\end{defn}

We need to check that $[\phi]_\sim[a\sm u]$ is well-defined, that is, if $\phi\sim\phi'$ then $\phi[a\sm u]\sim\phi'[a\sm u]$.
This is a fact of first-order logic.
 
\begin{defn}
Define $\ns U$ to be the termlike $\sigma$-algebra of terms where substitution is real substitution on syntax, and take $\mathcal L=(|\mathcal L|,\act,\leq,\ns U,\tf{sub}_{\ns U})$ to be $\sim$-equivalence classes of predicates, with the pointwise permutation and $\sigma$-actions.
\end{defn}

\begin{thrm}
\label{thrm.L.FOLeq}
$\mathcal L$ is a FOLeq algebra if we take $\ltop$ to be $[\ttop]_\sim$, $[\phi]_\sim\land[\phi']_\sim$ to be $[\phi\tand\phi']_\sim$, $\freshwedge{a}[\phi]_\sim$ to be $[\tall a.\phi]_\sim$, and $r{=}s$ to be $[r\teq s]_\sim$.
\end{thrm}
\begin{proof}
Well-definedness and limit properties are all just properties of first-order logic.
Finite support is from Theorem~\ref{thrm.no.increase.of.supp}, since syntax is finite.
\end{proof}

\section{Conclusions}

\subsection{Semantics out of context}

Traditionally we admit numbers as a primitive datatype, and we emulate names using numbers (since numbers are countably infinite) or functional arguments.
This is neatly packaged up by Tarski-style semantics, which broadly speaking has the following shape:
$$
(\text{Natural numbers}\Rightarrow\text{Semantics of terms})\Rightarrow\text{Semantics of predicates}
$$
Tarski-style semantics have exerted a powerful, almost subliminal, influence on how semantics have been designed, and how the relationship between syntax and semantics has been understood.
Proposing a good alternative to this view is one of our main goals. 

In Tarski-style semantics the valuation (mapping numbers to semantics of terms as illustrated above) is a \emph{context} of variable-to-denotation assignments.
The semantics of a predicate $\phi$ exists \emph{in the context} of a valuation $\varsigma$, and we write $\model{\phi}_\varsigma$.
$\phi$ provides the syntax and $\varsigma$ provides the context, and what Tarski really taught us was that semantics exists \emph{in context}.

Yet, from our point of view context is only compensating for a mathematical foundation that is too poor to directly represent variables.
We would consider the slogan \emph{``there is no such thing as a free variable''} from \cite{perlis:epiop} to be succinct, eloquent---and incorrect.

In this paper we do things differently.
We interpret names more-or-less as themselves, using urelemente in Fraenkel-Mostowski foundations.
Denotationally speaking there \emph{is} such a thing as a free variable, and it is an \emph{urelement}---an \emph{atom}. 

In this paper we apply this idea to give semantics to first-order logic.
Given a set, its powerset is naturally a Boolean algebra: in a nutshell, this paper is notes that given a set of terms, the $\sigma$-powerset of its $\amgis$-powerset is naturally a model of first-order logic with equality (see Example~\ref{xmpl.approx})---and furthermore, this concrete powersets model is susceptible to simple and attractive nominal algebraic and lattice characterisations (and a topological treatment is possible too, though not in this paper). 
More on this in Subsection~\ref{subsect.trio}.

The Tarski notion of variable context is transmuted into the nominal notion of support (Definition~\ref{defn.supp}).\footnote{%
In \cite{gabbay:nomhss} are nominal Henkin-style semantics for higher-order logic in a typing context, but not a valuation context.  We calculated denotation in a context of static (typing) information, but we were not forced to \emph{also} use an explicit context of valuations---and we did not.} 
In the terminology of the title of this paper, the semantics we obtain is \emph{absolute}; we talk about \emph{the} semantics of $\phi$, and write $\model{\phi}$.
Once the model is fixed, so is $\model{\phi}$ the meaning of $\phi$.

We can recover the minimal relevant context of $\model{\phi}$ by calculating its support $\supp(\model{\phi})$.
However, if we do not care about support we can ignore it.
This simplifies reasoning; it is easier to talk about `an element' than `an element in context'.

But support is more general than `the free variables of'.
It can be read as expressing what a nominal element $x$ depends on---without needing to know what logic, program, or computation $x$ represents (if any).
If we use our foundations correctly then the support of $x$ will relate sensibly to what $x$ is being used to represent, and furthermore, the general nominal constructs will give us a useful structure within which to work.
For instance in this paper:
\begin{itemize*}
\item
Lemma~\ref{lemm.supp.lemma} notes that $\supp(\model{\phi})\subseteq\fa(\phi)$---the nominal element $\model{\phi}$ is supported by (at most) the atoms free in $\phi$.
This connects support of syntax with the corresponding notion of support of semantics.
\item
Similarly \rulefont{\sigma\#} from Figure~\ref{fig.nom.sigma} expresses that if $a{\not\in}\supp(x)$ then substitution $[a\sm u]$ does not change $x$.
This connects nominal support---being affected by a permutation like $(b\ a)$---with the notion of semantic dependency that comes from being in a $\sigma$-algebra and being affected by $[a\sm u]$.
\end{itemize*}
The two points above are visible in syntax: if $a$ is not free in $\phi$ then $\phi$ with $a$ substituted for a term $r$ is equal to $\phi$.
By the two points above, if $a$ is not free in $\phi$ then also $\model{\phi}[a\sm u]=\model{\phi}$.
In fact, the semantics is `compositional' with respect to substitution; $\model{\phi[a\sm r]}=\model{\phi}[a\sm\model{r}]$ (Lemma~\ref{lemm.sub.commute}).
The reader used to Tarski semantics might recognise the corresponding result $\model{\phi[a\sm r]}_\varsigma=\model{\phi}_{\varsigma[a\sm\model{r}_\varsigma]}$.
From the point of view of this paper, that result is a special case of a more general result for a more general semantics.

Because an element $x$ stands on its own, we do not need to label everything with a context, or think or prove lemmas about weakening or strengthening contexts, or about making the contexts of $x$ and $y$ equal so that we can combine them, or even about $\alpha$-renaming contexts.
We do not need valuations, or higher types (though we can have them if we like, as noted in Section~\ref{sect.complete}).
This is taken care of in the background.

Name-management is simpler in nominal techniques, and name-management is \emph{the} thing that separates first-order logic from Boolean algebra, and it should be clear that things other than first-order logic would be susceptible to the mathematical style and nominal semantics we have introduced in this paper; some are suggested in the \emph{Future work} below.

\subsection{Recalling the trio of semantics}
\label{subsect.trio}

\noindent In this paper we have considered a trio of semantics: lattices, sets, and algebra.\footnote{A pure nominal treatment of syntax-with-binding, rather than its semantics, is also possible.  Indeed, that problem is what motivated the initial developments behind nominal techniques \cite{gabbay:thesis,gabbay:newaas-jv,gabbay:fountl}.}
Let us recall what they look like:
\begin{itemize*}
\item
A lattice-flavoured semantics specifies the meaning of a logic in terms of a partially-ordered set (a set with a transitive reflexive antisymmetric relation).
The order relation models logical entailment; that is, the semantics is designed such that $\phi\cent\psi$ should imply $\model{\phi}\leq\model{\psi}$.

We model logical connectives as limits; $\ttop$ is modelled as a top element (a greatest lower bound for nothing), $\tand$ is modelled as a greatest lower bound for two elements.
This is standard, and in categorical language we say that we assume \emph{finite limits} of the lattice.

The new idea here is to model $\tall a$ using fresh-finite limits $\freshwedge{a}$ (Definition~\ref{defn.nom.poset}; this was proposed in \cite{gabbay:nomspl}).

So by this semantics, $\model{\phi}$ is an element of a poset with certain \emph{fresh-}finite limits.
Instead of finite limits, we assume \emph{fresh-finite limits}.
\item
The nominal algebraic semantics is an algebra not over sets, but over \emph{nominal} sets. 

Nominal algebra enriches the syntax, judgement form, and semantics of universal algebra, while retaining (somewhat surprisingly) a purely equational flavour \cite{gabbay:nomuae,gabbay:nomahs,gabbay:nomtnl,gabbay:uniner}. 

Nominal algebra, the nominal algebra axiomatisation of substitution (called $\sigma$-algebras here), and that of first-order logic, were proposed in a sequence of papers \cite{gabbay:capasn,gabbay:oneaah,gabbay:capasn-jv,gabbay:oneaah-jv,gabbay:stodfo,gabbay:pernl-jv}.
$\amgis$-algebras are more recent and were created specifically to prove completeness/duality results.

By this semantics, $\model{\phi}$ is an element of a nominal set which happens to be equipped with operations satisfying certain nominal equalities (see Appendix~\ref{subsect.foleq.alg}).
\item
The sets semantics builds on the semantics of \cite{gabbay:stodnb,gabbay:stodfo}.\footnote{\emph{Note:} The logic of \cite{gabbay:stodfo} did not include equality and considers duality instead of completeness.}

Conjunction is interpreted as sets intersection and negation as sets complement. 
This is again standard.

We see universal quantification interpreted as an infinite intersection in Definition~\ref{defn.nu.U}, and Proposition~\ref{prop.char.freshwedge} notes that this is equal to a fresh-finite limit.
Proposition~\ref{prop.char.freshwedge} is one of a family of characterisation results\footnote{E.g. $\{x'{\in}|\mathcal L| \mid \supp(x'){\subseteq}\supp(x){\setminus}\{a\}\wedge x'{\leq} x\}$ and $\bigwedge_{n{\in}\mathbb A} x[a\sm n]$.} 
and an algebraic one is in Appendix~\ref{subsect.foleq.alg}.

So by this semantics, $\model{\phi}$ is a set which thanks to its $\sigma$ and $\amgis$ structure turns out to be able to interpret universal quantification, and as we saw in Subsection~\ref{subsect.powsigma.equality} also equality.
We can write: $\sigma$ + $\amgis$ + powersets = First-Order Logic with equality. 
\end{itemize*}
The Stone duality theorem for Boolean algebras adds a fourth framework: topology.
A topological duality result for FOL (but not FOLeq) algebras exists and uses technology similar to some of the material in this paper; see a sister paper \cite{gabbay:stodfo}.

\subsection{Related work}

\subsubsection{Nominal precedents}

The idea of nominal axiomatisations of mathematics dates back to \emph{nominal algebra} which was used to axiomatise in the first instance substitution and first-order logic \cite{gabbay:capasn,gabbay:oneaah,gabbay:capasn-jv,gabbay:oneaah-jv,gabbay:nomuae}.

The most developed such logic is at the time of writing \emph{permissive-nominal logic} \cite{gabbay:pernl,gabbay:pernl-jv} which extends nominal algebra with quantification (of unknowns: $\forall X$).
We use it to finitely axiomatise first-order logic and arithmetic in a first-order setting.

Concerning the sets semantics we see in this paper, precursors are \cite{gabbay:subfmf,gabbay:stusun}, which technically have little to do with this paper but which do build families of sets representations of what in this paper we call a $\sigma$-action, and \cite{gabbay:stodfo} which as mentioned above is a sister paper considering topology and with a slightly different technical treatment (different version of $\amgis$-algebras, no equality, no lattices, and so on; at the time of writing, it has not yet been published).

More distantly related are the term equational systems of \cite{fiore:teresl,fiore:secoel}.
These have no concrete sets representations and the method of axiomatisation is quite different---so technically, this work is not very relevant to what we do in this paper---\emph{but} it is based on presheaf semantics, and nominal sets admit a presheaf presentation (as noted in Subsection~\ref{subsubsect.context.back.in}), so there is some overlap.

\subsubsection{Tarski-style valuations and Herbrand-style quotiented syntax}

We consider Tarski-style valuation-based semantics in Section~\ref{sect.complete}, and predicates up to derivable equivalence in Section~\ref{sect.herbrand}.
We showed how both can be considered subclasses of FOLeq algebras.

As usual the Herbrand semantics is easy to construct but uninformative; building it is more a `sanity check'.

Concerning the valuation semantics, we noted in Remark~\ref{rmrk.argue} that we get more limits than needed.
We only really need fresh-finite limits to interpret first-order logic, but valuation semantics will not let us express that.
In contrast, FOLeq algebras have only exactly those limits necessary to interpret first-order logic, which seems a significant point in their favour.

\subsubsection{Combinators}

Variables are not indispensable.
This is the idea behind combinatory logic \cite{schonfinkel:ubebml-book,curry:cli}, and it can be pushed a long way \cite{barendregt:comtsi,tarski:forstw} (the latter paper formalises set theory without using variables).

We see combinatory techniques as orthogonal to this paper.
Yes, languages can be constructed without variables; but that does not make variables obsolete.
Variables are useful because humans like to use them; combinators do not change that.

Sometimes people talk about combinators `reducing' variables to a simpler problem.
Not so.
With names (nominal or otherwise) we can talk about a location independently of its site of binding.
With combinators we cannot talk about a location except by exhibiting it as being at the other end of a combinator.  
These may be able to express the same computable functions, but that does not make them equivalent.

\subsubsection{Cylindric and polyadic algebras}

Cylindric and polyadic algebras treat quantification as a modality satisfying certain axioms.
To model languages with infinitely many variables, we must admit infinitely many modalities.

The interested reader is referred to \cite[Chapter~8, page~243]{halmos:algl} (whose exposition is incredibly clear) and to \cite{monk:intcsa} and \cite{henkin:cyla}.

The resemblance to parts of this paper are interesting; for instance, the infinite schemata of axioms \rulefont{P1} to \rulefont{P7} of a polyadic algebra in \cite[Chapter~8, page~244]{halmos:algl} are similar to projections of the axioms of Subsection~\ref{subsect.foleq.alg} into a first-order setting (where there is a symbol for every $\pi$ and $\hnu a$)---except, that polyadic and cylindric algebras are based on the monoid of finitely supported atoms-substitutions (so: not necessarily bijective functions like permutations), and also, these systems exist in ZF so do not always insist on finite support, or in their terminology on \emph{local finite dimensionality}.

(Tangentially we can also mention Fine's notion of \emph{arbitrary objects} \cite{fine:reaao}, not because it has anything to do with algebra or first-order logic as such, but just because it too is based on a monoid of substitutions.)

Cylindric and polyadic algebras, and Fine's arbitrary objects, are distinctive because they do \emph{not} belong to the Tarski family of variables-with-valuations and denotations-in-context.
In that sense they are closer to the spirit of this paper.

However, we are not aware of anything like the mathematics of this paper being executed in the frameworks above.
This is partly because we take permutations as primitive and so can capture finitely what it takes infinite axiom schemes to capture otherwise.\footnote{Polyadic and cylindric algebras are $\epsilon$ away from us here, but by using \emph{monoids} instead of \emph{groups} they cannot guarantee that atoms that were distinct, remain distinct.

In this light, our \emph{permutative convention} from Definition~\ref{defn.atoms}, which goes back to \cite{gabbay:capasn,gabbay:oneaah,gabbay:oneaah-jv}, that in informal mathematical discourse in nominal techniques $a$, $b$, and $c$ range over \emph{distinct} atoms, seems very important.} 
It is also because the nominal framework enforces symmetry (as discussed in the opening to Subsection~\ref{subsect.pre-equivar}) and finite support.
These are all things that somehow fit together to make possible what we have done in this paper. 
Perhaps the research above \emph{could not} do what we have done in this paper, because of the influence of the implicit Zermelo-Fraenkel foundation. 

We should also mention \emph{hyperdoctrines}.
These are a categorical framework within which to consider semantics for logics and are very general, so we briefly sketch how hyperdoctrines work for the specific case of a one-sorted first-order classical logic.\footnote{\dots which is what we axiomatise in this paper. We hope it is quite clear by now how a multi-sorted logic, or an intuitionistic logic, would work in our nominal semantics.}

Consider a category $\theory T$ whose objects are natural numbers $0$, $1$, \dots and for each term-former $\tf f$ of arity $n$ an arrow $\tf f^{\theory T}$ from $n$ to $0$.
A hyperdoctrine is a functor $F$ from $\theory T$ to the category of Boolean algebras, along with an assignment to each $\tf P$ of arity $n$ of an element $\tf P^F$ in $F(n)$.

In spite of the categorical generality in which this idea is phrased, it really just re-states Tarski-style semantics.
Yes, there are no valuations as such, but this is replaced by a finite context and the categorical framework enforces well-formedness (we only give a predicate semantics if its free variables are in the context).
Useful as hyperdoctrines are for the reader wanting to work entirely within the vocabulary of categories, all they do is translate---and update---the ideas of Tarski into a new categorical language.
The basic nature of the thing has not changed.
Nominal techniques as applied in this paper are doing something different.

\subsubsection{Boolean algebras with operators}

A \emph{Boolean algebra with operators} is a Boolean algebra $\ns B$ equipped with functions $|\ns B|^n\to|\ns B|$ which preserve intersections on each component.
These have been studied extensively; first by Jonnson and Tarski \cite{jonnson:booao}, and then by Goldblatt \cite{goldblatt:varca}.
A concise and very readable survey is in the Introduction of \cite{haim:dualwo}.

It is oversimplistic but reasonable to characterise FOLeq algebras as just being fancy Boolean algebras with operators.
The operator concerned is not $\hnu$ but the $\sigma$-algebra structure $[a\sm u]$---which since we are oversimplifying we might as well call \emph{substitution}, though it is not necessarily syntactic. 

Viewed as an operator, the substitution $[a\sm u]$ is a modality (a unary operator) with special properties: it commutes with meets $\land$, joins $\lor$, and negation $\lneg$.
Thus, $[a\sm u]$ can be viewed as both a modal Box and a modal Diamond operator, which is also a homomorphism of Boolean algebras, and the corresponding Kripke semantics is functional in the sense that each world has precisely one future world. 

The flavour of our representation theorem is in keeping with this.
Notably, the pointwise actions of Definitions~\ref{defn.p.action} and~\ref{defn.sub.sets}
$$
X[a\sm u]=\{p\mid \New{c}p[u\ms c]\in (c\ a)\act X\}
\qquad\text{and}\qquad
p[u\ms a]=\{x\mid x[a\sm u]\in X\}
$$
are like the functional preimage under the relevant accessibility relation---though not quite, as we see in the use of $\new$ on the left above, which hard-wires capture-avoidance and $\alpha$-equivalence (Lemma~\ref{lemm.sigma.alpha}).

Still, from the point of view of \cite{jonnson:booao,goldblatt:varca} it is not unreasonable to view this paper as being `just' about a particular kind of Boolean algebra with a particular family of operators satisfying certain axioms described in Figure~\ref{fig.nom.sigma}.
We mention this because it is not a bad way of viewing how the material of this paper fits in to the broader mathematical context.
 
Of course there is more going on than \emph{just} that. 
Notably: we are working in FM sets so we have to make sure that our constructions preserve finite support; also we do not just have $[a\sm u]$ but also the universal quantifier $\hnu$; but perhaps most interestingly, the axioms for $\sigma$ are determined by properties of substitution---but the axioms of $\amgis$ are \emph{not} so determined.
They emerge from the specific requirements of the specific proofs we need to carry out.
So even if a FOLeq algebra is just a Boolean algebra with operators, the operators are organised and axiomatised into a specific and non-trivial structure, and how that structure gets `inverted' under a modified capture-avoiding functional preimage operation, to form $\amgis$-algebras, is subtle and certainly not obvious.

\subsubsection{Varieties of FOLeq algebras}

In Appendix~\ref{sect.alg} we give a purely nominal equational axiomatisation of FOLeq algebras, using \emph{nominal algebra} \cite{gabbay:nomuae}.
The syntax of nominal algebra is based on equalities between nominal terms subject to syntactic freshness side-conditions which behave a bit like typing conditions.
These do not impact on the algebraic flavour of the logic and one way of making that formal is with a version of the HSP theorem, or Birkhoff's theorem.
We get a notion of \emph{variety} suitable for nominal sets, and a theorem that a nominal algebra theory is precisely characterised by the variety of its models, for a suitable nominal notion of variety.
For more details see \cite{gabbay:nomuae,gabbay:nomahs}.

\subsubsection{The Schanuel topos: putting context back in, if we want it}
\label{subsubsect.context.back.in}

The reader may know that nominal sets are equivalent to the Schanuel Topos, which can be presented as pullback-preserving presheaves (for a proof with calculations see \cite[Theorem~9.14]{gabbay:fountl}).

What this means in plain English is that every nominal $x{\in}|\ns X|$ can be thought of as a family of $A\cent x$ for $A\subseteq\mathbb A$ such that $\supp(x)\subseteq A$.
In other words, $x$ `exists' in a context if and only if the context contains the support of $x$.

So in fact, everything we have done in this paper, admits a `contextual' presentation.
Furthermore, it may admit generalisations to presheaves not necessarily preserving pullbacks (to be more precise, not necessarily preserving pullbacks of \emph{monos} \cite{gabbay:nomrs}).

And yet, even if for the sake of argument we imagine such a generalisation is found and written, the proofs in this paper would be harder to discover, harder to present, and harder to disseminate, if we did not first have the simpler `context-free' sets-based reasoning of nominal sets.

\subsection{Future work}

\begin{enumerate}
\item
\emph{Apply the sets semantics to the $\lambda$-calculus.}\quad
A nominal algebra axiomatisation of the $\lambda$-calculus is in \cite{gabbay:lamcna,gabbay:nomalc}, building on the nominal rewrite systems for the $\lambda$-calculus from \cite{gabbay:nomr,gabbay:nomr-jv}.
A sequent logic style presentation of the $\lambda$-calculus is in \cite{gabbay:simcks}, along with a (non-nominal) Kripke style semantics.

In a sequel to the current paper, we build on this work to give for the first time a topological duality result for the $\lambda$-calculus \cite{gabbay:repdul}.   
\item
\emph{Generalise the language beyond first-order logic.}\quad
This paper has only scratched the surface of what might be possible with our semantics.
First-order logic has binding, but it does not have binding term-formers nor does it have the ability to reason directly on variable names.
We do this all the time in informal mathematical practice---e.g. when we specify first-order logic itself.

It is a tricky question what kind of logic is being used here, though it seems to be a nominal one. 
Permissive-nominal logic is one such \cite{gabbay:pernl,gabbay:pernl-jv}; it enriches first-order logic with binding term-formers and we have used it to axiomatise first-order logic and arithmetic \cite{gabbay:pernl-jv}. 

However, there may be much more to say here.
We hope that the sets-based semantics of this paper, which is new, will be a good guide to us here.
\item
\emph{Category theory in nominal sets.}\quad
The fresh-finite limit from Subsection~\ref{subsect.fresh-finite.limit} is natural if we build category theory internally to nominal sets; because a category is a set of objects and a set of arrows, it is natural in a nominal context to talk about `the limit object satisfying a freshness side-condition'.
Our treatment of posets in this paper begs generalisation to categories, and we look forward to finding out what story might be told here.\footnote{There is a general theory of categories in universes other than ZF sets \cite{kelly:bascec}.
But we are not just enriching the notion of category; we are also enriching the notion of \emph{limit}.
So whatever mathematics emerges from that seems unlikely to be just a special case of the general machinery.}
A hint of this is already in the paper; freshness conditions and adjoints (especially the adjoint explanation of quantifiers) are clearly related.
This must be the topic of a later paper. 

Another interesting generalisation is to take the notion of $\sigma$-algebra seriously as a categorical framework.
The clue to this is to view a category as a kind of $\sigma$-algebra in which each arrow has one `free variable' (its source).
In a similar spirit, we might consider a (termlike) $\sigma$-algebra as a kind of semigroup in which we consider $x[a\sm y]$ as composition of $x$ and $y$ `at $a$'.
If we do this, it will not just be because we can: in computing, composition often happens at a location, in some sense.
\item
\emph{Stone duality.}\quad
The representation theorem of this paper uses maximally consistent filters (which we call \emph{points}), for a suitable notion of filter (Definitions~\ref{defn.filter} and~\ref{defn.points}).

Our logic is classical, so maximally consistent filters do indeed coincide with ultrafilters; see Lemma~\ref{lemm.prime.ultra}.

We can develop this further by characterising the right notion of topology on a (nominal) topological space and obtain a Stone duality result for FOLeq algebras.

This has been studied for the simpler case of first-order logic without equality.
See \cite{gabbay:stodfo} (a sister paper to this one) and in particular see Definition~6.18 of \cite{gabbay:stodfo} where the topology on points is described.
See also \cite{forssell:firold}, which uses different techniques but has similar goals.

Adapting \cite{gabbay:stodfo} to the definitions and results here would probably require a distinct paper. 
\item
\emph{The many-sorted case.}\quad
Our definition of termlike $\sigma$-algebra in Definition~\ref{defn.term.sub.alg} is single-sorted; $\ns U$ is a single set with a $\sigma$-action over itself.
For applications in computer science it would be useful to have a many-sorted version of FOLeq algebras.

This should be quite easy.
It would suffice to generalise Definition~\ref{defn.term.sub.alg}: instead of $\ns U$ we have a type-indexed family $\ns U_\tau$ along with atoms for each type and a $\sigma$-action of type $\ns U_\tau\times\mathbb A_{\tau'}\times\ns U_{\tau'}\to\ns U_\tau$.
The rest of the mathematics in this paper would be orthogonal and should remain unaffected.
\end{enumerate}

\subsection{Summary}

We hope that the mathematics of this paper supports the following two arguments:
\begin{itemize*}
\item
The correct notion of model of first-order logic is a FOLeq algebra; this contains just what is necessary to model first-order logic and furthermore it supports good sets-based and poset-based characterisations.
Underlying these are $\sigma$- and $\amgis$-powersets, which are remarkable structures, and nominal algebra (itself a logic of independent interest).
To express these we need nominal techniques---ZF sets and algebras are not expressive enough.
\item
Nominal foundations give other new, and not necessarily obvious, opportunities in logic and semantics.
We study first-order logic in this paper and the untyped $\lambda$-calculus in \cite{gabbay:repdul}.
We hope and suspect that the ideas behind these examples, and many of the tools we had to develop to carry out these examples, will have further uses (cf. the \emph{Future work} above).
\end{itemize*}
So this paper gives non-trivial answers to some specific technical questions, but stepping back to look at the larger picture, we see this work as one example of an exciting new way in which to apply nominal techniques in logic and semantics.


\newcommand{\etalchar}[1]{$^{#1}$}
\hyphenation{Mathe-ma-ti-sche}
\providecommand{\bysame}{\leavevmode\hbox to3em{\hrulefill}\thinspace}
\providecommand{\MR}{\relax\ifhmode\unskip\space\fi MR }
\providecommand{\MRhref}[2]{%
  \href{http://www.ams.org/mathscinet-getitem?mr=#1}{#2}
}
\providecommand{\href}[2]{#2}

\appendix

\section{Nominal algebraic axiomatisation of fresh-finite limits and equality}
\label{sect.alg}

Definition~\ref{defn.nom.poset} set up the notion of a fresh-finite limit $\freshwedge{A}X$.
Proposition~\ref{prop.ffc.char} split this notion up into three parts which are sometimes more convenient for proofs: $\ltop$, $\land$, and $\freshwedge{a}$.
These definitions are `poset-flavoured'; that is, we talk about $\leq$ and greatest lower bounds for $\leq$.

We now set about axiomatising this structure in \emph{nominal algebra} \cite{gabbay:nomuae}. 
That is, we will characterise FOLeq algebras in terms of an underlying nominal set and functions on that set satisfying nominal equalities.

It is convenient to do this in three stages: bounded meet-semilattices with $\hnu$, then bounded lattices with $\hnu$, then FOLeq algebras.

\subsection{Bounded meet-semilattices with $\protect\hnu$}

\begin{defn}
A \deffont{bounded meet-semilattice} in nominal sets is a tuple $\ns M=(|\ns M|,\act,\land,\ltop)$ where $(|\ns M|,\act)$ is a nominal set and $\ltop{\in}|\ns M|$ is an equivariant \deffont{top element} and $\land:(\ns M\times\ns M)\Func\ns M$ is an equivariant function, such that $\land$ and $\ltop$ form an idempotent monoid:
$$
(x\land y)\land z=x\land(y\land z)
\qquad
x\land y=y\land x
\qquad
x\land x=x
\qquad
x\land \ltop=x
$$
Here, $x,y,z$ range over elements of $|\ns M|$.
\end{defn} 

\begin{defn}
\label{defn.meet.semi.poset}
As standard, a meet-semilattice acquires a partial order by setting $x\leq y$ when $x\land y=x$.
A bounded meet-semilattice yields a poset with a top element ($\ltop$).
\end{defn}

\begin{defn}
\label{defn.hnu.M}
A \deffont{generalised universal quantifier} $\hnu$ on a bounded meet-semilattice $\ns M$ is an equivariant map $\hnu: (\mathbb A\times\ns M)\Func\ns M$ such that:
\begin{frameqn}
\begin{array}{l@{\qquad}l@{\quad}r@{\ }l}
\rulefont{\hnu\alpha}
&
b\#x\limp& \hnu b.(b\ a)\act x=&x
\\
\rulefont{\hnu{\land}}
&
& \hnu a.(x\land y)=&(\hnu a.x)\land(\hnu a.y)
\\
\rulefont{\hnu\#}
&
a\#x\limp& \hnu a.x=&x
\\
\rulefont{\hnu{\leq}}
&
&\hnu a.x\leq& x
\end{array}
\end{frameqn}
\end{defn}
(Recall that $\hnu a.x\leq x$ is shorthand for $(\hnu a.x)\land x=\hnu a.x$.)

\begin{lemm}
\label{lemm.hnu.leq}
If $x\leq y$ then $\hnu a.x\leq \hnu a.y$.
\end{lemm}
\begin{proof}
$x\leq y$ means $x\land y=x$.
We apply $\hnu a$ to both sides and use \rulefont{\hnu{\land}}.
\end{proof}

\begin{prop}
\label{prop.poset.structure}
Suppose $\ns M$ is a bounded meet-semilattice with $\hnu$.
Then with the partial order from Definition~\ref{defn.meet.semi.poset} we have the following:
\begin{itemize*}
\item
$\ltop$ is a top element.
\item
$x\land y$ is a limit for $\{x,y\}$.
\item
$\hnu a.x$ is an $a\#$limit for $\{x\}$.
\end{itemize*}
\end{prop}
\begin{proof}
The first two parts are standard; the interesting case is part~3. 

From \rulefont{\hnu\alpha} and part~3 of Corollary~\ref{corr.stuff} we derive that $a\#\hnu a.x$. 
With this and \rulefont{\hnu{\leq}} we have that $\hnu a.x$ is an $a\#$lower bound for $x$.

Now suppose $z\leq x$ and $a\#z$.
By Lemma~\ref{lemm.hnu.leq} and \rulefont{\hnu\#} we deduce $z\leq\hnu a.x$.
\end{proof}

\begin{corr}
Every bounded meet-semilattice with $\hnu$ can be viewed as a nominal poset with fresh-finite limits, and vice versa.
\end{corr}
\begin{proof}
Using the constructions of Propositions~\ref{prop.ffc.char} and~\ref{prop.poset.structure}.
\end{proof} 

Definition~\ref{defn.SemiAll.hom} and Proposition~\ref{prop.semi.iso} merely package the observations above in categorical language:

\begin{defn}
\label{defn.SemiAll.hom}
Given $\ns M$ and $\ns M'$ bounded meet-semilattices with $\hnu$, a \deffont{homomorphism} $F:\ns M\Func\ns M'$ is a function on the underlying nominal sets such that:
\begin{itemize*}
\item
$F(\pi\act x)=\pi\act F(x)$, meaning that $F$ is \emph{equivariant} (Definition~\ref{defn.equivariant}).
\item
$F(\ltop_{\ns M})=\ltop_{\ns M'}$ and $F(x\land_{\ns M} y)=F(x)\land_{\ns M'}F(y)$ and $F(\hnu_{\ns M} a.x)=\hnu_{\ns M'} a.x$.
\end{itemize*}
\label{defn.FFL.hom}
Given $\ns N$ and $\ns N'$ nominal posets with fresh-finite limits, a \deffont{homomorphism} $G:\ns N\Func\ns N'$ is a function on the underlying nominal sets such that:
\begin{itemize*}
\item
$G(\pi\act x)=\pi\act G(x)$ (so $G$ is \emph{equivariant}).
\item
$G(\freshwedge{A}X)=\freshwedge{A}\{G(x)\mid x\in X\}$.
\end{itemize*}
Write \theory{Semi\hnu} for the category of bounded meet-semilattices with $\hnu$ and homomorphisms between them, and \theory{Poset\freshwedge{A}} for the category of nominal posets with fresh-finite limits.
\end{defn}

\begin{prop}
\label{prop.semi.iso}
The natural functors mapping between \theory{Semi\hnu} and \theory{Poset\freshwedge{A}} define an isomorphism of categories.
\end{prop}

\subsection{Bounded lattices with $\protect\hnu$}

Our version of FOLeq is classical and has negation.
This means that it has $\land$, but by dualising with negation it also has $\lor$.
This has a minor effect on the natural axiomatisation of $\hnu$.
We therefore briefly sketch the case of $\hnu$ for bounded lattices, of which FOLeq is a special case.

\begin{defn}
\label{defn.bounded.lattice}
A \deffont{bounded lattice} in nominal sets is a tuple $\ns L=(|\ns L|,\act,\land,\lor,\ltop,\lbot)$ where $(|\ns L|,\act)$ is a nominal set which we may just write $\ns L$, and $\ltop,\lbot{\in}|\ns L|$ are equivariant \deffont{top} and \deffont{bottom} elements and $\land,\lor:(\ns L\times\ns L)\Func\ns L$ are equivariant functions, such that $\land$ and $\ltop$ form an idempotent monoid and so do $\lor$ and $\lbot$,
$$
\begin{array}{c@{\qquad}c@{\qquad}c@{\qquad}c}
(x\land y)\land z=x\land(y\land z)
&
x\land y=y\land x
&
x\land x=x
&
x\land \ltop=x
\\
(x\lor y)\lor z=x\lor(y\lor z)
&
x\lor y=y\lor x
&
x\lor x=x
&
x\lor \lbot=x
\end{array}
$$
and $\land$ and $\lor$ satisfy \emph{absorption}
$$
x\land(x\lor y)=x
\qquad\qquad
x\lor(x\land y)=x .
$$
Here, $x,y,z$ range over elements of $|\ns L|$.
\end{defn} 

Similarly to Definition~\ref{defn.meet.semi.poset}, a bounded lattice is a poset by taking $x\leq y$ to mean $x\land y=x$ or $x\lor y=y$ (the two conditions are provably equivalent). 

\begin{defn}
\label{defn.hnu.L}
A \deffont{generalised universal quantifier} $\hnu$ on a bounded lattice $\ns L$ is an equivariant map $\hnu: (\mathbb A\times\ns L)\Func\ns L$ satisfying the equalities in Figure~\ref{fig.genhnu}.
\end{defn}

\begin{figure}[tH]
$$
\begin{array}{l@{\qquad}l@{\quad}r@{\ }l}
\rulefont{\hnu\alpha}
&
b\#x\limp& \hnu b.(b\ a)\act x=&x
\\
\rulefont{\hnu{\land}}
&
& \hnu a.(x\land y)=&(\hnu a.x)\land(\hnu a.y)
\\
\rulefont{\hnu{\lor}}
&
a\#y\limp& \hnu a.(x\lor y)=&(\hnu a.x)\lor y
\\
\rulefont{\hnu{\leq}}
&
&\hnu a.x\leq& x
\end{array}
$$
\caption{Nominal algebra axioms for $\protect\hnu$ in a bounded lattice} 
\label{fig.genhnu}
\end{figure}

\begin{rmrk}
Comparing Definitions~\ref{defn.hnu.M} and Definition~\ref{defn.hnu.L}, we note that \rulefont{\hnu\#} has vanished, and that \rulefont{\hnu{\lor}} is not a perfect dual to \rulefont{\hnu{\land}}.

The disappearance of \rulefont{\hnu\#} is explained in Proposition~\ref{prop.bounded.simple}.

The lack of duality between \rulefont{\hnu{\land}} and \rulefont{\hnu{\lor}} is natural; \rulefont{\hnu{\leq}} orients $\hnu$ with respect to $\leq$ so it would be surprising (and wrong) if $\hnu$ displayed perfectly dual behaviour with respect to $\lor$.\footnote{This is to be contrasted with the axiomatisation of and Stone duality result for the $\new$ quantifier in \cite{gabbay:stodnb}. The $\new$ quantifier is self-dual, and \emph{is} symmetric in its interaction between $\land$ and $\lor$.} 
\end{rmrk}

\begin{prop}
\label{prop.bounded.simple}
Suppose $\ns L$ is a bounded lattice with $\hnu$ and $x{\in}|\ns L|$.
Then \rulefont{\hnu\#} from Definition~\ref{defn.hnu.M} follows from the other axioms of a generalised universal quantifier on a bounded lattice.

That is, $a\#x$ implies $\hnu a.x=x$ 
\end{prop}
\begin{proof}
We note that 
$
\hnu a.x=
\hnu a.(x\lor x)
\stackrel{\rulefont{\hnu{\lor}}}{=}
(\hnu a.x)\lor x
$
so that $x\leq\hnu a.x$.
Furthermore by \rulefont{\hnu{\leq}} $\hnu a.x\leq x$, and we are done.
\end{proof}

\begin{corr}
\label{corr.a.fresh.bounded}
Suppose $\ns L$ is a bounded lattice with a generalised universal quantifier $\hnu$.
Then $\hnu a.x$ is the $a\#$limit for $\{x\}$.
\end{corr}
\begin{proof}
By Proposition~\ref{prop.bounded.simple} \rulefont{\hnu\#} holds, so the restriction of $\ns L$ to $\ltop$,\ $\land$, and $\hnu$ forms a bounded meet-semilattice, and it is a fact that this has the same partial order.
We use Proposition~\ref{prop.poset.structure}.
\end{proof}

Axioms of the type presented in Definition~\ref{defn.hnu.L} underlie the nominal algebraic axiomatisation of first-order quantification in e.g. \cite{gabbay:oneaah-jv,gabbay:stodfo}.

\subsection{FOLeq algebras algebraically}
\label{subsect.foleq.alg}

We are now ready to give a full axiomatisation of a FOLeq algebra, in nominal algebra in the sense of the formal language described in \cite{gabbay:nomuae}.
This nominal algebra theory builds on the nominal algebra axiomatisations of substitution and first-order logic originally developed in \cite{gabbay:capasn,gabbay:oneaah,gabbay:capasn-jv,gabbay:oneaah-jv}.

\maketab{tab9}{@{\hspace{-3em}}R{14em}@{\ }L{14em}}
\begin{thrm}
\label{thrm.eq}
Suppose $\ns U$ is a termlike $\sigma$-algebra.
The FOLeq algebras over $\ns U$ (Definition~\ref{defn.FOLeq}) are precisely characterised as an underlying nonempty nominal set $\mathcal L$ along with
\begin{itemize*}
\item
an equivariant $\sigma$-action 
$\tf{sub}:(\mathcal  L\times\mathbb A\times\ns U)\Func\ns L$, written infix $x[a{\sm}u]$;
\item
an equivariant function $\land:\mathcal L\times\mathcal L\Func\mathcal L$;
\item
an equivariant function $\lneg:\mathcal L\Func\mathcal L$,
\item
an equivariant function $\hnu:(\atoms\times\mathcal L)\Func\mathcal L$, and
\item
an equivariant function $=^\lmathcal:(\ns U\times\ns U)\Func\mathcal L$,
\end{itemize*}
such that:
\begin{enumerate*}
\item
\label{item.sub.cont}
$\tf{sub}$ satisfies the axioms of a $\sigma$-action over $\ns U$ from Figure~\ref{fig.nom.sigma} and
\begin{tab9}
(x\land y)[a\sm u]=&(x[a\sm u])\land(y[a\sm u])
\\
(\lneg x)[a\sm u]=&\lneg(x[a\sm u])
\\
b\#u\limp\ \ (\hnu b.y)[a\sm u]=&\hnu b.(y[a\sm u]),
\\
(v'{=^\lmathcal}v)[a\sm u]=&(v'[a\sm u])=^\lmathcal(v[a\sm u]) 
\end{tab9}
\item
if we take $\ltop=\lneg(x\land\lneg x)$ (for any $x$) and $x\lor y=\lneg(\lneg x\land\lneg y)$ and $\lbot=\lneg\ltop$ then $\ltop$,\ $\land$,\ $\lbot$, and $\lor$ satisfy the axioms of a bounded lattice from Definition~\ref{defn.bounded.lattice},
\item
in addition\footnote{This axiomatisation is not minimal, but it is readable.  The interested reader is referred elsewhere \cite{mccune:shosab} for what can be achieved in terms of trading off the number of axioms against readability.}
\begin{tab9}
x\land(y\lor z)=&(x\land y)\lor(x\land z) 
\\
\lneg\lneg x=&x,
\end{tab9} 
\item
$\hnu$ satisfies the axioms of a generalised universal quantifier from Figure~\ref{fig.genhnu}, and finally
\item
\label{item.eq.cont}
$=^\lmathcal$ satisfies the axioms
\begin{tab9}
(u{=^\lmathcal}u)=&\ltop
\\
(u{=^\lmathcal}v)\land z[a\sm u] =& (u{=^\lmathcal}v)\land z[a\sm v]
\end{tab9}
\end{enumerate*}
\end{thrm}
\begin{proof}
We have all the pieces, we just need to put them together.

Suppose $\mathcal L$ is a FOLeq algebra.
We take $\ltop$ to be the top element, $x\land y$ to be the limit of $\{x,y\}$, $\hnu a.x$ to be $\freshwedge{a}x$ the $a\#$limit of $\{x\}$, $\lneg x$ to be the complement of $x$ in $\mathcal L$ (which is unique by Lemma~\ref{lemm.comp.unique}), and $u{=^\lmathcal}v$ to be $(a{=}b)[a\sm u,b\sm v]$ where ${=^\lmathcal}$ is the equality of $\mathcal L$ (the simultaneous $\sigma$-action is from Definition~\ref{defn.sim.sub.alpha}).
Then: 
\begin{itemize*}
\item
By assumption in Definition~\ref{defn.FOLeq} the $\sigma$-action satisfies the axioms of Figure~\ref{fig.nom.sigma}.
\item
By properties of limits and complements $\ltop$, $\land$, $\lbot$, and $\lor$ satisfy the axioms of a bounded lattice.
\item
By assumption $\mathcal L$ is distributive (Definition~\ref{defn.distrib}) and complemented, and it follows that $\land$ distributes over $\lor$ and $\lneg\lneg x=x$.
\item 
By Lemma~\ref{lemm.freshwedge.alpha} $\hnu$ satisfies \rulefont{\hnu\alpha} from Figure~\ref{fig.genhnu}.
By routine calculations on limits $\hnu$ satisfies \rulefont{\hnu\land}.\footnote{It is not quite that simple; we are using \emph{fresh-finite} limits.  But it all works out.}
By assumption $\mathcal L$ is distributive (Definition~\ref{defn.distrib}), so \rulefont{\hnu{\lor}} holds.
By the definition of a fresh-finite limit $\hnu a.x\leq x$.
\end{itemize*}

Now suppose $\mathcal L$ has all the structure of the statement of this theorem.
Definition~\ref{defn.FOLeq} defines a FOLeq algebra to be a finitely fresh-complete complemented nominal poset with a compatible $\sigma$-algebra structure and an equality. 
So we consider each property in turn:
\begin{itemize*}
\item
\emph{Finitely fresh-complete.}\quad
By standard arguments using the axioms $\ltop$ is a top element and $x\land y$ is a limit for $\{x,y\}$.
By Corollary~\ref{corr.a.fresh.bounded} $\hnu a.x$ is an $a\#$limit for $\{x\}$.
It follows by Proposition~\ref{prop.ffc.char} that $\mathcal L$ is finitely fresh-complete.
\item
\emph{Complements.}\quad
Again it follows from the axioms that $\lneg$ is a complement.
\item
\emph{Compatible $\sigma$-algebra structure.}\quad
By assumption $\tf{sub}$ satisfies the axioms from Figure~\ref{fig.nom.sigma}, so $\mathcal L$ has a $\sigma$-algebra structure of $\ns U$.
By our assumptions (condition~\ref{item.sub.cont} in the statement of this theorem) it is compatible.
\item
\emph{Equality.}\quad
It follows from our assumptions that $a{=^\lmathcal}b$ is an equality in the sense of Definition~\ref{defn.eq}.\footnote{Note that the statements are slightly different; Definition~\ref{defn.eq} assumes an element $a{=^\lmathcal}b$ whereas above we assume an equivariant function ${=^\lmathcal}$ and in condition~\ref{item.sub.cont} we assume that substitution distributes over it.}
\qedhere\end{itemize*}
\end{proof}

\section{On the Leibnitz equality}
\label{sect.more.equality}

\begin{rmrk}
In Remark~\ref{rmrk.leibnitz} we noted that Definition~\ref{defn.eq.powamgis} is a form of Leibnitz equality: $p\in(a=^{\ns P}b)$ when it `cannot distinguish between $a$ and $b$'.

Could it be that $(a=^{\ns P}b)$ is also equal to $\{p{\in}|\ns P|\mid (a\ b)\act p=p\}$?
After all, if $(a\ b)\act p=p$ then $p$ also `cannot distinguish $a$ and $b$'.

This notion is too weak.
To see this intuitively, consider that the unordered pair $\{a,b\}$ is invariant under swapping $a$ and $b$, but is not invariant under converting $b$ into $a$ to obtain $\{a\}$.
\end{rmrk}

We note that a strengthened version of Definition~\ref{defn.eq.powamgis} is possible:
\begin{lemm}
\label{lemm.eq.strong}
$p\in(u=^{\ns P}v)$ if and only if for \emph{all} atoms $c$,\ $p[u\ms c]=p[v\ms c]$.

In symbols: $(u=^{\ns P}v)=\{p{\in}|\ns P| \mid \Forall{c}p[u\ms c]=p[v\ms c]\}$.
\end{lemm}
\begin{proof}
Clearly if for \emph{all} atoms $c$,\ $p[u\ms c]=p[v\ms c]$, then for all but finitely many $c$,\ $p[u\ms c]=p[v\ms c]$.

Conversely, suppose $p[u\ms c]=p[v\ms c]$ for all but finitely many $c$, and consider any atom $a$ and element $x$.
Now choose one $c$ such that $c\#u,v,x$ and $p[u\ms c]=p[v\ms c]$ (this is always possible because all but finitely many $c$ satisfy each condition). 
By \rulefont{\sigma\alpha} $x[a\sm u]=((c\ a)\act x)[c\sm u]$ and $x[a\sm v]=((c\ a)\act x)[c\sm v]$ and it follows using Proposition~\ref{prop.sigma.iff} repeatedly that $x\in p[u\ms a]$ if and only $x\in p[v\ms a]$. 
\end{proof}

\begin{rmrk}
Given Lemma~\ref{lemm.eq.strong}, could Definition~\ref{defn.eq.powamgis} use $\forall c$ instead of the less familiar $\new c$? 

Yes, but the use of the $\new$-quantifier is preferable.
It is a weaker proof-obligation; 
it is strictly easier to prove $\New{c}p[u\ms c]=p[u\ms c]$ than it is to prove $\Forall{c}p[u\ms c]=p[u\ms c]$, because we have to worry about `fewer values' for $c$.

This is exploited immediately after Definition~\ref{defn.eq.powamgis} in Lemma~\ref{lemm.eq.sanity.check} when we use two $\new$-quantifiers, one for $a'$ and one for $c$.
If we had taken Lemma~\ref{lemm.eq.strong} as our definition instead, then we would have had to worry about the case $a'=c$; something similar happens in Proposition~\ref{prop.xeqx}.

Thus, the $\new$ form used in Definition~\ref{defn.eq.powamgis} is more elementary and convenient, though it is indeed logically equivalent to the $\forall$ form of Lemma~\ref{lemm.eq.strong}.
\end{rmrk}

We mention that Definitions~\ref{defn.exact.amgis.algebra} and~\ref{defn.eq.powamgis} have a kind of purely nominal precedent:
\begin{lemm}
\label{lemm.precedent}
Suppose $\ns X$ is a nominal set and $X,Y{\in}|\nompow(\ns X)|$ are finitely supported subsets of $\ns X$ (Subsection~\ref{subsect.finsupp.pow}), and suppose $a\#X$ and $a\#Y$.
Write $X_{\#a}=\{x{\in}X\mid a\#x\}$ and $Y_{\#a}=\{y{\in}Y\mid a\#y\}$.
Then 
$$
X=Y
\quad\text{if and only if}\quad
X_{\#a}=Y_{\#a}  .
$$
\end{lemm}
\begin{proof}
The left-to-right implication is immediate.
For the right-to-left implication it suffices to prove that $X\subseteq Y$.
So suppose $x\in X$.
Choose fresh $b$; then $(b\ a)\act x\in X$ and by Proposition~\ref{prop.pi.supp} $a\#(b\ a)\act x$ so $(b\ a)\act x\in Y$ so that $x\in Y$.
\end{proof}
For more on this, see \cite[Subsection~9.5]{gabbay:fountl}.

\end{document}